\documentclass[runningheads]{llncs}

\usepackage{amssymb}

\makeatletter
\RequirePackage[bookmarks,unicode,colorlinks=true,breaklinks=true]{hyperref}%
   \def\@citecolor{blue}%
   \def\@urlcolor{blue}%
   \def\@linkcolor{blue}%

\def\orcidID#1{\smash{\href{http://orcid.org/#1}{\protect\raisebox{-1.25pt}{\protect\includegraphics{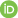}}}}}
\makeatother

\usepackage{scalerel}
\usepackage{lipsum}
\usepackage{amsfonts,amssymb}

\usepackage{graphicx}
\usepackage{epstopdf}
\usepackage{dsfont}
\usepackage{bbm}
\usepackage{indentfirst}
\usepackage{subcaption}
\usepackage{hhline}
\usepackage{mathtools}
\usepackage{commath}
\usepackage{extarrows}
\usepackage{mathrsfs} 
\usepackage{lineno}
\usepackage{multirow}
\usepackage{carshapes}

\usepackage{scalerel}
\newcommand\scale[2]{\vstretch{#1}{\hstretch{#1}{#2}}}

\usepackage{textcomp}
\usepackage{stfloats}
\usepackage{pgfplots}
\pgfplotsset{compat=1.17} 
\usepackage{stmaryrd}
\usepackage{algorithm}
\usepackage[noend]{algpseudocode}
\usepackage{subcaption}
\usepackage{bbding}

\allowdisplaybreaks

\makeatletter
\newcommand\multiline[1]{\parbox[t]{\dimexpr\linewidth-\ALG@thistlm}{#1}}
\makeatother

\usepackage{tikz}
\usetikzlibrary{arrows.meta,shapes,decorations.pathmorphing,backgrounds,positioning,fit,petri}
\usetikzlibrary{decorations,decorations.markings} 
\usetikzlibrary{positioning,arrows}
\usetikzlibrary{calc}

\usepackage{centernot}
\usepackage{hyperref}
\hypersetup{colorlinks=true, linkcolor=blue, breaklinks=true, urlcolor=blue}

\ifpdf
  \DeclareGraphicsExtensions{.eps,.pdf,.png,.jpg}
\else
  \DeclareGraphicsExtensions{.eps}
\fi

\newtheorem{defi}{\textbf{Definition}}

\newtheorem{asp}{Assumption}

\newtheorem{lema}{\textbf{Lemma}}

\newcounter{exampcount}
\newcommand{\defiref}[1]{Definition~\ref{#1}}

\newcommand{\thomref}[1]{Theorem~\ref{#1}}
\newcommand{\aspref}[1]{Assumption~\ref{#1}}

\newcommand{\algoref}[1]{Algorithm~\ref{#1}}

\newcommand{\lemaref}[1]{Lemma~\ref{#1}}
\newcommand{\tabref}[1]{Table~\ref{#1}}
\renewcommand{\figref}[1]{Fig.~\ref{#1}}
\newcommand{\sectref}[1]{Section~\ref{#1}}

\newcommand{\appxref}[1]{Appx.~\ref{#1}}

\renewcommand{\emptyset}{\varnothing}

\DeclareMathOperator*{\argmax}{argmax}

\newcommand{\cO}{\mathcal{O}}
\newcommand{\obs}{\mathit{obs}}

\newcommand{\Loc}{\mathit{Loc}}
\newcommand{\Per}{\mathit{Per}}

\newcommand{\loc}{\mathit{loc}}
\newcommand{\per}{\mathit{per}}
\newcommand{\csg}{\mathsf{C}}

\newcommand{\game}{\mathsf{G}}

\newcommand{\agent}{\mathsf{Ag}}

\newcommand{\sem}[1]{\llbracket {#1} \rrbracket}

\newcommand{\ipaths}{\mathit{IPaths}}
\newcommand{\fpaths}{\mathit{FPaths}}

\pgfkeys{/tikz/.cd,
    street mark distance/.store in=\StreetMarkDistance,
    street mark distance=10pt,
    street mark step/.store in=\StreetMarkStep,
    street mark step=1pt,
}

\pgfdeclaredecoration{street mark}{initial}
{%
\state{initial}[width=\StreetMarkStep,next state=cont] {
    \pgfmoveto{\pgfpoint{\StreetMarkStep}{\StreetMarkDistance}}
    \pgfpathlineto{\pgfpoint{0.3\pgflinewidth}{\StreetMarkDistance}}
    \pgfcoordinate{lastup}{\pgfpoint{1pt}{\StreetMarkDistance}}
    
  }
  \state{cont}[width=\StreetMarkStep]{
     \pgfmoveto{\pgfpointanchor{lastup}{center}}
     \pgfpathlineto{\pgfpoint{\StreetMarkStep}{\StreetMarkDistance}}
     \pgfcoordinate{lastup}{\pgfpoint{\StreetMarkStep}{\StreetMarkDistance}}
  }
  \state{final}[width=\StreetMarkStep]
  {
    \pgfmoveto{\pgfpointdecoratedpathlast}
  }
}

\newcommand{\custompath}{(0,0) to[out=90,in=270] (0,4)}

\newcommand{\startpara}[1]{{%
\vskip5pt\noindent
{\bf #1.}}}

\def\squareforqed{\hbox{\rlap{$\sqcap$}$\sqcup$}}
\def\qed{\ifmmode\squareforqed\else{\unskip\nobreak\hfil
\penalty50\hskip1em\null\nobreak\hfil\squareforqed
\parfillskip=0pt\finalhyphendemerits=0\endgraf}\fi}

\makeatletter
\renewcommand{\ALG@name}{\sc Algorithm}
\makeatother

\usepackage{amsopn}

\def\techreport{}

\ifthenelse{\isundefined{\techreport}}{%
\renewcommand{\appxref}[1]{\cite{arxiv}}
}{%
}

\title{Partially Observable Stochastic Games \\ with Neural Perception Mechanisms}

\author{
Rui~Yan\inst{1}$^{\text{\,\raisebox{-1pt}{\Envelope}}}$\,\orcidID{0000-0002-8685-5055}
\and Gabriel~Santos\inst{1} \orcidID{0000-0002-6570-9737}
\and Gethin~Norman\inst{1,2} \orcidID{0000-0001-9326-4344}
\and\\ David~Parker\inst{1} \orcidID{0000-0003-4137-8862}
\and Marta~Kwiatkowska\inst{1} \orcidID{0000-0001-9022-7599}
}

\institute{University of Oxford, Oxford, OX1 2JD, UK \\
\email{\{rui.yan,gabriel.santos,david.parker,marta.kwiatkowska\}@cs.ox.ac.uk}
\and 
University of Glasgow, Glasgow, G12 8QQ, UK \\
\email{gethin.norman@glasgow.ac.uk}} 

\authorrunning{R.~Yan, G.~Santos, G.~Norman, D.~Parker and M.~Kwiatkowska}
\titlerunning{Partially Observable Stochastic Games with Neural Perception Mechanisms}

\begin{document}

\maketitle
\begin{abstract}
Stochastic games are a well established model for multi-agent sequential decision making under uncertainty.
In practical applications, though, agents often have only partial observability of their environment.
Furthermore, agents increasingly perceive their environment using data-driven approaches such as neural networks trained on continuous data.
We propose the model of neuro-symbolic partially-observable stochastic games (NS-POSGs),
a variant of continuous-space concurrent stochastic games that explicitly incorporates neural perception mechanisms.
We focus on a one-sided setting with a partially-informed agent using discrete, data-driven observations and another, fully-informed agent. 
We present a new method, called one-sided NS-HSVI, for approximate solution of one-sided NS-POSGs, which exploits the piecewise constant structure of the model.
Using neural network pre-image analysis to construct finite polyhedral representations
and particle-based representations for beliefs,
we implement our approach and illustrate its practical applicability to
the analysis of pedestrian-vehicle and pursuit-evasion scenarios.
\end{abstract}

\setlength{\textfloatsep}{4pt}

\section{Introduction}

\noindent
Strategic reasoning is essential to ensure stable multi-agent coordination in complex 
environments, e.g., autonomous driving or multi-robot planning.
\emph{Partially-observable stochastic games} (POSGs) are a natural model for settings 
involving multiple agents, uncertainty and partial information.
They allow the synthesis of optimal (or near-optimal) strategies 
and equilibria that guarantee expected outcomes, even in adversarial scenarios.
But POSGs also present significant challenges:
key problems are undecidable, already for the single-agent case of partially observable Markov decision processes (POMDPs)~\cite{OM-SH-AC:03}, and practical algorithms for 
finding 
optimal values and strategies are lacking.

Computational tractability can be improved using {\em one-sided POSGs},
a subclass of two-agent, zero-sum POSGs where only one agent has partial information while the other agent is assumed to have full knowledge of the state~\cite{WZ-TJ-HL:22,WZ-TJ-HL:23}.
This can be useful when making worst-case assumptions about one agent,
such as in an adversarial setting (e.g., an attacker-defender scenario) or 
a safety-critical domain (e.g., a pedestrian in an autonomous driving application).

From a computational perspective, one-sided POSGs
avoid the need for nested beliefs~\cite{LZ-BM-LK:21},
i.e., reasoning about beliefs not only over states but also over opponents' beliefs.
This is because the fully-informed agent
can reconstruct beliefs from observation histories.
Recent advances~\cite{KH-BB-VK-CK:23} have led to the first practical variant of heuristic search value iteration (HSVI)~\cite{TS-RS:04} for computing approximately optimal values and strategies in (finite) one-sided POSGs.


However, in many realistic autonomous coordination scenarios, agents perceive {\em 
continuous} environments using {\em data-driven} observation functions, typically 
implemented as neural networks (NNs). Examples include autonomous vehicles using NNs to 
perform object recognition or to estimate pedestrian intention, and NN-enabled vision in an airborne pursuit-evasion scenario. 

In this paper, we introduce {\em one-sided neuro-symbolic POSGs (NS-POSGs)},
a variant of continuous-space POSGs that explicitly incorporates neural perception mechanisms.
We assume one partially-informed agent with a (finite-valued) observation function synthesised in a data-driven fashion, and a second agent with full observation of the (continuous) state.
Continuous-space models with neural perception mechanisms have already been developed,
but are limited to the simpler cases of POMDPs~\cite{RY-GS-GN-DP-MK:23}
and (fully-observable) stochastic games~\cite{nscsgs}.
Our model provides the ability to reason about an agent with a realistic
perception mechanism \emph{and} operating in an adversarial or worst-case setting.

Solving continuous-space models, even approximately, is computationally challenging.
One approach is to discretise and then use techniques for finite-state models
(e.g.,~\cite{KH-BB-VK-CK:23} in our case).
But this can yield exponential growth of the state space,
depending on the granularity and time-horizon used.
Furthermore, decision boundaries for data-driven perception are typically irregular
and can be misaligned with gridding schemes for discretisation, limiting precision.

An alternative is to exploit structure in the underlying model
and work directly with the continuous-state model.
For example, classic dynamic programming approaches to solving
MDPs 
can be lifted to continuous-state variants~\cite{ZF-RD-NM-RW:04}:
a piecewise constant representation of the value function is computed,
based on a partition of the state space created dynamically during solution.
It is demonstrated that this approach can outperform discretisation
and that it can also be generalised to solving POMDPs.
We can adapt this approach to models with neural perception mechanisms~\cite{RY-GS-GN-DP-MK:23},
exploiting the fact that ReLU NN classifiers induce a finite decomposition
of the continuous environment into polyhedra.

\startpara{Contributions}
The contributions of this paper are as follows.
We first define the model of one-sided NS-POSGs
and motivate it via an autonomous driving scenario
based on a ReLU NN classifier for pedestrian intention
learnt from public datasets~\cite{AR-IK-TK-JKT:19}.
We then prove that the (discounted reward) value function for NS-POSGs
is continuous and convex, and is a fixed point of a minimax operator.
Based on mild assumptions about the model,
we give a piecewise linear and convex representation of the value function,
which admits a finite polyhedral representation
and which is closed with respect to the minimax operator. 

In order to provide a feasible approach to approximating values of NS-POSGs,
we present a variant of HSVI, which is a popular anytime algorithm for POMDPs
that iteratively computes lower and upper bounds on values.
We build on ideas from HSVI for finite one-sided POSGs~\cite{KH-BB-VK-CK:23}
(but there are multiple challenges when moving to a continuous state space and NNs)
and for POMDPs with neural perception mechanisms~\cite{RY-GS-GN-DP-MK:23}
(but, for us, the move to games brings a number of complications);
see \sectref{sec:NS-HSVI} for a detailed discussion.

We implement our one-sided NS-HSVI algorithm
using the popular particle-based representation for beliefs
and employing NN pre-image computation~\cite{KM-FF:20}
to construct an initial finite polyhedral representation of perception functions.
We apply this to the pedestrian-vehicle interaction scenario
and a pursuit-evasion game inspired by mobile robotics applications,
demonstrating the ability to synthesise agent strategies
for models with complex perception functions,
and to explore trade-offs when
using perception mechanisms of varying precision.

\startpara{Related work}
Solving POSGs is largely intractable.
Methods based on exact dynamic programming~\cite{EAH-DSB-SZ:04}
and approximations~\cite{AK-SZ:09,REM-GG-JS-ST:04} exist
but have high computational cost.
Further approaches exist for \emph{zero-sum} POSGs, including
conversion to extensive-form games~\cite{BB-CK-VL-MP:14},
counterfactual regret minimisation~\cite{MZ-MJ-MB-CP:07,VK-MS-NB-MB-VL:22,VK-DS-VL-JR-SS-KH:23} and methods based on reinforcement learning and search~\cite{NB-AB-AL-QG:20,MM-MS-NB-VL-et-al:17}.
In \cite{AD-OB-JSD-AS:23}, an HSVI-like finite-horizon solver that provably converges to an $\varepsilon$-optimal solution is proposed;
\cite{AJW-FAO-DMR:16} provides convexity and concavity results but no algorithmic solution. 

Methods exist for \emph{one-sided} POSGs:
a space partition approach when actions are public~\cite{WZ-TJ-HL:22},
a point-based approximate algorithm when observations are continuous \cite{WZ-TJ-HL:23}
and projection to POMDPs based on factored representations~\cite{SC-NJ-SB-MS-UT:21}.
But these are all restricted to \emph{finite-state} games. 
Closer to our work, but still for finite models, is~\cite{KH-BB-VK-CK:23},
which proposes an HSVI method for 
POSGs.


For the \emph{continuous-state} but \emph{single-agent} (POMDP) setting,
point-based value iteration \cite{JMP-NV-MTS-PP:06,LB-IL-NRA:19,ZZ-SS-PP-KK:12} 
and discrete space approximation~\cite{SB-TG-RD:13} can be used;
the former also uses $\alpha$-functions 
but works with (approximate) Gaussian mixtures or beta-densities,
whereas we exploit structure, similarly to~\cite{ZF-RD-NM-RW:04}.
%
As discussed above, in earlier work,
we proposed models and techniques for 
extending several simpler probabilistic models
with neural perception mechanisms~\cite{RY-GS-GN-DP-MK:23,YSD+22,nscsgs}.
Recent work~\cite{YSN+24} builds on the one-sided NS-POSG model proposed in this paper,
but focuses instead on \emph{online} methods for strategy synthesis.


\vspace*{-0.3em}
\section{Background}\label{background-sect}
\vspace*{-0.2em}


\startpara{POSGs}
The semantics of our models are continuous-state
\emph{partially observable concurrent stochastic games} (POSGs)~\cite{VK-MS-NB-MB-VL:22,NB-AB-AL-QG:20,KH-BB:19}.
Letting $\mathbb{P}(X)$ denote the space of probability measures on a Borel space $X$,
POSGs are defined as follows.
%

A two-player POSG is a tuple $\game = (N, S, A, \delta, \cO,Z)$,
where: $N=\{1, 2\}$ is a set of two agents; $S$ a Borel measurable set of states;  $A \triangleq A_1 {\times} A_2$
a finite set of joint actions where $A_i$ are actions of agent $i$;
$\delta : (S {\times} A) \rightarrow \mathbb{P}(S)$ a probabilistic transition function; $\cO \triangleq \cO_1 {\times} \cO_2$ a finite set of joint observations where $\cO_i$ are observations of agent $i$; and $Z: (S {\times} A {\times} S) \to \cO$ an observation function.

In a state $s$ of a POSG $\game$, each agent $i$ 
selects an action $a_i$ from $A_i$. The probability to move to a state $s'$ is $\delta(s, (a_1, a_2))(s')$, and the subsequent observation is $Z(s, (a_1, a_2), s')=(o_1, o_2)$, where agent $i$ can only observe $o_i$.
A \emph{history} of $\mathsf{G}$ is a sequence of states and joint actions 
$\pi=(s^{0}, a^{0}, s^{1}, \dots, a^{t - 1}, s^t)$
such that
$\delta(s^k,a^k)(s^{k+1})>0$ for each $k$. 
For a history $\pi$, we denote by $\pi(k)$ the $(k{+}1)$th state,
and $\pi[k]$ the $(k{+}1)$th action.
%
A (local) \emph{action-observation history (AOH)} 
is the view of a history $\pi$ from agent $i$'s perspective:
$\pi_i = (o_{i}^0, a_{i}^0, o_{i}^1,\dots, a_{i}^{t - 1}, o_{i}^t)$.
If an agent has full information about the state, then we assume the agent is also informed of the history of joint actions.
Let $\fpaths_{\game}$ and $\fpaths_{\game,i}$ denote the sets of finite histories of $\game$ and AOHs of agent $i$, respectively. 

A  (behaviour) \emph{strategy} of agent $i$ is a 
mapping 
$\sigma_i: \fpaths_{\game,i} \to \mathbb{P}(A_i)$. 
We denote by $\Sigma_i$ the set of strategies of agent $i$. A \emph{profile}  $\sigma=(\sigma_1, \sigma_2)$ is a pair of strategies for each agent and we denote by $\Sigma=\Sigma_1\times \Sigma_2$ the set of profiles.

\startpara{Objectives}
Agents 1 and 2 maximise and minimise, respectively, the expected value of the \emph{discounted reward}  $Y(\pi) = \sum_{k=0}^{\infty} \beta^k r(\pi(k), \pi[k])$, where $\pi$ is an infinite history, $r : (S {\times} A) \to \mathbb{R}$ a reward structure and $\beta \in (0, 1)$.
The expected value of $Y$ starting from state distribution $b$
under profile $\sigma$ is denoted $\mathbb{E}_b^{\sigma} [Y]$. 

\startpara{Values and minimax strategies}
If $V^\star(b) \triangleq \sup\nolimits_{\sigma_1 \in \Sigma_1}\inf\nolimits_{\sigma_2 \in \Sigma_2} \mathbb{E}_b^{\sigma_1,\sigma_2}[Y]$ = $\inf\nolimits_{\sigma_2 \in \Sigma_2} \sup\nolimits_{\sigma_1 \in \Sigma_1}\mathbb{E}_b^{\sigma_1,\sigma_2}[Y]$ for all $b \in \mathbb{P}(S)$, then $V^\star$ is called the \emph{value} of~$\game$. A profile $\sigma^{\star} = (\sigma_1^{\star}, \sigma_2^{\star})$ is a \emph{minimax strategy profile} if, for any $b \in \mathbb{P}(S)$, $\smash{\mathbb{E}_b^{\sigma_1^{\star},\sigma_2}[Y] \ge \mathbb{E}_b^{\sigma_1^{\star},\sigma_2^{\star}}[Y] \ge \mathbb{E}_b^{\sigma_1,\sigma_2^{\star}}[Y]}$ for all $\sigma_1 \in \Sigma_1$ and $\sigma_2 \in \Sigma_2$.

\vspace*{-0.3em}
\section{One-Sided Neuro-Symbolic POSGs}\label{nscsgs-sect}
\vspace*{-0.2em}

\noindent
We now introduce our model, aimed at commonly deployed multi-agent scenarios with data-driven perception, necessitating the use of continuous environments.


\startpara{One-sided NS-POSGs}
A \emph{one-sided neuro-symbolic POSG (NS-POSG)} comprises a \emph{partially informed, neuro-symbolic} agent and a \emph{fully informed} agent in a continuous-state environment.
The first agent has a finite set of local states, 
and is endowed with a data-driven perception mechanism,
through which (and only through which) it makes finite-valued observations of the environment's state,
stored locally as \emph{percepts}.
The second agent can directly observe both the local state and percept of the first agent, and the state of the environment. 

\begin{defi}[NS-POSG]\label{defi:NS-CSG}
A one-sided NS-POSG $\csg$ comprises
agents $\agent_1  = (S_1,A_1,\obs_1,\delta_1)$ 
and $\agent_2 {=} (A_2)$,
and environment $E {=} (S_E,\delta_E)$, where:
\begin{itemize}
    \item $S_1 = \Loc_1 {\times} \Per_1$ is a set of states for $\agent_1$,
    where  $\Loc_1$ and $\Per_1$ are finite sets of local states and percepts, respectively;
    
    \item $S_E\subseteq \mathbb{R}^e$ is a closed set of continuous environment states;  

    \item $A_i$ is a finite set of actions for $\agent_i$ and
    $A \triangleq A_1 {\times} A_2$ is a set of joint actions;
  
    \item $\obs_1 : (\Loc_1 {\times} S_E)\to \Per_1$ is $\agent_1\!$'s perception function; 
    
    \item $\delta_1: (S_1 {\times} A) \to \mathbb{P}(\Loc_1)$ is
   $\agent_1\!$'s local probabilistic transition function;
    
    \item $\delta_E: (\Loc_1 {\times} S_E {\times} A) \to \mathbb{P}(S_E)$ is a finitely-branching probabilistic transition function for the environment. 
\end{itemize}
\end{defi}

One-sided NS-POSGs are a subclass of two-agent, hybrid-state POSGs
with discrete observations ($S_1$)
and actions for $\agent_1$, and continuous observations ($S_1 {\times} S_E$) 
and discrete actions for $\agent_2$. 
Additionally, $\agent_1$ is informed of its own actions
and $\agent_2$ of joint actions.
Thus, $\agent_1$ is partially informed, without access to environment states and actions of $\agent_2$, and $\agent_2$ is fully informed.
Since $\agent_2$ needs no percepts, its local state and transition function are omitted.

The game executes as follows. 
A global state of $\csg$ comprises a state $s_1 = (\loc_1, \per_1)$ for 
$\agent_1$ and an environment state $s_E$. In state $s=(s_1,s_E)$, 
the two agents concurrently choose one of their actions,
resulting in a joint action $a=(a_1,a_2)\in A$.
Next, the local state of $\agent_1$ is updated to some $\loc_1'\in \Loc_1$,
according to $\delta_1(s_1,a)$. At the same time, the environment state is updated
to some $s_E'\in S_E$ according to $\delta_E(\loc_1, s_E,a)$.
Finally, the first agent $\agent_1$, based on $\loc_1'$,
generates a percept $\per_1' = \obs_1(\loc'_1,s'_E)$ by observing the environment state $s_E'$ and $\csg$ reaches the global state $s'=((\loc_1', \per_1'), s_E')$. 

We focus on neural perception functions, i.e., for each local state $\loc_1$, we associate an NN classifier $f_{\loc_1} : S_E \to \mathbb{P}(\Per_1)$ that returns a distribution over percepts for each environment state $s_E \in S_E$.
Then $\obs_1(\loc_1, s_E) = f^{\max}_{\loc_1}(s_E)$, where $f^{\max}_{\loc_1}(s_E)$ is the percept with the largest probability in $f_{\loc_1}(s_E)$
(a tie-breaking rule is applied if multiple percepts have the largest probability).

\startpara{Motivating example: Pedestrian-vehicle interaction}
A key challenge for 
autonomous driving in urban environments
is predicting pedestrians' intentions or actions.
One solution is NN classifiers, 
e.g., trained on video datasets~\cite{AR-IK-JKT:17,AR-IK-TK-JKT:19}.
To illustrate our NS-POSG model,
we consider decision making for an autonomous vehicle
using an NN-based intention estimation model for a pedestrian at a crossing~\cite{AR-IK-TK-JKT:19}.
We use their simpler ``vanilla'' model,
which takes two successive (relative) locations of the pedestrian (the top-left coordinates $(x_1,y_1)$ and $(x_2,y_2)$ of two fixed size bounding boxes around the pedestrian) and classifies its intention as:
\emph{unlikely}, \emph{likely} or \emph{very likely} to cross.
We train a feed-forward NN classifier with ReLU activation functions over the PIE dataset~\cite{AR-IK-TK-JKT:19}.

\tikzset{
    pics/man/.style={code=
    {
    \draw[#1]   
        (0,0) .. controls ++(0,-0.8) and ++(0.2,0.6) ..
        (-0.4,-1.8) .. controls ++(0.2,-0.8) and ++(0.1,0.6) ..
        (-0.5,-4.4) .. controls ++(-0.6,-0.2) and ++(0.7,0.1) ..
        (-2,-4.8) .. controls ++(0,0.3) and ++(-0.5,-0.2)  ..
        (-1,-3.8) .. controls ++(-0.1,0.9) and ++(-0.1,-0.8)  ..
        (-1,-1.8) .. controls ++(-0.3,1) and ++(-0.2,-0.8)  ..              
        (-0.9,0.9) .. controls ++(-0.1,1) and ++(0,-0.8)  ..
        (-1.2,2.8) .. controls ++(-0.4,-1) and ++(0.4,0.5)  ..
        (-2.6,0.8) .. controls ++(0.5,-0.8) and ++(0.2,-0.1)  ..
        (-3.2,-0.1) .. controls ++(-0.2,0) and ++(-0.3,-0.5)  ..
        (-3.3,0.8) .. controls ++(0.4,0.5) and ++(-0.5,-0.5)  ..        
        (-1.8,3.4) .. controls ++(0.5,0.5) and ++(-0.3,-0.1)  ..                
        (-0.7,3.9) .. controls ++(0.3,0.1) and ++(0,-0.2)  ..
        (-0.4,4.3) .. controls ++(-1.2,0.3) and ++(-1.2,0)  ..
        (0,6.2) coordinate (-head) .. controls ++(1.2,0) and ++(1.2,0.3) .. 
        (0.4,4.3) .. controls ++(0,-0.2) and ++(-0.3,0.1) ..
        (0.7,3.9) .. controls ++(0.3,-0.1) and ++(-0.5,0.5) ..
        (1.8,3.4) .. controls ++(0.5,-0.5) and ++(-0.4,0.5) ..
        (3.3,0.8) .. controls ++(0.3,-0.5) and ++(0.2,0) ..
        (3.2,-0.1) .. controls ++(-0.2,-0.1) and ++(-0.5,-0.8) ..
        (2.6,0.8) .. controls ++(-0.4,0.5) and ++(0.4,-1) ..
        (1.2,2.8) .. controls ++(0,-0.8) and ++(0.1,1) ..
        (0.9,0.9) .. controls ++(0.2,-0.8) and ++(0.3,1) ..
        (1,-1.8) .. controls ++(0.1,-0.8) and ++(0.1,0.9) ..
        (1,-3.8) .. controls ++(0.5,-0.2) and ++(0,0.3) ..
        (2,-4.8) .. controls ++(-0.7,0.1) and ++(0.6,-0.2) ..
        (0.5,-4.4) .. controls ++(-0.1,0.6) and ++(-0.2,-0.8) ..
        (0.4,-1.8) .. controls ++(-0.2,0.6) and ++(0,-0.8) ..
        (0,0) ++ (0,2) coordinate (-heart) -- cycle
        ;
    },
    }
}

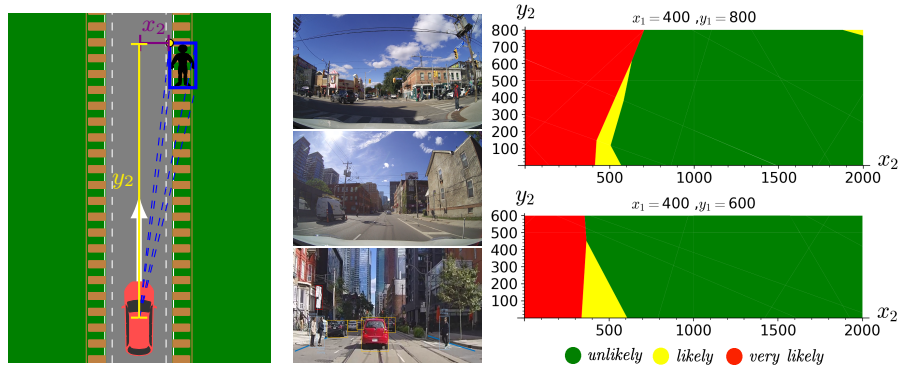
\begin{figure}[t]
\centering
\hspace{-0.25cm}
\begin{subfigure}{0.4\textwidth}
\begin{tikzpicture}[scale=2.9]
    \clip (-0.6,0.15) rectangle (0.6,1.75);
    \path[fill=green!50!black] (-1,0) rectangle (1,2.5);
    \draw[line width=25,gray] \custompath;
    \draw[draw=white,dashed,double=white,double distance=20] \custompath;
    \draw[line width=20,gray] \custompath;
    \draw[white,decorate,decoration={street mark},street mark distance=13] \custompath;
    \draw[white,decorate,decoration={street mark},street mark distance=-13] \custompath;
    
    \draw[draw=brown,dashed,double=brown,double distance=6, xshift=5.5] \custompath;
    \draw[brown,decorate,decoration={street mark},street mark distance=-20] \custompath;
    \draw[draw=brown,dashed,double=brown,double distance=6, xshift=-5.5] \custompath;
    \draw[brown,decorate,decoration={street mark},street mark distance=20] \custompath;

    \draw[decorate,decoration={markings,
     mark=at position 0.1 with {\draw[white,-latex,line width=2pt] (0.25,0)
     coordinate (0.25,0) -- (1.5,0);}
    }] \custompath;
    \node[sedan top,body color=red!70,window color=black!80,minimum width=0.9cm,rotate = 90,scale=1.2] (car) at
    (0,0.35) {};

    \def\xp{0.2}
    \def\yp{1.5}
    
    \draw (\xp,\yp) pic(M){man={scale=0.05,black!50!black,fill=black}};
    \draw[blue, very thick] (\xp-0.06,\yp-0.09) rectangle (\xp+0.06, \yp+0.115);
    \draw[blue, dashed] (car) -- (\xp+0.06,\yp+0.115);
    \draw[blue, dashed] (car) -- (\xp+0.06,\yp-0.09);
    \draw[blue, dashed] (car) -- (\xp-0.06,\yp-0.09);
    \draw[blue, dashed] (car) -- (\xp-0.06,\yp+0.115);
    \node[draw, circle, fill=yellow, inner sep=0pt, minimum size=0.1cm] at (\xp-0.06,\yp+0.115) (point) {}; 

    \draw[|-|, violet, thick] (\xp-0.06,\yp+0.115) -- (0.0,\yp+0.115) node [midway, above, xshift=-0.0] {\small $x_2$};
    \draw[|-|, yellow, thick] (car) -- (0.0,\yp+0.115) node [midway, left, xshift=2.0] {\small $y_2$};

\end{tikzpicture}
\end{subfigure}
\hspace{-1.3cm}
\begin{subfigure}{0.2\textwidth}
    \centering
    \includegraphics[width=2.5cm, height=1.52cm]{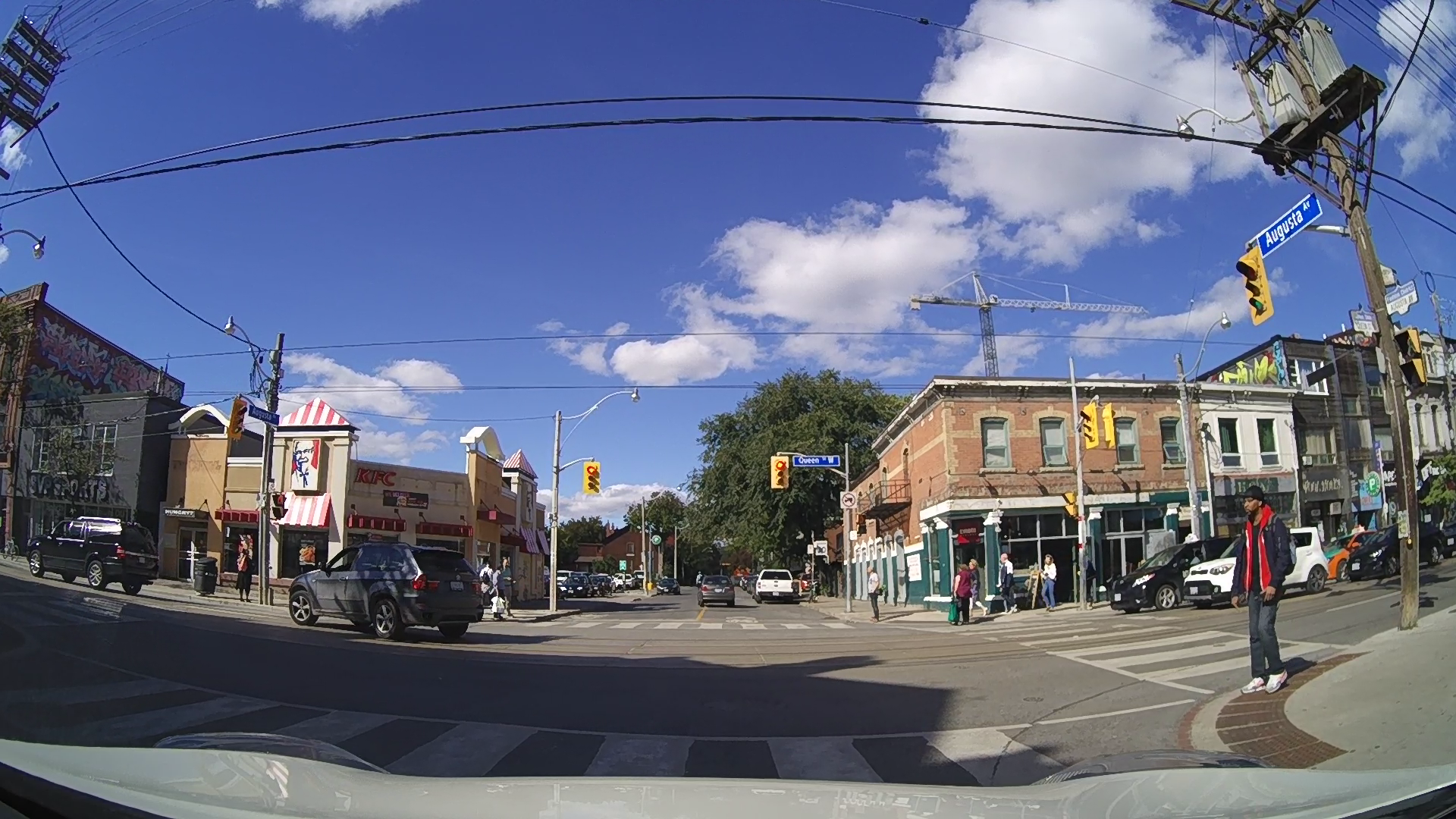}
    \includegraphics[width=2.5cm, height=1.52cm]{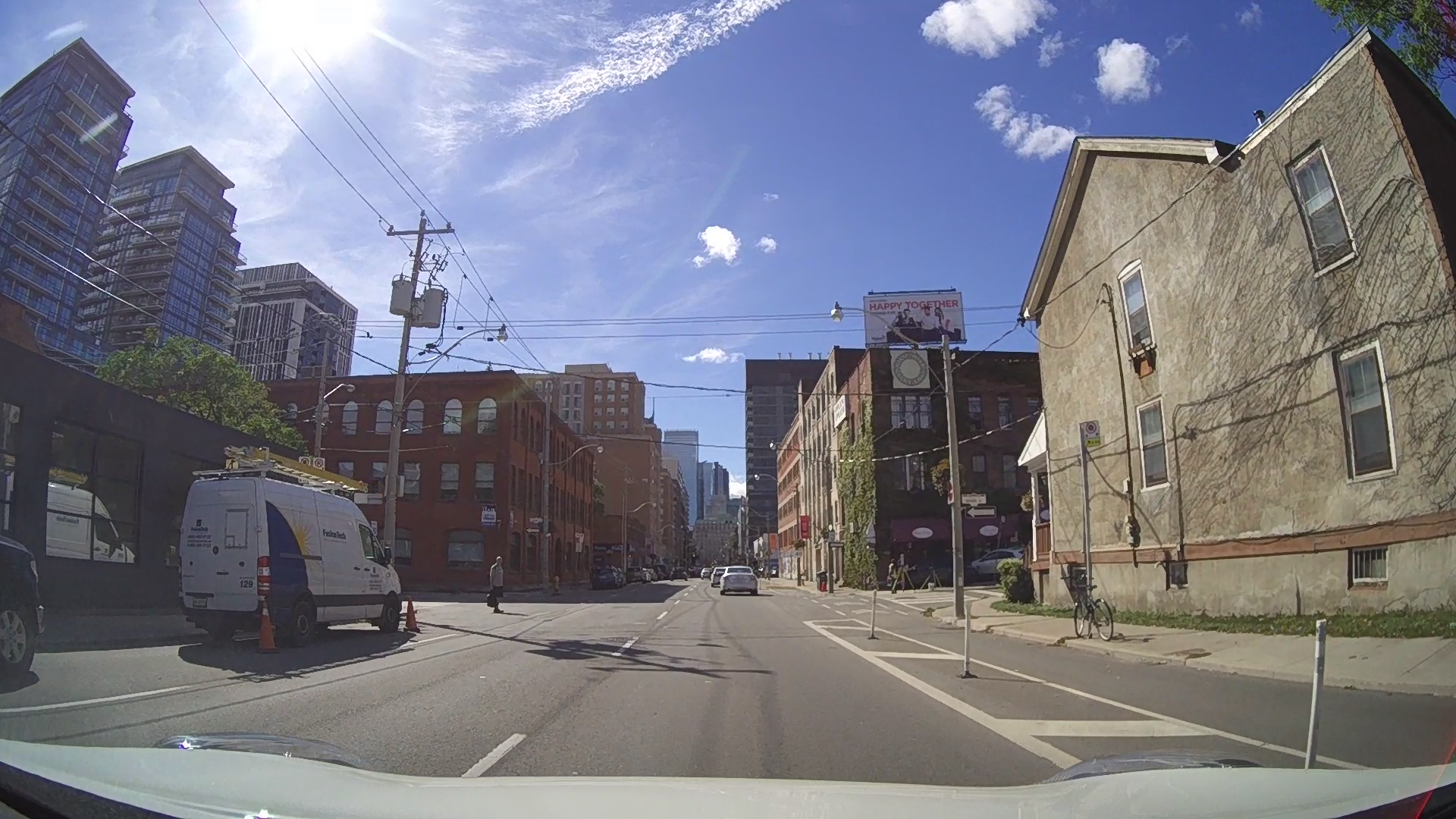}
    \includegraphics[width=2.5cm, height=1.52cm]{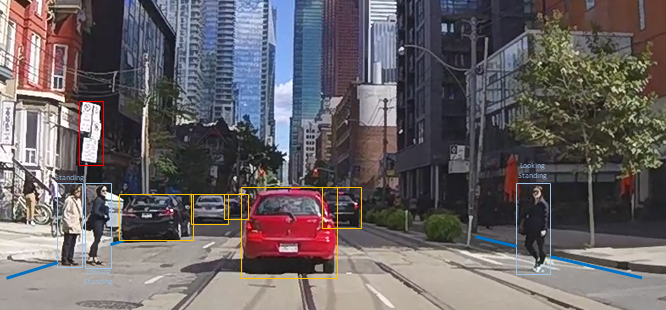}
\end{subfigure}
\hspace{-0.2cm}
\raisebox{-0.05cm}{
\begin{subfigure}{0.4\textwidth}
    \centering
    \includegraphics[scale=0.34]{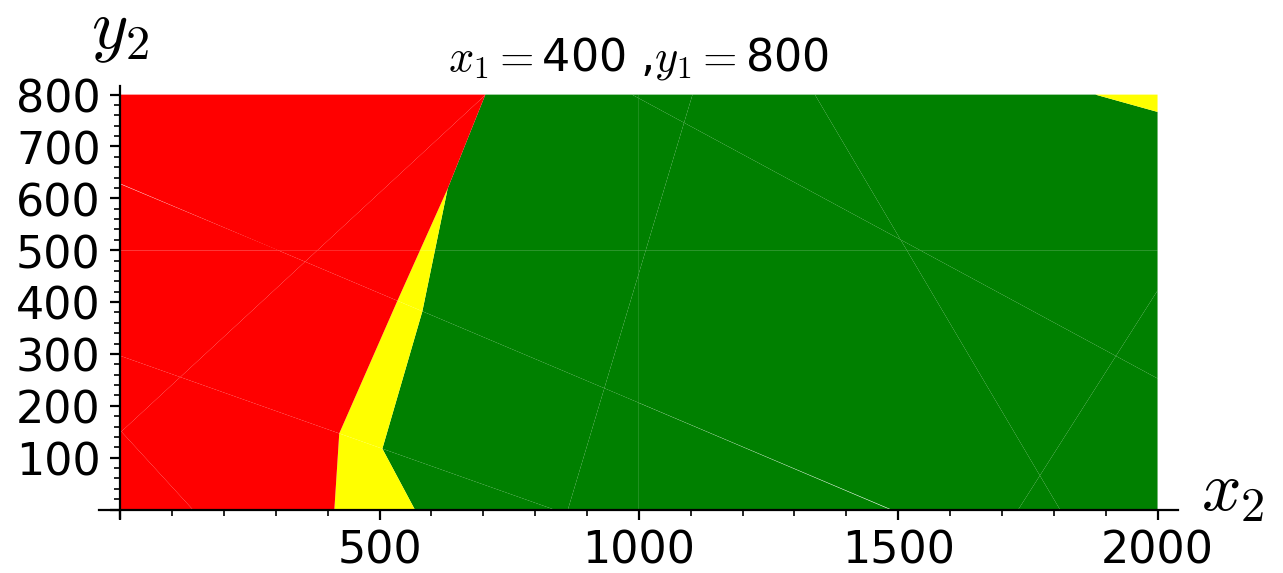}
    \includegraphics[scale=0.34]{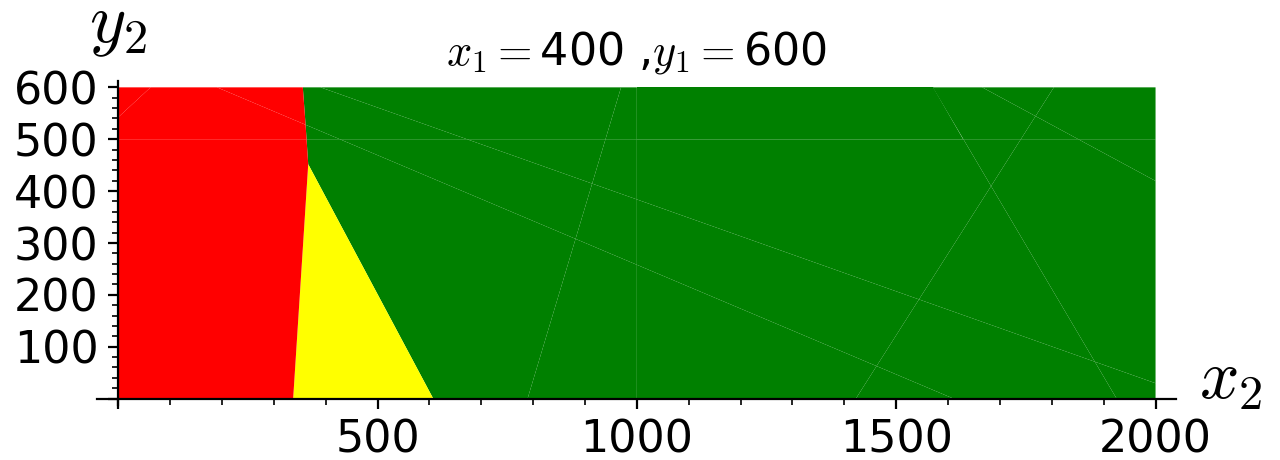}
    \hspace*{0.5cm}
    \includegraphics[scale=0.5]{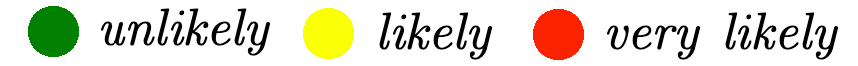}
\end{subfigure}
}
\caption{Pedestrian-vehicle example. Left: Positions of two agents. Middle: Sample images from the PIE dataset~\cite{AR-IK-TK-JKT:19}. Right: Slices of learnt perception function, where $(x_1,y_1),(x_2,y_2)$ are two successive (relative) positions of the pedestrian.}
\label{fig:pedestrian_vehicle_collated}
\end{figure}

We build this perception mechanism into an NS-POSG model
of a vehicle yielding at a pedestrian crossing, based on~\cite{TF-LM-NS:18}, illustrated in \figref{fig:pedestrian_vehicle_collated}.
A pedestrian further ahead at the side of the road may decide to cross
and the vehicle must decide how to adapt its speed.
The first, partially-informed agent represents the vehicle.
It observes the environment (comprising the successive pedestrian locations) using the NN-based perception mechanism to predict the pedestrian's intention.
This is stored as a percept and its speed as its local state.
The vehicle chooses between selected (positive or negative) acceleration actions.
The second agent, the pedestrian, is fully informed,
providing a worst-case analysis of the vehicle decisions,
and can decide to cross or return to the roadside.
The goal of the vehicle is to minimise the likelihood of a collision with the pedestrian, which is achieved by associating a negative reward with this event.

\figref{fig:pedestrian_vehicle_collated} also shows selected slices of
the state space decomposition 
obtained by computing the pre-image~\cite{KM-FF:20} of the learnt NN classifier,
for each of the three predicted intentions.
The decision boundaries are non-trivial,
justifying our goal of performing a formal analysis,
but some intuitive characteristics can be seen.
When $x_2 \geq x_1$, meaning that the pedestrian is stationary or moving away from the road,
it will generally be classified as \emph{unlikely} to cross.
We also see the prediction model is \emph{cautious} when trying to make an estimation if its first observation is made from greater distance.
More details are in \appxref{sec:appendix-examples}.

\startpara{One-sided NS-POSG semantics}
A one-sided NS-POSG $\csg$ induces a POSG $\sem{\csg}$, where we
restrict to states that are \emph{percept compatible}, i.e., where $\per_1 = \obs_1(\loc_1, s_E)$ for $s = ((\loc_1, \per_1), s_E)$.
The semantics of a one-sided NS-POSG 
is closed with respect to percept compatible states.

\begin{defi}[Semantics]\label{semantics-def}
Given a one-sided NS-POSG $\csg$, as in \defiref{defi:NS-CSG},
its semantics is the POSG $\sem{\csg} = (N,S,A, \delta, \cO, Z)$
where:
\begin{itemize}
    \item $N=\{1, 2\}$ is a set of two agents and $A = A_1 \times A_2$;
    \item $S \subseteq S_1 \times S_E$ is the set of percept compatible states;
    \item 
    for $s=(s_1,s_E),s'=(s_1',s_E')\in S$ and $a \in A$ where $s_1=(\loc_1, \per_1)$ and $s_1'=(\loc_1', \per_1')$, we have 
    $\delta(s,a)(s')= \delta_1(s_1, a)(\loc'_1)  \delta_E(\loc_1, s_E,a)(s_E')$;
    \item $\cO = \cO_1 \times \cO_2$, where $\cO_1 = S_1$ and $\cO_2 = S$;
    \item $Z(s, a, s') = (s_1', s')$ for $s \in S$, $a \in A$ and $s'=(s_1', s_E') \in S$.
\end{itemize}
\end{defi}


\vspace*{-0.8em}
\startpara{Strategies} As $\sem{\csg}$ is a POSG,  we consider (behaviour) \emph{strategies} for the two agents.
Since $\agent_2$ is fully informed, it can recover the beliefs of $\agent_1$, thus removing nested beliefs.
Hence, the AOHs of $\agent_2$ are equal to the histories of $\sem{\csg}$, i.e., $\fpaths_{\sem{\csg},2} = \fpaths_{\sem{\csg}} $.
We also consider the \emph{stage strategies} at 
a history of $\sem{\csg}$,
which will later be required for solving the induced zero-sum normal-form games in the minimax operator. For a history $\pi$ of $\sem{\csg}$, 
a stage strategy for $\agent_1$ is a distribution $u_1 \in \mathbb{P}(A_1)$ and  a stage strategy for $\agent_2$ is a function $u_2: S \to \mathbb{P}(A_2)$, i.e., $u_2 \in \mathbb{P}(A_2 \mid S)$.

\startpara{Beliefs} Since $\agent_1$ is partially informed, it may need to infer the current state from its AOH. For an $\agent_1$ state  $s_1 = (\loc_1, \per_1)$, we let $S_{E}^{s_1}$ be the set of environment states compatible with $s_1$, i.e., $ S_{E}^{s_1} = \{ s_E \in S_E \mid \obs_1(\loc_1, s_E) = \per_1  \}$. Since the states of $\agent_1$ are also the observations of $\agent_1$ and states of $\sem{\csg}$ are percept compatible, a \emph{belief} for $\agent_1$, which can also be reconstructed by $\agent_2$, can be represented as a pair $ b = (s_1, b_1)$, where $s_1 \in S_1$, $b_1 \in \mathbb{P}(S_E)$ and $b_1(s_E) = 0$ for all $s_E \in S_E \setminus S_E^{s_1}$.
We denote by $S_B$ the set of beliefs of $\agent_1$.

\label{nsposgbelief}
Given a belief $(s_1, b_1)$, if action $a_1$ is selected by $\agent_1$, $\agent_2$ is \emph{assumed} to take stage strategy $u_2 \in \mathbb{P}(A_2 \mid S)$ and $s_1'$ is observed, then the updated belief of $\agent_1$  via Bayesian inference is denoted $\smash{(s_1', b_1^{s_1,a_1,u_2,s_1'})}$; see \appxref{sec:appendix-probabilities} for details.

\vspace*{-0.4em}
\section{Values of One-Sided NS-POSGs}\label{values-sect}
\vspace*{-0.3em}

\noindent
We establish the \emph{value function} of a one-sided NS-POSG $\csg$ with semantics $\sem{\csg}$,
which gives the minimax expected reward from an initial belief,
and show its convexity and continuity.
Next, to compute it, we introduce minimax and maxsup operators
specialised for one-sided NS-POSGs, and prove their equivalence.
Finally, we provide a fixed-point characterisation of the value function.

\startpara{Value function}
%
We assume a fixed reward structure $r$ and discount factor $\beta$.
The \emph{value function} of $\csg$ 
represents the minimax expected reward 
in each possible initial belief of the game, given by $V^{\star} : S_B \to \mathbb{R}$, where $V^{\star}(s_1, b_1) = \mathbb{E}_{(s_1, b_1)}^{\sigma^{\star}}[Y]$ for all $(s_1, b_1) \in S_B$ and $\sigma^{\star}$ is a minimax strategy profile of $\sem{\csg}$.

The value function for zero-sum POSGs may not exist when the state space is uncountable~\cite{MKG-DM-SS:04,AB-SS:22,SS:14} as in our case. In this paper, we only consider one-sided NS-POSGs that are determined, i.e., for which the value function exists.

\startpara{Convexity and continuity}
Since $r$ is bounded, the value function $V^{\star}$ has lower and upper bounds $L = \min\nolimits_{s \in S, a \in A} r(s,a) /(1- \beta)$ and $
U = \max\nolimits_{s \in  S, a \in A} r(s,a)/(1- \beta)$.
The proof of the following and all other results can be found in \appxref{sec:appendix-proofs}.

\begin{theorem}[Convexity and continuity]\label{thom:convexity-continuity}
    For $s_1 \in S_1$, $V^{\star}(s_1, \cdot) : \mathbb{P}(S_E) \to \mathbb{R}$ is convex and continuous, and for $b_1, b_1' \in  \mathbb{P}(S_E):$ $|V^{\star}(s_1, b_1) - V^{\star}(s_1, b_1')| \leq K(b_1, b_1')$
    where $ K(b_1, b_1') = \frac{1}{2} (U - L) \mbox{$\int_{s_E \in S_E^{s_1}}$} \big| b_1(s_E) - b_1'(s_E) \big | \textup{d}s_E $.
\vspace{-0.2cm}
\end{theorem}

\startpara{Minimax and maxsup operators} 
We give a fixed-point characterisation of the value function  $V^{\star}$, first introducing a minimax operator
and then simplifying to an equivalent maxsup variant.
The latter will be used in \sectref{sec:closure-exact-vi} to prove closure of our representation for value functions and in \sectref{sec:NS-HSVI} to formulate HSVI.
For $f: S \to \mathbb{R}$ and belief $(s_1, b_1)$, let $\langle f, (s_1, b_1) \rangle = \, \mbox{$ \int_{s_E \in S_E}$} f(s_1, s_E) b_1(s_E) \textup{d}s_E$
and $\mathbb{F}(S_B)$ denote the space of functions mapping the beliefs $S_B$ to reals $\mathbb{R}$.

\begin{defi}[Minimax]\label{defi:minimax-operator}
The minimax operator $T : \mathbb{F}(S_B) {\rightarrow} \mathbb{F}(S_B)$ is given~by:
\begin{align}
  [TV](s_1, b_1) &  = \max\nolimits_{u_1\in \mathbb{P}(A_1)} \min\nolimits_{u_2\in \mathbb{P}(A_2 \mid S)}  \mathbb{E}_{(s_1,b_1),u_1,u_2} [r(s,a)]  \nonumber \\
  & + \beta \mbox{$\sum_{(a_1,s_1') \in A_1 \times S_1}$} P(a_1, s_1' \mid (s_1, b_1), u_1, u_2 ) V(s'_1, b_1^{s_1,a_1, u_2, s_1'}) \label{eq:minimax-operator}
\end{align}
for $V \in \mathbb{F}(S_B)$ and $(s_1, b_1) \in S_B$, where $\mathbb{E}_{(s_1,b_1),u_1,u_2}[r(s,a)] = \int_{s_E \in S_E} b_1(s_E)$ $\sum_{(a_1, a_2) \in A} u_1(a_1) u_2(a_2 \mid s_1, s_E )  r((s_1, s_E),(a_1,a_2)) \textup{d} s_E  $.
\end{defi}
Motivated by \cite{KH-BB-VK-CK:23}, which proposed an 
equivalent operator for the discrete case, we instead prove that the minimax operator has an equivalent simplified form over convex continuous functions of  $\mathbb{F}(S_B)$.

For $\Gamma \subseteq \mathbb{F}(S)$,
we let $\Gamma^{A_1 \times S_1}$ denote the set of vectors of elements of
the convex hull of $\Gamma$ indexed by $A_1 {\times} S_1$.
Furthermore, for $u_1 \in \mathbb{P}(A_1)$, $\overline{\alpha} = (\alpha^{a_1,s_1'} )_{(a_1,s_1') \in A_1 \times S_1} \in \Gamma^{A_1 \times S_1}$ and $a_2 \in A_2$,  we define $f_{u_1,\overline{\alpha}, a_2} : S \to \mathbb{R}$ to be the function such that, for $s \in S$:
\begin{align}
   \lefteqn{f_{u_1,\overline{\alpha}, a_2} (s) = \mbox{$ \sum_{a_1 \in A_1}$} u_1(a_1) r(s,(a_1,a_2)} \nonumber \\
& \qquad \quad + \beta \mbox{$ \sum_{(a_1,s_1') \in A_1 \times S_1}$} u_1(a_1)  \mbox{$ \sum_{s_E' \in S_E}$} \delta(s, (a_1, a_2)) (s_1', s_E') \alpha^{a_1,s_1'}(s_1', s_E')  \label{eq:f-u1-alpha-a2}
\end{align}
where the sum over $s_E'$ is due to the finite branching of $\delta(s, (a_1, a_2))$.

\begin{defi}[Maxsup]\label{defi:simplified-operator}
For $\emptyset \neq \Gamma \subseteq \mathbb{F}(S)$, if $V(s_1, b_1) = \sup_{\alpha \in \Gamma} \langle \alpha, (s_1,b _1) \rangle $ for $(s_1, b_1) \in S_B$, then the maxsup operator $T_\Gamma : \mathbb{F}(S_B) \rightarrow \mathbb{F}(S_B)$ is defined as $[T_\Gamma V](s_1, b_1) = \mbox{$ \max_{u_1\in \mathbb{P}(A_1)} $} \mbox{$ \sup_{\overline{\alpha} \in \Gamma^{A_1 \times S_1} } $} \langle f_{u_1,\overline{\alpha}}, (s_1, b_1) \rangle$
for $(s_1, b_1) \in S_B$ where $f_{u_1,\overline{\alpha}} (s) = \min_{a_2 \in A_2} f_{u_1,\overline{\alpha},a_2} (s)$ for $s \in S$. 
\end{defi}
In the maxsup operator, $u_1$ and $\overline{\alpha}$ are aligned with $\agent_1 \!$'s goal of maximising the objective, where $u_1$ is over action distributions and $\overline{\alpha}$ is over convex combinations of elements of $\Gamma$.
The minimisation by $\agent_2$ is simplified to an optimisation over its finite action set in the function $f_{u_1, \overline{\alpha}}$. Note that each state may require a different minimiser $a_2$, as $\agent_2$ knows the current state before taking an action.

The maxsup operator avoids the minimisation over Markov kernels with continuous states in the original minimax operator.
Given $u_1$ and $\overline{\alpha}$, the minimisation can induce a pure best-response stage strategy $u_2 \in \mathbb{P}(A_2 \mid S)$ such that, for any $s \in S$, $u_2(a_2' \mid s) = 1$ for some $a_2' \in \arg\min_{a_2 \in A_2} f_{u_1,\overline{\alpha},a_2} (s)$.
Using \thomref{thom:convexity-continuity}, the operator equivalence 
and fixed-point result are as follows.
\begin{theorem}[Operator equivalence and fixed point]\label{thom:operator-equivalence}
For $\emptyset \neq \Gamma \subseteq \mathbb{F}(S)$, if
$V(s_1, b_1) = \sup_{\alpha \in \Gamma} \langle \alpha, (s_1,b _1) \rangle $ for $(s_1, b_1) \in S_B$, then the minimax operator $T$ and maxsup operator $T_\Gamma$ are equivalent 
and their unique fixed point is $V^{\star}$. 
\end{theorem}

\vspace*{-0.8em}
\section{P-PWLC Value Iteration}\label{sec:closure-exact-vi}
\vspace*{-0.2em}

\noindent
We next discuss a representation for value functions
using \emph{piecewise constant} (PWC) $\alpha$-functions,
called P-PWLC (\emph{piecewise linear and convex under PWC}),
originally introduced in~\cite{RY-GS-GN-DP-MK:23}.
This representation extends the $\alpha$-functions of \cite{JMP-NV-MTS-PP:06,LB-IL-NRA:19,ZZ-SS-PP-KK:12} for continuous-state POMDPs,
but a key difference is that we work with polyhedral representations (induced precisely from NNs)
rather than approximations based on Gaussian mixtures~\cite{JMP-NV-MTS-PP:06} or beta densities~\cite{CG-MH-BK:04}.

We show that, given PWC representations for an NS-POSG's perception, reward and transition functions, and under mild assumptions on model structure,
P-PWLC value functions are closed with respect to the minimax operator.
This yields a (non-scalable) \emph{value iteration} algorithm and,
subsequently, the basis for a more practical point-based HSVI algorithm
in \sectref{sec:NS-HSVI}.

\startpara{PWC representations}
A \emph{finite connected partition} (FCP) of $S$, denoted $\Phi$, is a finite collection of disjoint connected \emph{regions} (subsets) of $S$ that cover it. 
\begin{defi}[PWC function]\label{defi:PWC-func}
A function $f: S \to \mathbb{R}$ is piecewise constant (PWC) if there exists an FCP $\Phi$ of $S$ such that $f : \phi \to \mathbb{R}$ is constant for $\phi \in \Phi$. Let $\mathbb{F}_{C}(S)$ be the set of PWC functions in $\mathbb{F}(S)$.
\end{defi}
Since we focus on NNs for $\agent_1\!$'s perception function $\obs_1$, it is  PWC (as for the one-agent case~\cite{RY-GS-GN-DP-MK:23}) and the state space $S$ of a one-sided NS-POSG can be decomposed into a finite set of \emph{regions}, each with the same observation. 
Formally, there exists a \emph{perception FCP} $\Phi_{P}$,
the smallest FCP of $S$ such that all states in any $\phi \in \Phi_{P}$ are observationally equivalent, i.e., if $(s_1,s_E),(s_1',s_E')\in \phi$, then $s_1=s_1'$.
We can use $\Phi_P$ to find the set $\smash{S_E^{s_1}}$ for any agent state $s_1 \in S_1$. 
Given an NN representation of $\obs_1$,
the corresponding FCP $\Phi_P$ can be extracted (or approximated) offline by analysing its pre-image \cite{KM-FF:20}.

We also need to make some assumptions
about the transitions and rewards of one-sided NS-POSGs  
(in a similar style to \cite{RY-GS-GN-DP-MK:23}).
Informally, we require that, for any decomposition $\Phi'$ of the state-space into regions (i.e., an FCP), there is a second decomposition $\Phi$, the \emph{pre-image FCP}, such that states in regions of  $\Phi$ have the same rewards and transition probabilities into regions of  $\Phi'$.
The transitions of the (continuous) environment must also be decomposable into regions.


\begin{asp}[Transitions and rewards]\label{asp:transitions-rewards} Given any FCP $\Phi'$ of $S$, there exists an FCP $\Phi$ of $S$, called the \emph{pre-image FCP} of $\Phi'$, where  for $\phi \in \Phi$, $a \in A$ and $\phi' \in \Phi'$ there exists $\delta_{\Phi} : (\Phi {\times} A) \to \mathbb{P}(\Phi')$ and $r_{\Phi} : (\Phi {\times} A) \to \mathbb{R}$ such that $\delta(s, a)(s') = \delta_{\Phi}(\phi, a)(\phi')$ and $r(s, a) = r_{\Phi}(\phi, a)$ for $s \in \phi$ and $s' \in \phi'$. In addition, $\delta_E$ can be expressed in the form $\sum_{i=1}^{n} \mu_i \delta_E^i$,
where $n \in \mathbb{N}$, $\mu_i\in[0,1]$, $\sum_{i = 1}^{n} \mu_i = 1$
and $\delta_E^i : (\Loc_1 {\times} S_E {\times} A) \to S_E$ are piecewise continuous functions.
\end{asp}

\noindent
The need for this assumption also becomes clear in our later algorithms,
which compute a representation for an NS-POSG's value function over a (polyhedral) partition of the state space. This partition is created dynamically over the iterations of the solution, using a pre-image based splitting operation.

We now show, using results
for continuous-state POMDPs~\cite{RY-GS-GN-DP-MK:23,JMP-NV-MTS-PP:06}, that $V^{\star}$ is the limit of a sequence of $\alpha$-functions, called \emph{piecewise linear and convex under PWC $\alpha$-functions},
%
first introduced in~\cite{RY-GS-GN-DP-MK:23} for neuro-symbolic POMDPs.

\begin{defi}[P-PWLC function]\label{defi:PWLC}
A function $V: S_B \to \mathbb{R}$ is \emph{piecewise linear and convex under PWC $\alpha$-functions (P-PWLC)} if there exists a finite set $\Gamma \subseteq \mathbb{F}_C(S)$ such that $V(s_1, b_1) = \max_{\alpha \in \Gamma} \langle \alpha, (s_1, b_1) \rangle$ for $(s_1, b_1) \in S_B$,
where the functions in $\Gamma$ are called PWC \emph{$\alpha$-functions}.
\end{defi}
If $V \in \mathbb{F}(S_B)$ is P-PWLC, then it can be represented by a set of PWC 
functions over $S$, i.e., as a finite set of FCP regions and a value vector. Recall that $\langle \alpha, (s_1, b_1) \rangle = \int_{s_E \in S_E}  \alpha(s_1, s_E) b_1(s_E) \textup{d}s_E$, and therefore computing the value for a belief involves integration.  For one-sided NS-POSGs, we demonstrate, under Assumption~\ref{asp:transitions-rewards}, closure of the P-PWLC representation for value functions under the minimax
operator and the convergence of value iteration.
\startpara{LP, closure property and convergence} 
By showing that $f_{u_1, \overline{\alpha}, a_2}$ in \eqref{eq:f-u1-alpha-a2} is PWC in $S$
(see \ifthenelse{\isundefined{\techreport}}{\cite{arxiv}}{\lemaref{lema:PWC-intermidiate-func} in Appx.\ref{sec:appendix-proofs}}),
we use \thomref{thom:operator-equivalence} to demonstrate that, if $V$ is P-PWLC, the minimax operation can be computed by solving an LP.

\begin{lema}[LP for minimax and P-PWLC]\label{lema:LP-minimax-P-PWLC}
    If $V \in \mathbb{F}(S_B)$ is P-PWLC,
    then  $[TV](s_1, b_1)$ is given by an LP for $(s_1, b_1) \in S_B$.
\end{lema}
Using \lemaref{lema:LP-minimax-P-PWLC}, we show that the P-PWLC representation is closed under the minimax operator. This closure property enables  iterative computation of a sequence of such functions to approximate $V^{\star}$ to within a convergence guarantee.
\vspace*{-0.3cm}
\begin{theorem}[P-PWLC closure and convergence]\label{thom:P-PWLC-closure}
If $V \in \mathbb{F}(S_B)$ is P-PWLC, then so is $[TV]$. If $V^0 \in \mathbb{F}(S_B)$ is P-PWLC, then the sequence $(V^t)_{t = 0}^{\infty}$, such that $V^{t+1} = [TV^t]$, is P-PWLC and converges to $V^\star$. 
\end{theorem}
An implementation of value iteration for one-sided NS-POSGs is therefore feasible, 
since each $\alpha$-function involved is PWC and thus allows for a finite representation. However, as the number of $\alpha$-functions grows exponentially in the number of iterations,
it is not scalable in practice.

\vspace*{-0.5em}
\section{Heuristic Search Value Iteration for NS-POSGs}\label{sec:NS-HSVI}
\vspace*{-0.3em}

\noindent
To provide a more practical approach to solving one-sided NS-POSGs,
we now present a variant of HSVI (heuristic search value iteration)~\cite{TS-RS:04},
an anytime algorithm that approximates the value function $V^\star$ 
via lower and upper bound functions, updated through heuristically generated beliefs.

Our approach 
broadly follows the structure of
HSVI for \emph{finite} POSGs~\cite{KH-BB-VK-CK:23},
but every step presents challenges when extending to continuous states and NN-based observations.
In particular, we must work with integrals over beliefs and deal with uncountability,
using P-PWLC (rather than PWLC) 
functions for lower bounds, 
and therefore different ingredients to prove convergence.
Value computations are also much more complex because NN perception function induce FCPs, which are used to compute images, pre-images and intersections.

We also build on ideas from HVSI for (single-agent) neuro-symbolic POMDPs in~\cite{RY-GS-GN-DP-MK:23}.
The presence of two opposing agents brings three main challenges.
First, value backups at belief points 
require solving normal-form games instead of maximising over one agent's actions.
Second, since the first agent 
is not informed of the joint action, in the value backups and belief updates of the maxsup operator uncountably many stage strategies of the second agent have to be considered, whereas, in the single-agent variant, the agent can decide the transition probabilistically on its own.
Third, the forward exploration heuristic is more complex as 
it depends on the stage strategies of the agents in two-stage games.

\vspace*{-0.3cm}
\subsection{Lower and Upper Bound Representations}\label{subsec:lb_up_representations}
\vspace*{-0.1cm}

\noindent
We first discuss representing and updating the lower and upper bound functions.

\startpara{Lower bound function} 
Selecting an appropriate representation for $\alpha$-functions requires closure properties with respect to the maxsup operator. 
Motivated by~\cite{RY-GS-GN-DP-MK:23}, we represent the lower bound $V_{\mathit{lb}}^{\Gamma} \in \mathbb{F}(S_B)$ as the P-PWLC function
for a finite set $\Gamma \subseteq \mathbb{F}_C(S)$ of PWC $\alpha$-functions (see \defiref{defi:PWLC}),
for which the closure is guaranteed by \thomref{thom:P-PWLC-closure}. The lower bound $V_{\mathit{lb}}^{\Gamma}$ has a finite representation as each $\alpha$-function is PWC, and is initialised as in \cite{KH-BB-VK-CK:23}.

\startpara{Upper bound function}
The upper bound $V_{\mathit{ub}}^{\Upsilon} \in \mathbb{F}(S_B)$ is represented by a finite set of belief-value points $\Upsilon = \{ ((s_1^i, b_1^i), y_i) \in 
S_B \times \mathbb{R}  \mid i \in I \}$, where $y_i$ is an upper bound of $V^{\star}(s_1^i, b_1^i)$. 
Similarly to  \cite{RY-GS-GN-DP-MK:23}, for any $(s_1, b_1) \in S_B$,
the upper bound $V_{\mathit{ub}}^{\Upsilon}(s_1, b_1)$ is the lower envelope of the lower convex hull of the points in $\Upsilon$ satisfying the following LP problem: minimise
\begin{align}
\mbox{$\sum\nolimits_{i \in I_{s_{\scale{.75}{1}}}}$}\lambda_i y_i + K_{\mathit{ub}}(b_1, \mbox{$\sum\nolimits_{i \in I_{s_{\scale{.75}{1}}}}$} \lambda_i b_1^i) \; \mbox{\rm subject to} \;
 \lambda_i \ge 0 \; \mbox{and} \; \mbox{$\sum\nolimits_{i \in I_{s_{\scale{.75}{1}}}}$}  \lambda_i = 1  \label{eq:new-ub}
\end{align}
for $i \in I_{s_{\scale{.75}{1}}}$ where $I_{s_1} = \{ i \in I \mid s_1^i = s_1 \}$ and  $K_{\mathit{ub}} : \mathbb{P}(S_E) \times \mathbb{P}(S_E) \to \mathbb{R}$ measures the difference between two beliefs
such that, if $K$ is the function from \thomref{thom:convexity-continuity}, 
then for any $b_1, b_1', b_1'' \in \mathbb{P}(S_E)$: $K_{\mathit{ub}}(b_1, b_1) = 0$,
\begin{align}
    K_{\mathit{ub}}(b_1, b_1') \ge K(b_1, b_1') \quad \mbox{and} \quad  |K_{\mathit{ub}}(b_1, b_1') - K_{\mathit{ub}}(b_1, b_1'') | \leq K_{\mathit{ub}}(b_1', b_1'') \label{eq:K-UB-condition-2} \, .  
\end{align}
Note that \eqref{eq:new-ub} is close to the upper bound in regular HSVI for finite-state spaces, except for the function $K_{\mathit{ub}}$ that measures the difference between two beliefs (two continuous-state functions). With respect to the upper bound used in~\cite{RY-GS-GN-DP-MK:23},  $K_{\mathit{ub}}$ here needs to satisfy an additional triangle property in \eqref{eq:K-UB-condition-2} to ensure the continuity of $V_{\mathit{ub}}^{\Upsilon}$, for the convergence of the point-based algorithm below. The properties of $K_{\mathit{ub}}$ imply 
that \eqref{eq:new-ub} is an upper bound after a value backup, as stated in \lemaref{lema:upper-bound-update} below. The upper bound $V_{\mathit{ub}}^{\Upsilon}$ is initialised as in \cite{KH-BB-VK-CK:23}.

\startpara{Lower bound updates} For the lower bound $V_{\mathit{lb}}^{\Gamma}$, 
in each iteration we add a new PWC $\alpha$-function $\alpha^{\star}$ to $\Gamma$
at a belief $(s_1, b_1) \in S_{B}$ such that:
\begin{equation}\label{eq:update-lb-condition}
    \langle \alpha^{\star}, (s_1, b_1) \rangle = [TV_{\mathit{lb}}^{\Gamma}](s_1, b_1) =  \langle f_{\overline{p}_1^{\star}, \overline{\alpha}^{\star}}, (s_1, b_1) \rangle
\end{equation}
where the second equality follows from 
\lemaref{lema:LP-minimax-P-PWLC} and $(\overline{p}_1^{\star}, \overline{\alpha}^{\star})$ is computed via the optimal solution to the LP in \lemaref{lema:LP-minimax-P-PWLC} at $(s_1, b_1)$.

Using $\overline{p}_1^\star$, $\overline{\alpha}^{\star}$ and the perception FCP $\Phi_P$, \algoref{alg:point-based-update-belief} computes a new $\alpha$-function $\alpha^{\star}$ at belief $(s_1, b_1)$. To guarantee \eqref{eq:update-lb-condition} and improve efficiency, we only compute the backup values for regions $\phi \in \Phi_P$ over which $(s_1, b_1)$ has positive probabilities, i.e., $s^\phi_1=s_1$ (where $s^\phi_1$ is the unique agent state appearing in $\phi$) and $\int_{(s_1, s_E) \in \phi} b_1(s_E) \textup{d} s_E > 0$, 
and assign the trivial lower bound $L$ otherwise. 

For each region $\phi$ either $\alpha^{\star}(\hat{s}_1,\hat{s}_E) =  f_{\overline{p}_1^\star, \overline{\alpha}^{\star}} (\hat{s}_1,\hat{s}_E)$ 
or $\alpha^{\star}(\hat{s}_1,\hat{s}_E) = L$ for all $(\hat{s}_1,\hat{s}_E) \in \phi$. 
Computing the backup values in line 4 of \algoref{alg:point-based-update-belief} state by state is computationally intractable, as $\phi$ contains an infinite number of states. 
However, the following lemma shows that $\alpha^{\star}$ is PWC, allowing a tractable region-by-region backup, called Image-Split-Preimage-Product (ISPP) backup, which is adapted from the single-agent variant in \cite{RY-GS-GN-DP-MK:23}. 
The details of the ISPP backup for one-sided NS-POSGs are in \appxref{sec:appendix-ISPP}.
The lemma also shows that the lower bound function increases
and is valid after each update.

\begin{algorithm}[t]
\caption{Point-based $\mathit{Update}(s_1, b_1)$ of $(V_{\mathit{lb}}^{\Gamma}, V_{\mathit{ub}}^{\Upsilon})$}
\label{alg:point-based-update-belief}
\begin{algorithmic}[1] 
\State $(\overline{p}_1^{\star}, \overline{\alpha}^{\star}) \leftarrow$ $[TV_{\mathit{lb}}^{\Gamma}](s_1, b_1)$ via an LP in  \lemaref{lema:LP-minimax-P-PWLC} 
\For{$\phi \in \Phi_P$}
\If{$s^\phi_1 = s_1$ and $\int_{(s_1, s_E) \in \phi} b_1(s_E) \textup{d} s_E > 0$}
\State $\alpha^{\star}(\hat{s}_1,\hat{s}_E) \leftarrow f_{\overline{p}_1^{\star}, \overline{\alpha}^{\star}} (\hat{s}_1,\hat{s}_E)$ for $(\hat{s}_1,\hat{s}_E) \in \phi$ \Comment{ISPP backup} 
\Else $\;\alpha^{\star}(\hat{s}_1,\hat{s}_E) \leftarrow L$ for $(\hat{s}_1,\hat{s}_E) \in \phi$
\EndIf
\EndFor
\State $\Gamma \leftarrow \Gamma \cup \{\alpha^{\star}\}$
\State $y^{\star} \leftarrow [TV_{\mathit{ub}}^{\Upsilon}](s_1, b_1)$ via \eqref{eq:minimax-operator} and \eqref{eq:new-ub}
\State $\Upsilon \leftarrow \Upsilon \cup \{((s_1, b_1), y^{\star}) \}$
\end{algorithmic}
\end{algorithm}

\begin{lema}[Lower bound]\label{lema:new-pwc-alpha}
The function $\alpha^{\star}$ generated by \algoref{alg:point-based-update-belief} is a PWC $\alpha$-function satisfying \eqref{eq:update-lb-condition}, and if $\Gamma' = \Gamma \cup \{\alpha^{\star}\}$, then $V_{\mathit{lb}}^{\Gamma} \leq V_{\mathit{lb}}^{\Gamma'} \leq V^{\star}$. 
\end{lema}


\startpara{Upper bound updates}
For the upper bound $V_{\mathit{ub}}^{\Upsilon}$, due to representation \eqref{eq:new-ub}, at a belief $(s_1, b_1) \in S_B$ in each iteration, we add a new belief-value point $((s_1, b_1), y^{\star})$ to $\Upsilon$ such that $y^{\star} = [TV_{\mathit{ub}}^{\Upsilon}](s_1, b_1)$. Computing $[TV_{\mathit{ub}}^{\Upsilon}](s_1, b_1)$ via \eqref{eq:minimax-operator} and \eqref{eq:new-ub} requires the concrete formula for $K_{\mathit{ub}}$ and the belief representations. Thus, we will show how to compute $[TV_{\mathit{ub}}^{\Upsilon}](s_1, b_1)$ when introducing belief representations below.
The following lemma shows that $y^{\star} \ge V^{\star} (s_1, b_1)$ required by \eqref{eq:new-ub}, and the upper bound function is decreasing and is valid after each update.

\begin{lema}[Upper bound]\label{lema:upper-bound-update}
Given a belief $(s_1, b_1) \in S_{B}$, if $y^{\star} = [TV_{\mathit{ub}}^{\Upsilon}](s_1, b_1)$, then $y^\star$ is an upper bound of $V^{\star}$ at $(s_1, b_1)$, i.e., $y^{\star} \ge V^{\star} (s_1, b_1)$, and if $\Upsilon' = \Upsilon \cup \{ ((s_1, b_1), y^{\star}) \}$, then $V_{\mathit{ub}}^{\Upsilon} \ge V_{\mathit{ub}}^{\Upsilon'} \ge V^{\star}$.
\end{lema}

\subsection{One-Sided NS-HSVI}

\noindent
\algoref{alg:NS-HSVI} presents our NS-HSVI algorithm for one-sided NS-POSGs.

\startpara{Forward exploration heuristic}
The algorithm uses a heuristic approach to select which belief 
will be considered next.
Similarly to 
finite-state one-sided POSGs \cite{KH-BB-VK-CK:23}, we focus on a belief that has the highest \emph{weighted excess gap}. The excess gap at a belief $(s_1, b_1)$ with depth $t$ from the initial belief is defined by $ \mathit{excess}_{t}(s_1, b_1) = V_{\mathit{ub}}^{\Upsilon}(s_1, b_1) -  V_{\mathit{lb}}^{\Gamma}(s_1, b_1) - \rho(t)$,
where $\rho(0) = \varepsilon$ and $\rho(t {+} 1) = (\rho(t) - 2(U - L) \bar{\varepsilon}) / \beta $, and $\bar{\varepsilon} \in (0, (1 - \beta ) \varepsilon / (2U - 2L))$. 
Using this excess gap, the next action-observation pair $(\hat{a}_1, \hat{s}_1)$ for exploration is selected from:
\vspace{-0.2cm}
\begin{equation} \label{eq:max-action-observation}
     \argmax\nolimits_{(a_1,s_1') \in A_1 \times S_1} P(a_1, s_1' \mid (s_1, b_1),u_1^{\mathit{ub}}, u_2^{\mathit{lb}}) \mathit{excess}_{t+1}(s_1', b_1^{s_1,a_1,u_2^{\mathit{lb}}, s_1'}) \,.
\end{equation}
To compute the next belief via lines 8 and 9 of \algoref{alg:NS-HSVI}, 
the minimax strategy profiles in stage games $[TV_{\mathit{lb}}^{\Gamma}](s_1, b_1)$ and $[TV_{\mathit{ub}}^{\Upsilon}](s_1, b_1)$, i.e., $(u_1^{\mathit{ub}}, u_2^{\mathit{lb}})$, are required. Since $V_{\mathit{lb}}^{\Gamma}$ is P-PWLC, using \lemaref{lema:LP-minimax-P-PWLC}, the strategy  $ u_2^{\mathit{lb}}$ is obtained by solving an LP.
However, the computation of the strategy  $u_1^{\mathit{ub}}$ depends on the representation of $(s_1, b_1)$ and the measure function $K_{\mathit{ub}}$, and thus will be discussed later.
One-sided NS-HSVI has the following convergence guarantees.


\begin{algorithm}[t]
\caption{One-sided NS-HSVI for one-sided NS-POSGs}
\label{alg:NS-HSVI}
\begin{algorithmic}[1] 
\While{$V_{\mathit{ub}}^{\Upsilon}(s_{1}^{\mathit{init}}, b_{1}^{\mathit{init}}) - V_{\mathit{lb}}^{\Gamma}(s_{1}^{\mathit{init}}, b_{1}^{\mathit{init}}) > \varepsilon$} $\mathit{Explore}((s_{1}^{\mathit{init}}, b_{1}^{\mathit{init}}), 0)$
\EndWhile
\State \Return $V_{\mathit{lb}}^{\Gamma}$  and $V_{\mathit{ub}}^{\Upsilon}$ via sets $\Gamma$ and $\Upsilon$
\Function{$\mathit{Explore}$}{$(s_1, b_1),  t$}
\State $(u_1^{\mathit{lb}}, u_2^{\mathit{lb}}) \leftarrow $ minimax strategy profile in $[TV_{\mathit{lb}}^{\Gamma}](s_1, b_1)$
\State $(u_1^{\mathit{ub}}, u_2^{\mathit{ub}}) \leftarrow $ minimax strategy profile in $[TV_{\mathit{ub}}^{\Upsilon}](s_1, b_1)$
\State $\mathit{Update}(s_1, b_1)$ \Comment{\algoref{alg:point-based-update-belief}}
\State $(\hat{a}_1, \hat{s}_1) \leftarrow$ select according to forward exploration heuristic 
\If{$ P(\hat{a}_1, \hat{s}_1 \mid (s_1, b_1), u_1^{\mathit{ub}}, u_2^{\mathit{lb}}) \mathit{excess}_{t+1}(\hat{s}_1, b_1^{s_1,\hat{a}_1,u_2^{\mathit{lb}}, \hat{s}_1}) > 0$}
\State $\mathit{Explore}((\hat{s}_1, b_1^{s_1,\hat{a}_1,u_2^{\mathit{lb}}, \hat{s}_1}),  t+1)$
\State $\mathit{Update}(s_1, b_1)$ \Comment{\algoref{alg:point-based-update-belief}}
\EndIf
\EndFunction
\end{algorithmic}
\end{algorithm}

\begin{theorem}[One-sided NS-HSVI]\label{thom:NS-HSVI} For any $(s_{1}^{\mathit{init}}, b_{1}^{\mathit{init}}) \in S_B$ and $\varepsilon > 0$,
\algoref{alg:NS-HSVI} will terminate and upon termination: $V_{\mathit{ub}}^{\Upsilon}(s_{1}^{\mathit{init}}, b_{1}^{\mathit{init}}) - V_{\mathit{lb}}^{\Gamma}(s_{1}^{\mathit{init}}, b_{1}^{\mathit{init}}) \leq \varepsilon$ and  $V_{\mathit{lb}}^{\Gamma}(s_{1}^{\mathit{init}}, b_{1}^{\mathit{init}}) \leq V^{\star}(s_{1}^{\mathit{init}}, b_{1}^{\mathit{init}}) \leq V_{\mathit{ub}}^{\Upsilon}(s_{1}^{\mathit{init}}, b_{1}^{\mathit{init}})$.
%
\end{theorem}

\subsection{Belief Representation and Computations}\label{sec:beliefref}

\noindent
Implementing one-sided NS-HSVI depends on belief representations, as closed forms are needed.
We use the popular \emph{particle-based representation} \cite{JMP-NV-MTS-PP:06,AD-NDF-NJG:01}, which can approximate arbitrary beliefs and handle non-Gaussian systems.
However, compared to region-based  representations \cite{RY-GS-GN-DP-MK:23},
it is more vulnerable to disturbances and can require many particles for a good approximation.
 
\startpara{Particle-based beliefs} 
A \emph{particle-based belief} $(s_1, b_1) \in S_B$ is represented by a weighted particle set $\{ (s_E^i, \kappa_i) \}_{i=1}^{n_s}$ with a normalised weight $\kappa_i$ for each particle $s_E^i \in S_E$, where $b_1(s_E) = \mbox{$\smash{\sum\nolimits_{i = 1}^{n_b}}$} \kappa_i D(s_E - s_E^i)$ for $s_E \in S_E$ and $D(s_E - s_E^i)$ is a Dirac delta function centred at $0$.  

 
To implement one-sided NS-HSVI using particle-based beliefs, we prove that $V_{\mathit{lb}}^{\Gamma}$ and $V_{\mathit{ub}}^{\Upsilon}$ are eligible representations,
as the belief update $\smash{b_1^{s_1, a_1, u_2, s_1'}}$, expected values $\langle \alpha, (s_1, b_1) \rangle$, $\langle r, (s_1, b_1) \rangle$ and probability $P(a_1, s_1' \mid (s_1,b_1), u_1, u_2)$ are computed as simple summations for a particle-based belief $(s_1, b_1)$ (\appxref{sec:appendix-probabilities}).

\startpara{Lower bound} 
Since $V_{\mathit{lb}}^{\Gamma}$ is P-PWLC with PWC $\alpha$-functions $\Gamma$, for a particle-based belief $(s_1, b_1)$ represented by $\{ (s_E^i, \kappa_i) \}_{i=1}^{n_b}$, using \defiref{defi:PWLC}, $V_{\mathit{lb}}^{\Gamma}(s_1, b_1) = \max\nolimits_{\alpha \in \Gamma} \sum_{i=1}^{n_b} \kappa_i \alpha(s_1, s_E^i)$. The stage game $[TV_{\mathit{lb}}^{\Gamma}](s_1, b_1)$ and minimax strategy profile $(u_1^{\mathit{lb}},u_2^{\mathit{lb}} )$ follow from solving the LP in \lemaref{lema:LP-minimax-P-PWLC}.

\startpara{Upper bound} 
To compute $V_{\mathit{ub}}^{\Upsilon}$ in \eqref{eq:new-ub}, we need a function $K_{\mathit{ub}}$ to measure belief differences that satisfies \eqref{eq:K-UB-condition-2}. We take $K_{\mathit{ub}} = K$, which does so by definition. 
Given $\Upsilon = \{ ((s_1^i, b_1^i), y_i) \mid i \in I\}$, the upper bound and stage game can be computed by solving an LP, respectively, as demonstrated by the following theorem, and then the minimax strategy profile
$(u_1^{\mathit{ub}}, u_2^{\mathit{ub}})$ is synthesised (see \appxref{sec:appendix-dual-LPs}).
\begin{theorem}[LPs for upper bound]\label{thom:LP-upper-bound}
    For a particle-based belief $(s_1, b_1) \in S_B$, $V_{\mathit{ub}}^{\Upsilon}(s_1,b_1)$ and $[TV_{\mathit{ub}}^{\Upsilon}](s_1,b_1)$ are the optimal value of an LP, respectively.
\end{theorem}

\vspace*{-0.3cm}
\section{Experimental Evaluation}\label{sec:experiments}
\vspace*{-0.2cm}

%
\noindent
We have built a prototype implementation in Python, using Gurobi~\cite{gurobi} to solve the LPs needed for computing lower and upper bound values, and the minimax values and strategies of one-shot games.
We use the Parma Polyhedra Library~\cite{BHZ08}
to operate over polyhedral pre-images of NNs, $\alpha$-functions and reward structures. 

Our evaluation uses two one-sided NS-POSG examples:
a \emph{pursuit-evasion} game and the \emph{pedestrian-vehicle} scenario from \sectref{nscsgs-sect}.
Below, we discuss the applicability and usefulness of our techniques on these examples.
Due to limited space,
we refer to \appxref{sec:appendix-examples} for more details of the models,
including the training of the ReLU NN classifiers,
and empirical results on performance.

\begin{figure}[t]
{\scriptsize
    \hspace*{0.9cm} Step 0 \hspace*{1.95cm} Step 1 \hspace*{1.9cm} Step 2 \hspace*{1.95cm} Step 3 \\
    \centering
    \raisebox{1.3cm}{(a)}
    \hspace{-0.15cm}
    \includegraphics[width=0.24\textwidth]{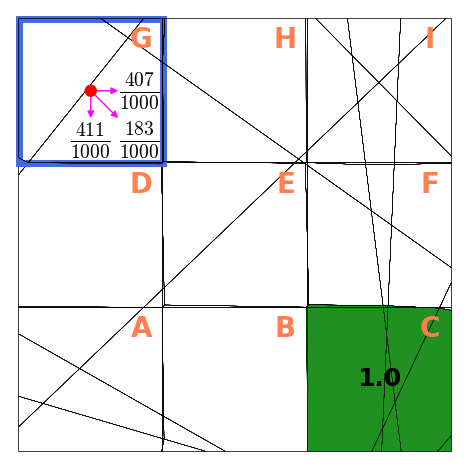}
    \hspace{-0.28cm}
    \includegraphics[width=0.24\textwidth]{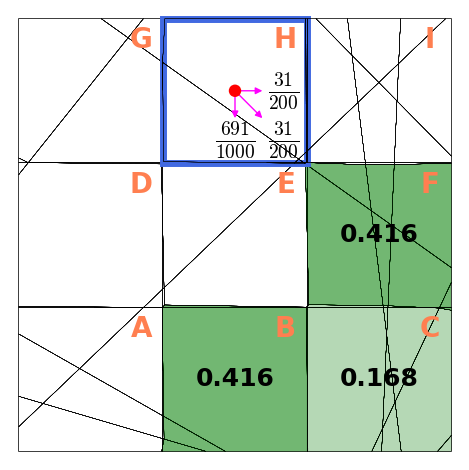}
    \hspace{-0.28cm}
    \includegraphics[width=0.24\textwidth]{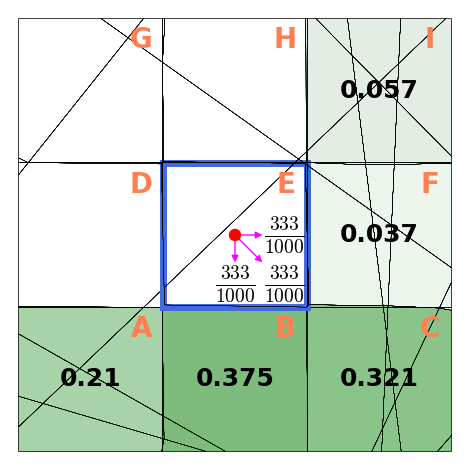}
    \hspace{-0.28cm}
    \includegraphics[width=0.24\textwidth]{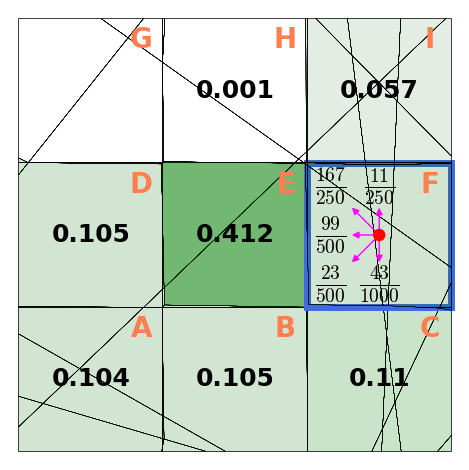} \\
    \vspace*{-0.15cm}
    \hspace{0.15cm}\raisebox{1.3cm}{(b)}
    \hspace{-0.15cm}
    \includegraphics[width=0.24\textwidth]{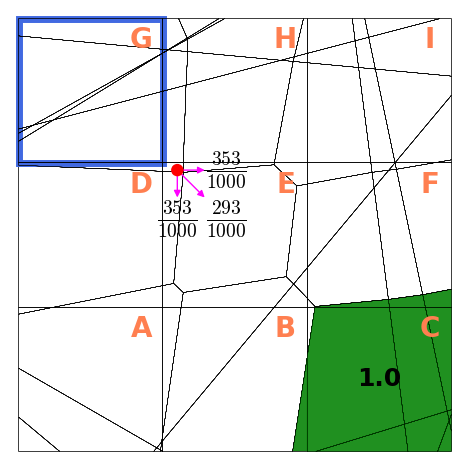}
    \hspace{-0.28cm}
    \includegraphics[width=0.24\textwidth]{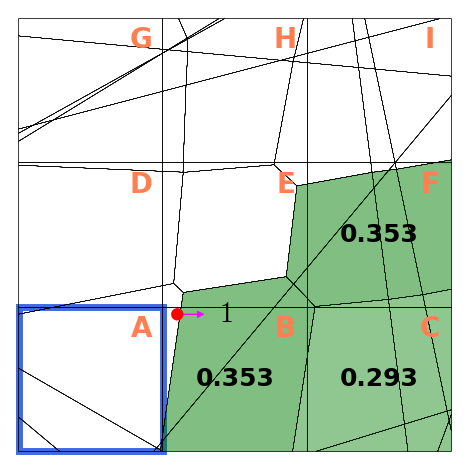}
    \hspace{-0.28cm}
    \includegraphics[width=0.24\textwidth]{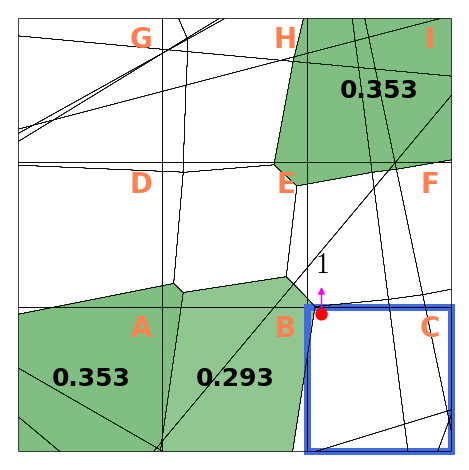}
    \hspace{-0.28cm}
    \includegraphics[width=0.24\textwidth]{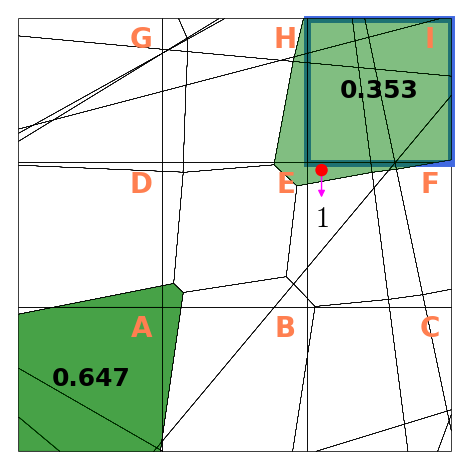}}
    
    \vspace*{-0.6em}
    \caption{Simulations of strategies for the pursuer, showing actual location (red), perceived location (blue), belief of evader location (green) and strategy (pink) for two different NN perception functions: (a) more precise; (b) coarser.}
    \label{fig:pursuit_evasion_strat_both}
\end{figure}

\vspace*{-0.3em}
\startpara{Pursuit-evasion} 
A pursuit-evasion game models a \emph{pursuer} trying to catch an \emph{evader} aiming to avoid capture.
We build a continuous-space variant of the model 
from~\cite{KH-BB-VK-CK:23} inspired by mobile robotics applications~\cite{CHI11,VN08}.
The environment includes the exact position of both agents.
The (partially informed) pursuer uses an NN classifier to perceive its own location,
which maps to one of $3{\times}3$ grid cells.
To showcase the ability of our methodology to assess the performance of realistic NN perception functions,
we train two NNs, the second with a coarser accuracy.

\figref{fig:pursuit_evasion_strat_both} shows simulations 
of strategies synthesised 
for the pursuer, using the two different NNs.
Its actual location is a red dot, and the pink arrows denote the strategy.
Blue squares show the cell that is output by the pursuer's perception function,
and black lines mark the underlying polyhedral decomposition.
The pursuer's belief over the evader's location is shown by the green shading and annotated probabilities;
it initially (correctly) believes that the evader is in cell $C$
and the belief evolves based on the optimal counter-strategy of the evader.


The plots show we can synthesise
non-trivial strategies for agents using NN-based perception in a partially observable setting.
We can also study the impact of a poorly trained perception function.
\figref{fig:pursuit_evasion_strat_both}(b), for the coarser NN,
shows the pursuer repeatedly mis-detecting its location
because the grid cells shapes are poorly approximated,
and subsequently taking incorrect actions.
This is exploited by the evader, leading to considerably worse performance for the pursuer. 

\startpara{Pedestrian-vehicle interaction}
\figref{fig:pedestrian_vehicle_paths} shows several simulations from strategies synthesised for the pedestrian-vehicle example described in \sectref{nscsgs-sect} (\figref{fig:pedestrian_vehicle_collated}),
plotting the position $(x_2,y_2)$ of the pedestrian, relative to the vehicle.
We fix the pedestrian's strategy, to simulate a crossing scenario:
it moves from right to left, i.e., decreasing $x_2$.
The (partially informed) vehicle's perception function predicts the intention of the pedestrian (green/yellow/red = \emph{unlikely}/\emph{likely}/\emph{very likely} to cross), shown as coloured dots.
Above and below each circle, we indicate the acceleration actions taken (black) and current speeds (orange), respectively,
which determine the distance $y_2$ to the pedestrian crossing.

Again, we investigate the feasibility of
generating strategies for agents with realistic NN-based perception.
Here, the goal is to avoid a crash scenario, denoted by the shaded region at the bottom left of the plots.
We find that, in many cases, safe strategies can be synthesised.
\figref{fig:pedestrian_vehicle_paths}(a) shows an example;
notice that the pedestrian intention is detected early.
This is not true in (b) and (c), which show two simulations
from a strategy and starting point
where the perception function results in much later detection;
(c) shows we were then unable to synthesise a strategy for the vehicle that is always safe.

\begin{figure}[t]
    \begin{subfigure}{0.32\textwidth}
    \centering
    \includegraphics[width=1\textwidth]{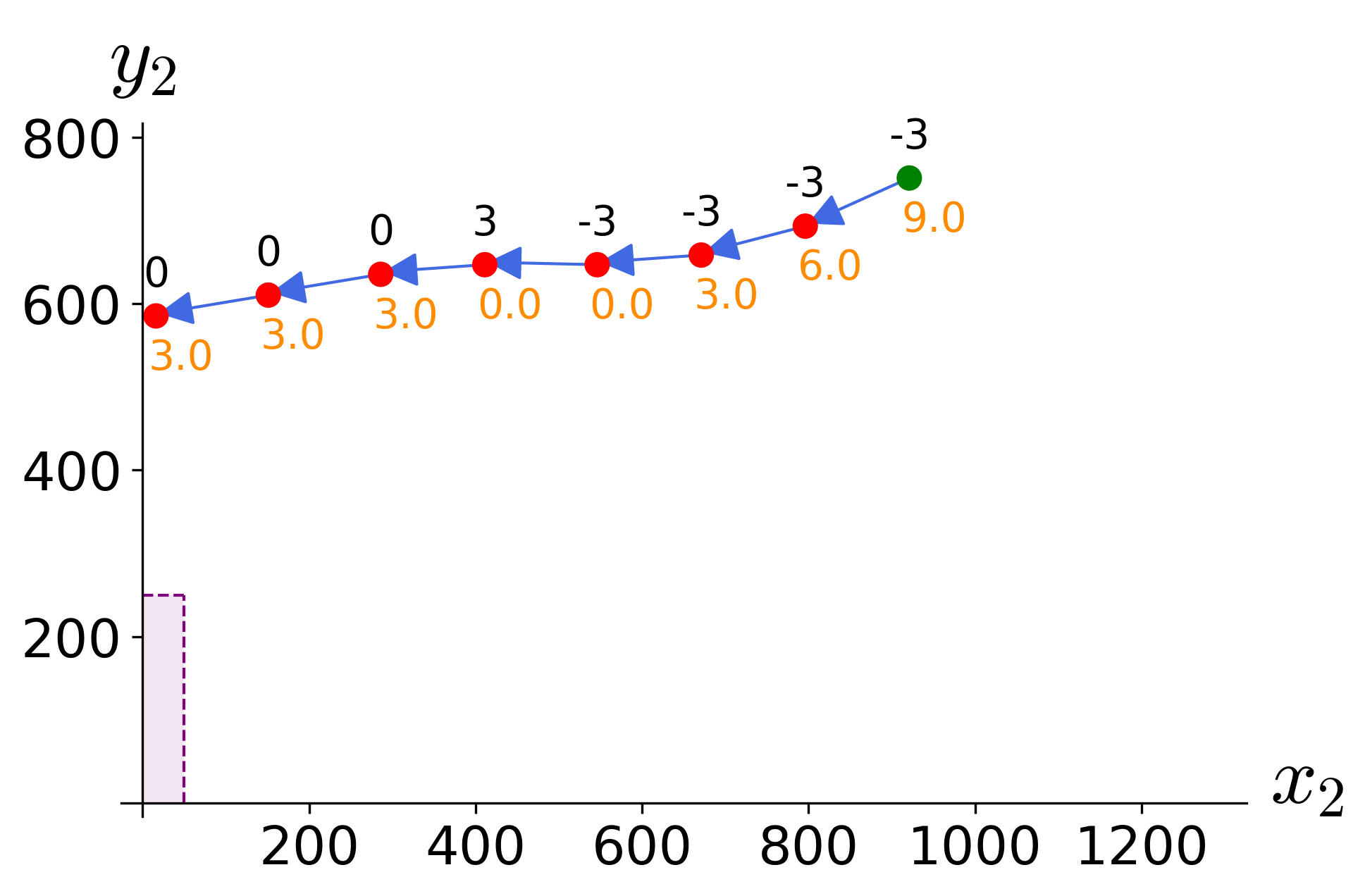}
    \caption{}
    \end{subfigure}
    \begin{subfigure}{0.32\textwidth}
    \centering
    \includegraphics[width=1\textwidth]{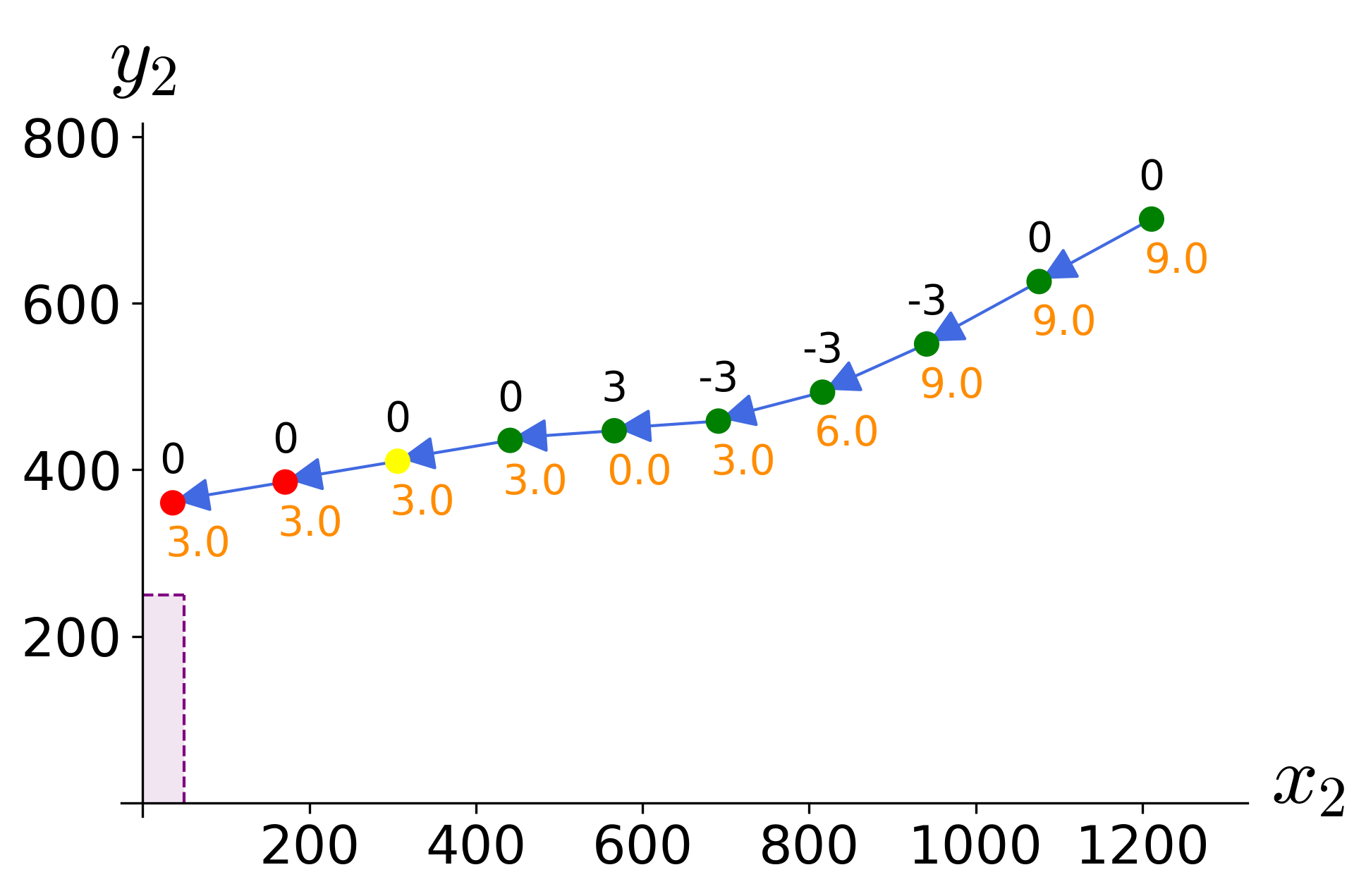}
    \caption{}
    \end{subfigure}
    \begin{subfigure}{0.32\textwidth}
    \centering
    \includegraphics[width=1\textwidth]{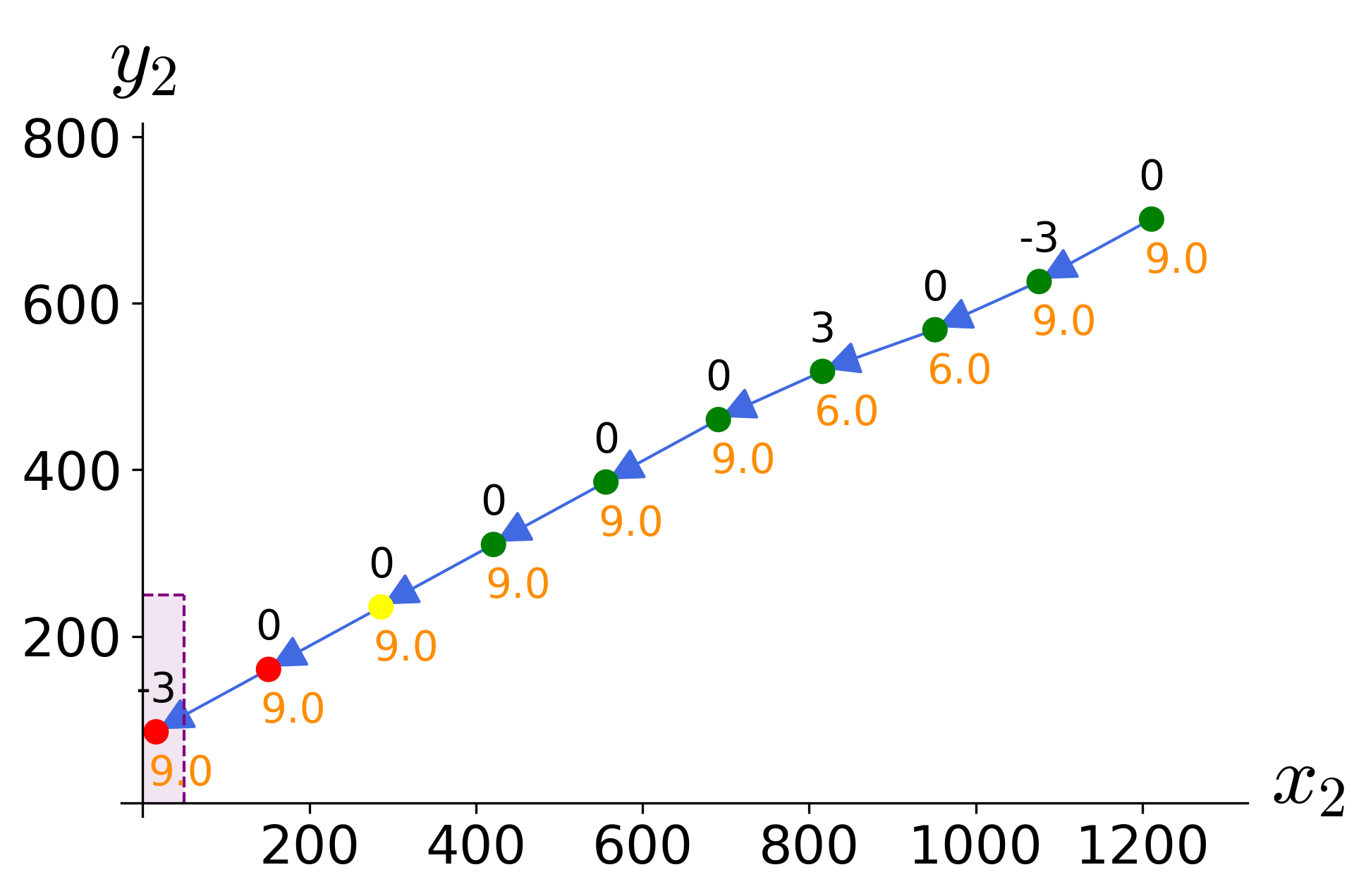}
    \caption{}
    \end{subfigure}
    \vspace{-0.2cm}
    \caption{Simulations of strategies for the vehicle, plotted as the pedestrian's current position $(x_2,y_2)$ relative to it. Also shown: perceived pedestrian intention (green/yellow/red = \emph{unlikely}/\emph{likely}/\emph{very likely} to cross), current speed (orange), acceleration (black) and crash region (shaded purple region).}
    \label{fig:pedestrian_vehicle_paths}
\end{figure}
    

\vspace*{-0.3cm}
\section{Conclusions}
\vspace*{-0.1cm}

\noindent
We have proposed one-sided neuro-symbolic POSGs,
designed to reason formally about partially observable agents
equipped with neural perception mechanisms.
We characterised the value function for discounted infinite-horizon rewards,
and designed, implemented and evaluated a HSVI algorithm for
approximate solution.
Computational complexity is high due to expensive polyhedral operations. Nevertheless, our method provides an important baseline that
can reason about true decision boundaries 
for game models with NN-based perception, against which efficiency improvements can later be benchmarked.
We plan to investigate ways to improve performance,
e.g., merging of adjacent polyhedra 
or Monte-Carlo planning methods, 
and to study restricted cases of two-sided NS-POSGs, e.g., those    with public observations \cite{KH-BB:19}.  


\startpara{Acknowledgements}
This project was funded by the ERC under the European
Union’s Horizon 2020 research and innovation programme
(FUN2MODEL, grant agreement No.834115).

\vfill

\bibliographystyle{splncs04} 
\bibliography{references}

\ifthenelse{\isundefined{\techreport}}{%
}{%
\setlength{\textfloatsep}{12pt}
\clearpage
\appendix


\section{Probability Measure Computations}\label{sec:appendix-probabilities}

\noindent
The main paper omits details of how to compute 
several required quantities in terms of probability measures via closed forms.
We provide these details below.

\startpara{Belief updates}
\sectref{nscsgs-sect} (p. \pageref{nsposgbelief})
discusses belief updates for agent $\agent_1$ of a one-sided NS-POSG.
Given a belief $(s_1, b_1)$, if action $a_1$ is selected by $\agent_1$, $\agent_2$ is \emph{assumed} to take the stage strategy $u_2 \in \mathbb{P}(A_2 \mid S)$ and $s_1'$ is observed, then the updated belief of $\agent_1$  via Bayesian inference is $(s_1', b_1^{s_1,a_1,u_2,s_1'})$ where for $s_E' \in S_E$:
\begin{align}
    b_1^{s_1,a_1,u_2,s_1'}(s_E') = \frac{P( (s_1', s_E') \mid (s_1, b_1), a_1, u_2)}{P( s_1' \mid (s_1, b_1), a_1, u_2 ) } \textup{ if $s_E' \in S_E^{s_1'}$ and $0$ otherwise.} \label{eq:belief-update-u2}
\end{align}
On the other hand, if it is \emph{assumed} that a joint action $a$ is taken, then the updated belief of $\agent_1$ is $(s_1', b_1^{s_1,a,s_1'})$, where for $s_E' \in S_E$:
\begin{equation}\label{eq:belief-update-omega}
  b_1^{s_1,a,s_1'}(s_E') = \frac{P( (s_1', s_E') \mid (s_1, b_1), a )}{P( s_1' \mid (s_1, b_1), a )} \textup{ if $s_E' \in S_E^{s_1'}$
and $0$ otherwise.}  
\end{equation} 
We now show how to compute the probability values given in the belief updates \eqref{eq:belief-update-u2} and \eqref{eq:belief-update-omega}. Recalling that $s_1 = (\loc_1, \per_1)$, for \eqref{eq:belief-update-u2}, using the syntax in \defiref{defi:NS-CSG}, $P( s_1' \mid (s_1, b_1), a_1, u_2 )$ equals
\begin{equation}\label{eq:belief-update-closed2}
    \mbox{$\int_{s_E \in S_E}$} b_1(s_E)  \mbox{$\sum_{a_2 \in A_2}$} u_2(a_2 \mid s_1, s_E)  \mbox{$ \int\nolimits_{s_E' \in S_E}$}
   \delta((s_1, s_E), (a_1, a_2))(s_1',s_E') \textup{d} s_E
\end{equation}
and if $s_E' \in S_E^{s_1'}$, then $P((s_1', s_E') \mid (s_1, b_1), a_1, u_2)$ equals
\begin{equation*}
    \mbox{$\int_{s_E \in S_E}$} b_1(s_E) \mbox{$\sum_{a_2 \in A_2}$}  u_2(a_2 \mid s_1, s_E) \delta((s_1, s_E), (a_1, a_2))(s_1',s_E') \textup{d} s_E \,.
\end{equation*}
For \eqref{eq:belief-update-omega}, we have that $P(s_1' \mid (s_1, b_1), a)$ equals
\begin{equation*}
    \mbox{$\int_{s_E \in S_E}$} b_1(s_E) \mbox{$\int_{s_E' \in S_E}$} 
   \delta((s_1, s_E), a)(s_1',s_E') \textup{d} s_E 
\end{equation*}
and if $s_E' \in S_E^{s_1'}$, then $P( (s_1', s_E') \mid (s_1, b_1), a )$ equals
\begin{equation*}
    \mbox{$\int_{s_E \in S_E}$} b_1(s_E) \delta((s_1, s_E), a)(s_1',s_E') \textup{d} s_E \,.
\end{equation*}

\startpara{Particle-based beliefs}
\sectref{sec:beliefref} discusses computation of particle-based beliefs.
For a particle-based belief $(s_1, b_1)$ with weighted particle set $\{ (s_E^i, \kappa_i) \}_{i=1}^{n_b}$, it follows from \eqref{eq:belief-update-u2} that for belief $b_1^{s_1, a_1, u_2, s_1'}$ we have, for any $s_E' \in S_E$, that $b_1^{s_1, a_1, u_2, s_1'}(s_E')$ equals
\begin{equation}
    \frac{\sum\nolimits_{i = 1}^{n_b}   \kappa_i \left( \mbox{$\sum_{a_2 \in A_2}$}  u_2(a_2 \mid s_1, s_E^i) \delta((s_1, s_E^i), (a_1, a_2)) (s_1', s_E') \right) }{\sum\nolimits_{i = 1}^{n_b} \kappa_i \left(\mbox{$\sum_{a_2 \in A_2}$} u_2(a_2 \mid s_1, s_E^i) \left( \sum\nolimits_{s_E'' \in S_E}
   \delta((s_1, s_E^i), (a_1, a_2)) (s_1', s_E'')\right) \right)  }
\end{equation}
if $s_E' \in S_E^{s_1'}$ and equals 0 otherwise.
Similarly, we can compute $\langle \alpha, (s_1, b_1) \rangle$, $\langle r, (s_1, b_1) \rangle$ and $P(a_1, s_1' \mid (s_1,b_1), u_1, u_2)$ as simple summations.

\section{Image-Split-Preimage-Product (ISPP) Backup}\label{sec:appendix-ISPP}

\noindent
In this section we present the Image-Split-Preimage-Product (ISPP) backup for one-sided NS-POSGs, adapted from the single-agent variant in \cite{RY-GS-GN-DP-MK:23},
as used for a region-by-region backup in line 4 of \algoref{alg:point-based-update-belief} (\sectref{subsec:lb_up_representations}). 

For FCPs $\Phi_1$ and $\Phi_2$ of $S$, we denote by $\Phi_1+\Phi_2$ the smallest FCP 
of $S$ such that $\Phi_1+\Phi_2$ is a refinement of both $\Phi_1$ and $\Phi_2$, which can be obtained by taking all the intersections between regions of $\Phi_1$ and $\Phi_2$. We call the FCP $\Phi$ in \defiref{defi:PWC-func} the \emph{constant-FCP} of $S$ for a PWC function $f \in \mathbb{F}_C(S)$.
Recall from \aspref{asp:transitions-rewards} that $\delta_E$ can be represented as $\sum_{i=1}^{n} \mu_i \delta_E^i$,
where $n \in \mathbb{N}$, $\mu_i\in[0,1]$, $\sum_{i = 1}^{n} \mu_i = 1$
and $\delta_E^i : (\Loc_1 {\times} S_E {\times} A) \to S_E$ are piecewise continuous functions.

\algoref{alg:ISPP-backup} shows the ISPP backup method.
 This method, inspired by \lemaref{lema:new-pwc-alpha}, is to divide a region $\phi$ into subregions where for each subregion $\alpha^{\star}$ is constant. Given any reachable local state $\loc_1'$ under $a$ and continuous transition function $\delta_E^i$, the \emph{image} of $\phi$ under $a$ and $\delta_E^i$ to $\loc_1'$ is divided into \emph{image} regions $\Phi_{\textup{image}}$ such that the states in each region have a unique agent state. Each image region $\phi_{\textup{image}}$ is then split into subregions by a constant-FCP of the PWC function $\alpha^{a_1, s_1^{\phi_{\scale{.75}{\textup{image}}}}}$ by pairwise intersections where $a = (a_1, a_2)$, and thus $\Phi_{\textup{image}}$ is \emph{split}  into a set of refined image regions $\Phi_{\textup{split}}$. An FCP over $\phi$, denoted by $\Phi_{\textup{pre}}$, is constructed by computing the \emph{pre-image} of each $\phi_{\textup{image}} \in \Phi_{\textup{split}}$ to $\phi$. Finally, the \emph{product} of these FCPs  $\Phi_{\textup{pre}}$ for all reachable local states and environment functions and reward FCPs $\{\Phi_R^{a} \mid a \in \bar{A}_1 \times A_2 \}$, denoted $\Phi_{\textup{product}}$, is computed. The following lemma demonstrates that $\alpha^{\star}$ is constant in each region of $\Phi_{\textup{product}}$, and therefore that line 4 of \algoref{alg:point-based-update-belief} 
can be computed by finite backups. 


\begin{algorithm}[t]
\caption{Image-Split-Preimage-Product (ISPP) backup over a region}
\label{alg:ISPP-backup}
\textbf{Input}: region $\phi$, action $\overline{p}_1^{\star}$, PWC functions $\overline{\alpha}^{\star}$
\begin{algorithmic}[1] 
\State $\bar{A}_1 \leftarrow \{ a_1 \in A_1 \mid \overline{p}_1^{\star}(a_1) > 0 \}$
\State $\Loc'_{a} \leftarrow \{ \loc_1' \in \Loc_1 \mid \delta_1(s^\phi_1, a)(\loc_1') > 0 \} $ for $a \in \bar{A}_1 \times A_2$, $\Phi_{\textup{product}} \leftarrow \phi $
\For{$a = (a_1, a_2) \in \bar{A}_1 \times A_2, \loc_1' \in \Loc_{a}', i = 1, \dots, n$} 
\State $\phi_{E}' \leftarrow \{ \delta_E^i(s_E, a) \mid (s^\phi_1, s_E) \in \phi \}$ \Comment{Image}
\State $\Phi_{\textup{image}} \leftarrow \textup{divide } \phi_{E}' \textup{ into regions over } S \textup{ by } \obs_1(\loc'_1, \cdot )$
\State $\Phi_{\textup{split}} \leftarrow \emptyset$ \Comment{Split}
\For{$\phi_{\textup{image}} \in \Phi_{\textup{image}}$}
\State $\Phi_{\alpha} \leftarrow \textup{a constant-FCP of } S \textup{ for the PWC function } \alpha^{\star a_1, s_1^{\phi_{\scale{.75}{\textup{image}}}}}$
\State $\Phi_{\textup{split}} \leftarrow \Phi_{\textup{split}} \cup  \{ \phi_{\textup{image}} \cap \phi' \mid \phi' \in \Phi_{\alpha} \}$
\EndFor
\State $\Phi_{\textup{pre}} \leftarrow \emptyset$  \Comment{Preimage}
\For{$\phi_{\textup{image}} \in \Phi_{\textup{split}}$}
\State $\Phi_{\textup{pre}} \leftarrow \Phi_{\textup{pre}}  \cup \{ (s^\phi_1, s_E) \in \phi \mid \delta_E^{i}(s_E, a ) \in \phi_{\textup{image}} \}$ 
\EndFor
\State $\Phi_{\textup{product}} \leftarrow \{ \phi_1 \cap \phi_2 \mid \phi_1 \in \Phi_{\textup{pre}} \wedge \phi_2 \in \Phi_{\textup{product}} \}$  \Comment{Product}
\EndFor
\State $\Phi_{\textup{product}} \leftarrow \{ \phi_1 \cap \phi_2 \mid \phi_1 \in \Phi_{\textup{product}} \wedge \phi_2 \in \sum_{a \in \bar{A}_1 \times A_2 } \Phi_{R}^{a} \}$
\For{$\phi_{\textup{product}} \in \Phi_{\textup{product}} $} \Comment{Value backup}
\State Take one state $(\hat{s}_1,\hat{s}_E) \in \phi_{\textup{product}}$
\State $\alpha^{\star}(\phi_{\textup{product}}) \leftarrow  f_{\overline{p}_1^{\star}, \overline{\alpha}^{\star}} (\hat{s}_1,\hat{s}_E)$
\EndFor
\State \textbf{return:} $(\Phi_{\textup{product}}, \alpha^{\star})$
\end{algorithmic}
\end{algorithm}

\begin{lema}[ISPP backup]\label{lema:ISPP-backup}
    The FCP $\Phi_{\textup{product}}$ returned by \algoref{alg:ISPP-backup} is a constant-FCP of $\phi$ for $\alpha^{\star}$ and the region-by-region backup for $\alpha^*$ satisfies the line 4 of \algoref{alg:point-based-update-belief}.
\end{lema}
\begin{proof}
     For the PWC $\alpha$-functions in the input of \algoref{alg:ISPP-backup}, if $\Phi_{a_1, s_1'}$ is an FCP of $S$ for $\alpha^{a_1,s_1'}$, then let $\Phi = \sum_{a_1 \in \bar{A}_1, s_1' \in S_1} \Phi_{a_1, s_1'}$, i.e., $\Phi$ is the smallest refinement of these FCPs.
    
    According to \aspref{asp:transitions-rewards}, there exists a preimage-FCP of $\Phi$ for each joint action $a$. Through the image, split, pre-image and product operations of \algoref{alg:ISPP-backup}, all the states in any region $\phi' \in \Phi_{\textup{product}}$ reach the same regions of $\Phi$. Since each $\alpha$-function $\alpha^{a_1, s_1'}$ is constant over each region in $\Phi$, all states in $\phi'$ have the same backup value from $\alpha^{a_1, s_1'}$ for $a_1 \in \bar{A}_1$ and $s_1' \in S_1$. This implies that $\Phi_{\textup{product}}$ is the product of the preimage-FCPs of $\Phi$ for all $a \in \bar{A}_1 \times A_2$. Since the value backup in line 4 of \algoref{alg:point-based-update-belief} 
    is used for each region in $\Phi_{\textup{product}}$ and the image is from the region $\phi$, then $\Phi_{\textup{product}}$ is a constant-FCP of $\phi$ for $\alpha^{\star}$, and thus the value backup in line 4 of \algoref{alg:point-based-update-belief} for $\alpha^{\star}$ is achieved by considering the regions of $\Phi_{\textup{product}}$. \qed
\end{proof}



\section{Linear Programs}\label{sec:appendix-dual-LPs}

\noindent
In this section we provide the linear programs (LPs) and their dual versions,
omitted for space reasons in the main paper,
in particular for the stage games $[TV_{\mathit{lb}}^{\Gamma}](s_1, b_1)$ and $[TV_{\mathit{ub}}^{\Upsilon}](s_1, b_1)$.
Consider a particle-based belief $(s_1, b_1)$ represented by $\{ (s_E^i, \kappa_i) \}_{i=1}^{n_b}$.

\startpara{Stage game over the lower bound} Using \lemaref{lema:LP-minimax-P-PWLC} and its extended version \lemaref{lema:LP-minimax-P-PWLC-extended}, the LP \eqref{eq:LP-minimax-2} for the stage game $[TV_{\mathit{lb}}^{\Gamma}](s_1, b_1)$ is simplified to the LP over the variables:
\begin{itemize}
\item
$(v_{s_E^i})_{i=1}^{n_b}$;
\item
$(\lambda^{a_1, s_1'}_{\alpha} )_{(a_1,s_1') \in A_1 \times S_1,\alpha \in \Gamma}$;
\item
$(p^{a_1})_{a_1 \in A_1}$;
\end{itemize}
and is given by {\rm maximise} $\sum_{i=1}^{n_b} \kappa_i v_{s_E^i}$ subject to:
    \begin{align}
         v_{s_E^i} & \; \leq \mbox{$\sum_{a_1 \in A_1}$} p^{a_1} r((s_1, s_E^i),(a_1,a_2))  + \beta \mbox{$\sum_{(a_1, s_1') \in A_1 \times S_1, s_E' \in S_E}$}  \nonumber \\
         & \qquad \delta((s_1, s_E^i), (a_1, a_2))(s_1', s_E') \mbox{$\sum_{\alpha \in \Gamma} \lambda_{\alpha}^{a_1, s_1'}$} \alpha (s_1', s_E') \nonumber \\
         \lambda^{a_1, s_1'}_{\alpha} & \; \ge 0 \nonumber \\
         p^{a_1} & \; = \mbox{$\sum_{\alpha \in \Gamma}$} \lambda_{\alpha}^{a_1, s_1'}  \nonumber \\
         \mbox{$ \sum_{a_1 \in A_1} $} p^{a_1} & \; = 1 \label{eq:lower-bound-LP-particle}
    \end{align} 
    for all $1\leq i \leq n_b$, $a_2 \in A_2$, $(a_1,s_1') \in A_1 \times S_1$ and $\alpha \in \Gamma$. 

\vskip9pt\noindent
The dual of LP problem \eqref{eq:lower-bound-LP-particle} is over the variables:
\begin{itemize}
\item
$v$;
\item
$(v_{a_1, s_1'} )_{(a_1, s_1') \in A_1 \times S_1}$;
\item
$(p_{a_2}^{s_1, s_E^i})_{a_2 \in A_2, 1 \leq i \leq n_b}$;
\end{itemize}
and is given by \mbox{\rm minimise} $v$ subject to:
    \begin{align}
         v & \;  \ge \mbox{$\sum_{i=1}^{n_b}\sum_{a_2 \in A_2}$} p_{a_2}^{s_1, s_E^i} r((s_1, s_E^i),(a_1,a_2)) + \beta \mbox{$\sum_{s_1' \in S_1}$} v_{a_1, s_1'} \nonumber\\
         v_{a_1, s_1'} & \; \ge \mbox{$\sum_{i=1}^{n_b}\sum_{a_2 \in A_2}$} p_{a_2}^{s_1, s_E^i} \delta((s_1, s_E^i), (a_1, a_2))(s_1', s_E') \alpha (s_1', s_E')  \nonumber \\
     \mbox{$\sum_{a_2 \in A_2} $}  p_{a_2}^{s_1, s_E^i} & \; = \kappa_i  \label{eq:lower-bound-dual-LP-particle}
    \end{align} 
    for all $a_1 \in A_1$, $(a_1, s_1') \in A_1 \times S_1$, $\alpha \in \Gamma$ and $1\leq i \leq n_b$.

\vskip9pt\noindent
By solving \eqref{eq:lower-bound-LP-particle} and \eqref{eq:lower-bound-dual-LP-particle}, we obtain the minimax strategy profile in the stage game $[TV_{\mathit{lb}}^{\Gamma}](s_1, b_1)$: $ u_1^{\mathit{lb}}(a_1) = p^{\star a_1}$ for $a_1 \in A_1$ and $u_2^{\mathit{lb}}(a_2 \mid s_1, s_E^i) = p_{a_2}^{\star s_1, s_E^i} / \kappa_i$ for $1 \leq i \leq n_b$ and $a_2 \in A_2$.


\vskip9pt\noindent
\startpara{Stage game over the upper bound} The LP for the stage game $[TV_{\mathit{ub}}^{\Upsilon}](s_1, b_1)$ is over the variables:
\begin{itemize}
\item
$v$;
\item
$( c_{s_{\scale{.75}{E}}'}^{a_1, s_1'} )_{(a_1, s_1') \in A_1 \times S_1 , s_{\scale{.75}{E}}' \in S_{\scale{.75}{E}}^{a_{\scale{.75}{1}},s_{\scale{.75}{1}}'}}$;
\item
$(\lambda_k^{a_1, s_1'} )_{(a_1, s_1') \in A_1 \times S_1, k \in I_{s_{\scale{.75}{1}}'}}$;
\item
$(p_{a_2}^{s_1, s_E^i} )_{1 \leq i \leq n_b, a_2 \in A_2}$
\end{itemize}
and is given by \mbox{\rm minimise} $v$ subject to:
    \begin{align}
         v & \; \ge  \mbox{$\sum_{i=1}^{n_b}\sum_{a_2 \in A_2}$}  \kappa_i p_{a_2}^{s_1, s_E^i} r((s_1, s_E^i),(a_1, a_2))   \nonumber \\
         & \qquad +  \beta \mbox{$\sum_{s_1' \in S_1} \sum_{k \in I_{s'_{\scale{.75}{1}}}}$} \lambda_k^{a_1, s_1'} y_k  + \mbox{$\frac{1}{2}$} \beta (U - L) \mbox{$\sum\nolimits_{s_1' \in S_1} \sum_{s_E' \in S_E^{a_{\scale{.75}{1}},s_{\scale{.75}{1}}'}}$} c_{s_{\scale{.75}{E}}'}^{a_1, s_1'}  \nonumber \\
         c_{s_{\scale{.75}{E}}'}^{a_1, s_1'} & \; \ge \Big| \mbox{$\sum_{i=1}^{n_b}\sum_{a_2 \in A_2}$}   \kappa_i p_{a_2}^{s_1, s_E^i} \delta((s_1, s_E^i), (a_1,a_2))(s_1', s_E') \nonumber  \\ 
         & \qquad - \mbox{$\sum\nolimits_{k \in I_{s'_{\scale{.75}{1}}}}$}  \lambda_k^{a_1, s_1'} P(s_E'; b_1^k) \Big| \nonumber \\
         \mbox{$\sum_{k \in I_{s_{\scale{.75}{1}}'}}$} \lambda_k^{a_1, s_1'} & \;  = \mbox{$\sum_{i=1}^{n_b}\sum_{a_2 \in A_2, s_E' \in S_E}$} \kappa_i  p_{a_2}^{s_1, s_E^i} \delta((s_1, s_E^i), (a_1, a_2))(s_1', s_E')  \nonumber \\
         \lambda_k^{a_1, s_1'} & \; \ge 0  \nonumber \\
         p_{a_2}^{s_1, s_E^i} & \; \ge 0  \nonumber \\
         \mbox{$\sum_{a_2 \in A_2}$} p_{a_2}^{s_1, s_E^i} & \; = 1
        \label{eq:LP-minimax-upper}
    \end{align} 
for all $a_1 \in A_1$, $(a_1, s_1') \in A_1 \times S_1$ and $s_E' \in S_E^{a_1, s_1'}$, $k \in I_{s'_{\scale{.75}{1}}}$, $a_2 \in A_2$ and $1 \leq i \leq n_b$  where $S_E^{a_{\scale{.75}{1}},s_{\scale{.75}{1}}'} = \{ s_E' \in S_E \mid  \mbox{$\sum_{a_2 \in A_2}$} b_1^{s_1, a_1, a_2, s_1'}(s_E') + \sum_{k \in I_{s_{\scale{.75}{1}}'}} b_1^k(s_E') > 0 \}$.

\vskip9pt\noindent
The dual of LP problem \eqref{eq:LP-minimax-upper} is the following LP problem over the variables:
\begin{itemize}
\item
$(v_{s_E^i})_{1 \leq i \leq n_b}$;
\item
$(v_{a_1, s_1'})_{(a_1, s_1') \in A_1 \times S_1}$;
\item
$(p^{a_1})_{a_1 \in A_1}$;
\item
$(d_{a_1, s_1', s_E'} )_{(a_1, s_1') \in A_1 \times S_1 , s_E' \in S_E^{a_1, s_1'}}$;
\item
$( e_{a_1, s_1', s_E'} )_{(a_1, s_1') \in A_1 \times S_1 , s_E' \in S_E^{a_1, s_1'}}$;
\end{itemize}
and is given by {\rm maximise} $\sum_{i=1}^{n_b} \kappa_i v_{s_E^i}$ subject to:
\begin{align}
         v_{s_E^i} & \; \leq  \mbox{$\sum_{a_1 \in A_1}$} p^{a_1} r((s_1, s_E^i),(a_1, a_2)) +  \beta \mbox{$\sum_{a_1 \in A_1, s_1' \in S_1, s_E' \in S_E^{a_1, s_1'}}$}  \nonumber \\
         & \qquad \delta((s_1, s_E^i), (a_1,a_2))(s_1', s_E')   (v_{a_1, s_1'} + d_{a_1, s_1', s_E'} - e_{a_1, s_1', s_E'})  \nonumber \\
         v_{a_1, s_1'} & \; \leq y_k p^{a_1} - \mbox{$\sum_{s_E' \in S_E^{a_1, s_1'}}$} (d_{a_1, s_1', s_E'} - e_{a_1, s_1', s_E'}) P(s_E'; b_1^k) \nonumber  \\ 
         d_{a_1, s_1', s_E'} - e_{a_1, s_1', s_E'}  &\; \leq  \mbox{$\frac{1}{2}$} (U - L)  \nonumber \\
         d_{a_1, s_1', s_E'} &\; \ge 0 \nonumber \\
         e_{a_1, s_1', s_E'} &\; \ge 0 \nonumber \\
         p^{a_1} &\; \ge 0 \nonumber \\
          \mbox{$\sum_{a_1 \in A_1}$} p^{a_1} &\; = 1
        \label{eq:LP-minimax-upper-dual}
    \end{align}
    for all $a_2 \in A_2$ and $1 \leq i \leq n_b$, $(a_1, s_1') \in A_1 \times S_1$, $k \in I_{s'_{\scale{.75}{1}}}$ and $s_E' \in S_E^{a_1, s_1'}$ 
    where  $S_E^{a_1, s_1'} = \{ s_E' \in S_E \mid  \exists 1 \leq i \leq n_b .\, \exists a_2 \in A_2 . \, \delta((s_1, s_E^i), (a_1,a_2))(s_1', s_E') > 0 \}$. 

    \vskip9pt\noindent
    By solving \eqref{eq:LP-minimax-upper} and \eqref{eq:LP-minimax-upper-dual}, we obtain the minimax strategy profile in stage game $[TV_{\mathit{ub}}^{\Upsilon}](s_1, b_1)$: $ u_1^{\mathit{ub}}(a_1) = p^{\star a_1}$ for $a_1 \in A_1$ and $u_2^{\mathit{ub}}(a_2 \mid s_1, s_E^i) = p_{a_2}^{\star s_1, s_E^i}$ for $1 \leq i \leq n_b$ and $a_2 \in A_2$.
    

\section{Proofs of Main Results}\label{sec:appendix-proofs}

\noindent
We provide here the proofs of the results from the main paper.


\begin{proof}[\textbf{Proof of \thomref{thom:convexity-continuity}}]
    Given $s_1 \in S_1$, we first prove that $V^{\star}(s_1, \cdot)$ is convex and continuous. For any $b_1 \in \mathbb{P}(S_E)$, since $V^{\star}(s_1, b_1)$ is the lower value of $Y$, then $V^{\star}(s_1, b_1) = \sup\nolimits_{\sigma_1 \in \Sigma_1}\inf\nolimits_{\sigma_2 \in \Sigma_2} \mathbb{E}_{(s_1,b_1)}^{\sigma_1,\sigma_2}[Y]$. We define a payoff function $V_{\sigma_1} : \mathbb{P}(S_E) \to \mathbb{R}$ to be the objective of the sup optimisation in the lower value such that for $b_1 \in \mathbb{P}(S_E)$ we have $V_{\sigma_1}(s_1, b_1) = \inf\nolimits_{\sigma_2 \in \Sigma_2} \mathbb{E}_{(s_1,b_1)}^{\sigma_1,\sigma_2}[Y]$. Note that the value $V_{\sigma_1}(s_1, b_1)$ is the expected reward of $\sigma_1$ against the best-response strategy $\sigma_2$, from the initial belief $(s_1, b_1)$. Since $\agent_2$ can observe the true initial state $(s_1, s_E)$ where $s_E$ is sampled from $b_1$, and thus can play a state-wise best-response to each initial state $(s_1, s_E)$, the value $V_{\sigma_1}(s_1, b_1)$ can be rewritten as:
    \begin{equation} \label{eq:V-sigma1}
        V_{\sigma_1}(s_1, b_1) = \mbox{$\int\nolimits_{s_E \in S_E}$} b_1(s_E) \big( \inf\nolimits_{\sigma_2 \in \Sigma_2} \mathbb{E}_{(s_1,s_E)}^{\sigma_1,\sigma_2}[Y] \big) \textup{d} s_E \,.
    \end{equation}
    Thus, $V_{\sigma_1}(s_1, \cdot)$ is a linear function in the belief $b_1 \in \mathbb{P}(S_E)$. Since $V^{\star}(s_1, b_1) = \sup\nolimits_{\sigma_1 \in \Sigma_1} V_{\sigma_1}(s_1, b_1)$ and any point-wise supremum of linear functions is convex and continuous (it follows from the convexity and continuity in the discrete case, see \cite[Proposition 5.9]{KH-BB-VK-CK:23}), we can conclude that $V^{\star}(s_1, \cdot)$ is convex and continuous. 

    Regarding the inequality in \thomref{thom:convexity-continuity}, for any $b_1, b_1' \in  \mathbb{P}(S_E)$, we have: 
    \begin{equation}\label{lip1-eqn}
    \mbox{$ \int_{s_E \in S_E^{s_1}} $} b_1(s_E) \textup{d} s_E = \mbox{$ \int_{s_E \in S_E^{s_1}} $} b_1'(s_E) \textup{d} s_E = 1 \, .
    \end{equation}
    Now, letting $S_E^{>} = \{ s_E \in S_E^{s_1} \mid b_1(s_E) - b_1'(s_E) > 0 \}$ and $S_E^{\leq} = \{ s_E \in S_E^{s_1} \mid b_1(s_E) - b_1'(s_E) \leq 0 \}$, rearranging \eqref{lip1-eqn} and using the fact that $S_E^{>} \cup S_E^{\leq} =S_E^{s_1}$ it follows that:
    \begin{align}
    \mbox{$\int_{s_E \in S_E^{\leq}} $} (b_1(s_E) - b_1'(s_E)) \textup{d} s_E = - \mbox{$ \int_{s_E \in S_E^{>}} $} (b_1(s_E) - b_1'(s_E)) \textup{d} s_E \nonumber
    \end{align}
    from which we have:
    \begin{align}
        &\mbox{$\int_{s_E \in S_E^{s_1}}$} |b_1(s_E) - b_1'(s_E)| \textup{d} s_E = \mbox{$ \int_{s_E \in S_E^{>} \cup S_E^{\leq}} $} |b_1(s_E) - b_1'(s_E)| \textup{d} s_E \nonumber  \\
        & = \mbox{$ \int_{s_E \in S_E^{>}} $} (b_1(s_E) - b_1'(s_E)) \textup{d} s_E - \mbox{$ \int_{s_E \in S_E^{\leq}} $} (b_1(s_E) - b_1'(s_E)) \textup{d} s_E  \nonumber \\
        & = 2 \mbox{$ \int_{s_E \in S_E^{>}} $} (b_1(s_E) - b_1'(s_E)) \textup{d} s_E    \label{eq:equvalent-difference}
    \end{align}
    and thus, using \eqref{eq:equvalent-difference} and \cite[Theorem 2]{RY-GS-GN-DP-MK:23}, the inequality in \thomref{thom:convexity-continuity} holds. \qed
\end{proof}


\begin{theorem}[Operator equivalence and fixed point, extended version of \thomref{thom:operator-equivalence}]
If $\Gamma \subseteq \mathbb{F}(S)$ and $V(s_1, b_1) = \sup_{\alpha \in \Gamma} \langle \alpha, (s_1,b _1) \rangle $ for $(s_1, b_1) \in S_B$, then the minimax operator $T$ and maxsup operator $T_\Gamma$ are equivalent, i.e., for $(s_1, b_1) \in S_B$ we have:
    \begin{align}
         \lefteqn{[TV](s_1, b_1) =  \mbox{$\max_{u_1\in \mathbb{P}(A_1)}$} \mbox{$\min_{u_2\in \mathbb{P}(A_2 \mid S)} $} \mathbb{E}_{(s_1,b_1),u_1,u_2} [r(s,a)]} \nonumber \\
        & \quad + \beta \mbox{$\sum_{a_1 \in A_1} \sum_{s_1'\in S_1}$} P( (a_1, s_1') \mid (s_1, b_1), u_1, u_2 ) V(s'_1, b_1^{s_1,a_1, u_2, s_1'})  \label{eq:operator-equivalence-1} \\
         & = \mbox{$\min_{u_2 \in \mathbb{P}(A_2 \mid S)} $}  \mbox{$ \max_{u_1\in \mathbb{P}(A_1)} $} \mathbb{E}_{(s_1,b_1),u_1,u_2} [r(s,a)] \nonumber \\
        &  \quad + \beta \mbox{$\sum_{a_1 \in A_1} \sum_{s_1'\in S_1}$} P( (a_1, s_1') \mid (s_1, b_1), u_1, u_2 ) V(s'_1, b_1^{s_1,a_1, u_2, s_1'})    \label{eq:operator-equivalence-2}
        \\
         & = \mbox{$ \max_{u_1\in \mathbb{P}(A_1)} $}  \mbox{$ \sup_{\overline{\alpha} \in \Gamma^{A_1 \times S_1} } $} \langle f_{u_1,\overline{\alpha}}, (s_1, b_1) \rangle \label{eq:operator-equivalence-3} \\
         & = [T_\Gamma V](s_1,b_1) \, . \nonumber
    \end{align}
    Moreover, the unique fixed point of $T$ and $T_\Gamma$ is $V^{\star}$. 
\end{theorem}

\begin{proof}
    Considering any $V \in \mathbb{F}(S_B)$ and $\Gamma\subseteq \mathbb{F}(S)$ such that:
    \begin{equation}\label{v-assumption}
    V(s_1, b_1) = \sup\nolimits_{\alpha \in \Gamma} \langle \alpha, (s_1,b _1) \rangle \quad \mbox{for all $(s_1, b_1) \in S_B$}.
    \end{equation}
    {\bf Operator equivalence.} We first show that the operators $T$ and $T_\Gamma$ are equivalent.  
    We first define a payoff function $J: \mathbb{P}(A_1) \times \mathbb{P}(A_2 \mid S) \to \mathbb{R}$ to be the objective of the maximin and minimax optimisation in \eqref{eq:operator-equivalence-1} and \eqref{eq:operator-equivalence-2} such that for $u_1 \in \mathbb{P}(A_1)$ and $u_2 \in \mathbb{P}(A_2 \mid S)$:
    \begin{align} 
        & J(u_1, u_2) \; = \; \mathbb{E}_{(s_1,b_1),u_1,u_2} [r(s,a)] + \nonumber \\
        & \qquad \qquad \beta \mbox{$\sum_{a_1 \in A_1} \sum_{s_1'\in S_1}$} P(a_1, s_1' \mid (s_1, b_1), u_1, u_2 ) V(s'_1, b_1^{s_1,a_1, u_2, s_1'}) \, . \label{eq:J-function}
    \end{align}
    Now for any belief $(s_1, b_1) \in S_B$ such that $s_1 = (\loc_1, \per_1)$, action $a_1 \in A_1$, agent state $s_1' \in S_1$ and stage strategy $u_2 \in \mathbb{P}(A_2 \mid S)$, letting $P_1 \triangleq P(s_1' \mid (s_1, b_1), a_1, u_2)$ by \eqref{v-assumption} we have:
    \begin{align}
    \lefteqn{V(s'_1, b_1^{s_1,a_1, u_2, s_1'}) \; = \; \mbox{$ \sup_{\alpha \in \Gamma} $} \langle \alpha, (s'_1, b_1^{s_1,a_1, u_2, s_1'} ) \rangle}  \nonumber \\
   & \!\!\!\!= \mbox{$ \sup_{\alpha \in \Gamma} $} \mbox{$\int_{s_E' \in S_E}$} \alpha (s_1', s_E') b_1^{s_1, a_1, u_2, s_1'}(s_E') \textup{d}s_E'  & \mbox{rearranging}  \nonumber  \\
    & \!\!\!\!= \mbox{$ \sup_{\alpha \in \Gamma}  $} \mbox{$\int_{s_E' \in S_E}$} \alpha(s_1', s_E') \frac{P((s_1', s_E') \mid (s_1,b_1), a_1, u_2 )}{P(s_1' \mid (s_1,b_1), a_1, u_2 )} \textup{d} s_E' & \mbox{by \eqref{eq:belief-update-u2}} \nonumber \\
   & \!\!\!\!=\frac{1}{P_1} \mbox{$ \sup_{\alpha \in \Gamma}  $} \mbox{$\int_{s_E' \in S_E}$} \alpha (s_1', s_E') P((s_1', s_E') \mid (s_1,b_1), a_1, u_2 )  \textup{d} s_E' \!\! & \mbox{rearranging} \nonumber  \\
   & \!\!\!\!\lefteqn{\;=\frac{1}{P_1} \mbox{$ \sup_{\alpha \in \Gamma} $} \big( \mbox{$\int_{s_E' \in S_E}$} \alpha(s_1', s_E') \mbox{$\int_{s_E' \in S_E^{s_1'} , s_E \in S_E}$} b_1(s_E) \mbox{$\sum_{a_2 \in A_2}$}  u_2(a_2 \mid s_1, s_E)} \nonumber \\
   & \qquad \qquad \cdot  \delta((s_1, s_E), (a_1, a_2))(s_1',s_E')  \textup{d} s_E \big) \textup{d} s_E' & \mbox{by \eqref{eq:belief-update-closed2}}\nonumber \\
   & \!\!\!\!\lefteqn{\;=\frac{1}{P_1} \mbox{$ \sup_{\alpha \in \Gamma} $} \big(  \mbox{$\int_{s_E \in S_E}$} \big( \mbox{$\int_{s_E' \in S_E^{s_1'}}$} \alpha(s_1', s_E')\mbox{$\sum_{a_2 \in A_2}$}  u_2(a_2 \mid s_1, s_E)} \nonumber \\
   & \qquad \qquad \cdot  \delta((s_1, s_E), (a_1, a_2))(s_1',s_E')   \textup{d} s_E' \big)  b_1(s_E) \textup{d} s_E & \mbox{rearranging.} \label{eq:V-at-next-belief}
    \end{align}
    Next, for any $\alpha \in \mathbb{F}(S)$, $s_1' \in S_1$, $a_1 \in A_1$ and $u_2 \in \mathbb{P}(A_2 \mid S)$ we let $\alpha^{a_1,u_2,s_1'} : S \rightarrow \mathbb{R}$ be the function where for any $s = ( (\loc_1, \per_1), s_E) \in S$: 
    \begin{align}
        \alpha^{a_1,u_2,s_1'} (s) &  = \mbox{$\int_{s_E' \in S_E^{s_1'}}$} \alpha(s_1', s_E')\mbox{$\sum_{a_2}$}  u_2(a_2 \mid s)  \delta(s, (a_1, a_2))(s_1',s_E')   \textup{d} s_E' \nonumber \\ 
        & = \mbox{$\sum_{a_2}$}  u_2(a_2 \mid s)  \mbox{$\sum_{s_E' \in S_E}$}\delta(s, (a_1, a_2))(s_1', s_E')  \alpha(s_1', s_E')  \label{eq:new-beta-function}
    \end{align}
    and the summation in $s_E'$ is due to the finite branching of $\delta$.
    Combining \eqref{eq:V-at-next-belief} and \eqref{eq:new-beta-function}  we have:
    \begin{align}
        V(s'_1, b_1^{s_1,a_1, u_2, s_1'}) & =  \frac{1}{P_1} \sup\nolimits_{\alpha \in \Gamma} \mbox{$\int_{s_E \in S_E}$} \alpha^{a_1,u_2,s_1'} (s_1, s_E) b_1(s_E)  \textup{d} s_E \nonumber  \\
        & = \frac{1}{P(s_1' \mid (s_1, b_1), a_1, u_2)} \sup\nolimits_{\alpha \in \Gamma} \langle \alpha^{a_1,u_2,s_1'}, (s_1, b_1) \rangle  \label{eq:V-function-u2}
    \end{align}
    by definition of $P_1$.
    Substituting \eqref{eq:V-function-u2} into \eqref{eq:J-function}, the payoff function $J(u_1, u_2)$ equals:
    \begin{align}
        & \mathbb{E}_{(s_1,b_1),u_1,u_2} [r(s,a)] + \beta \mbox{$\sum_{a_1, s_1'} $} u_1(a_1) P(s_1' \mid (s_1, b_1), a_1, u_2 ) V(s'_1, b_1^{s_1,a_1, u_2, s_1'}) \nonumber \\
        & = \mathbb{E}_{(s_1,b_1),u_1,u_2} [r(s,a)] + \beta \mbox{$\sum_{a_1, s_1'}$} u_1(a_1) \mbox{$ \sup_{\alpha \in \Gamma} $}\langle \alpha^{a_1,u_2,s_1'}, (s_1, b_1) \rangle \,. \label{eq:J-function-2}
    \end{align}
    We next show that the von Neumann's Minimax Theorem \cite{JvN:28} applies to the game $\sem{\csg}$ with the payoff function $J$ and strategy spaces $\mathbb{P}(A_1)$ and $\mathbb{P}(A_2 \mid S)$. This theorem requires that $\mathbb{P}(A_1)$ and $\mathbb{P}(A_2\mid S)$ are compact convex sets (which is straightforward to show) and that $J$ is a continuous function that is concave-convex, i.e., 
    \begin{itemize}
    \item
    $J(\cdot, u_2)$ is concave for fixed $u_2 \in \mathbb{P}(A_2\mid S)$;
    \item
    $J(u_1, \cdot)$ is convex for fixed $u_1 \in \mathbb{P}(A_1)$. 
    \end{itemize}
    By \defiref{defi:minimax-operator} the expectation $\mathbb{E}_{(s_1,b_1),u_1,u_2} [r(s,a)]$ can be rewritten as:
    \[
    \mbox{$\sum_{a_1}$} u_1(a_1) \mbox{$\int_{s_E \in S_E} $} b_1(s_E) \mbox{$\sum_{a_2}$} u_2(a_2 \mid s_1, s_E ) r((s_1, s_E),(a_1,a_2)) \textup{d} s_E 
    \]
    and thus, $\mathbb{E}_{(s_1,b_1),u_1,u_2} [r(s,a)]$ is bilinear in $u_1$ and $u_2$, and thus concave in $\mathbb{P}(A_1)$ and convex in $\mathbb{P}(A_2 \mid S)$.
    
    We next show that $u_1(a_1) \sup_{\alpha \in \Gamma} \langle \alpha^{a_1,u_2,s_1'}, (s_1, b_1) \rangle$ is continuous and concave in $u_1 \in \mathbb{P}(A_1)$ and convex in $u_2 \in \mathbb{P}(A_2 \mid S)$. The continuity and concavity in $u_1 \in \mathbb{P}(A_1)$ follows directly as it is linear in $u_1 \in \mathbb{P}(A_1)$. For $u_2 \in \mathbb{P}(A_2 \mid S)$, we consider the function $f(u_2) = \langle \alpha^{a_1,u_2,s_1'}, (s_1, b_1) \rangle$. By \eqref{eq:new-beta-function} we have that $f(u_2)$ equals:
    \begin{equation*}
         \mbox{$\int_{s_E \in S_E}$} \mbox{$\sum_{a_2}$}  u_2(a_2 \mid s_1, s_E)  \mbox{$\sum_{s_E' \in S_E}$} \delta((s_1, s_E), (a_1, a_2))(s_1', s_E') \alpha(s_1', s_E') b_1(s_E) \textup{d} s_E
    \end{equation*}
    and therefore $f(u_2)$ is linear in $u_2$. Since the point-wise maximum over linear functions is continuous and convex, it follows that $\sup_{\alpha \in \Gamma} f(u_2)$ is continuous and convex in $u_2 \in \mathbb{P}(A_2 \mid S)$, and hence $u_1(a_1) \sup_{\alpha \in \Gamma} \langle \alpha^{a_1,u_2,s_1'}, (s_1, b_1) \rangle$ is continuous and convex in $u_2 \in \mathbb{P}(A_2 \mid S)$. According to von Neumann's Minimax theorem:
    \[ \begin{array}{c}
    \max_{u_1 \in \mathbb{P}(A_1)} \min_{u_2 \in \mathbb{P}(A_2 \mid S)} J(u_1, u_2) = \min_{u_2 \in \mathbb{P}(A_2 \mid S)}  \max_{u_1 \in \mathbb{P}(A_1)} J(u_1, u_2)
    \end{array} \]
    and hence the equality between \eqref{eq:operator-equivalence-1} and \eqref{eq:operator-equivalence-2} holds.

    Next we prove the equality of \eqref{eq:operator-equivalence-1} and \eqref{eq:operator-equivalence-3}.
    Letting $\textup{Conv}(\Gamma)$ be the convex hull of $\Gamma$,
    recall that $\Gamma^{A_1 \times S_1}$ is the set of vectors of functions in $\textup{Conv}(\Gamma)$ indexed by the elements of $A_1 \times S_1$. The function $J(u_1, u_2)$ in \eqref{eq:J-function-2} can be rewritten as follows:
    \begin{align} 
        & \mbox{$ \sup_{\overline{\alpha} \in \Gamma^{A_1 \times S_1} } $}  \Big( \mathbb{E}_{(s_1,b_1),u_1,u_2} [r(s,a)] \nonumber \\
        & \hspace{3cm} + \beta \mbox{$\sum_{a_1 \in A_1, s_1' \in S_1} $} u_1(a_1) \langle \alpha^{a_1,u_2,s_1'}, (s_1, b_1) \rangle \Big) \label{eq:J-function-3}
    \end{align}
    where $\bar{\alpha} = ( \alpha^{a_1,s_1'} )_{a_1 \in A_1, s_1' \in S_1}$, %
    and given $u_1$ and $u_2$, the supremum over $\Gamma$ only depends on $a_1$ and $s_1'$ and using the same arguments as \cite[Proposition 4.11]{KH-BB-VK-CK:23} we have: 
    \[ \begin{array}{c}
    \sup_{\alpha \in \Gamma} \langle \alpha, (s_1, b_1) \rangle = \sup_{\alpha \in \textup{Conv}(\Gamma)} \langle \alpha, (s_1, b_1) \rangle
    \end{array} \]
    for $(s_1, b_1) \in S_B$.  We next define the game with strategy spaces $\Gamma^{A_1 \times S_1}$ and $\mathbb{P}(A_2 \mid S)$ and payoff function $J_{u_1}: \Gamma^{A_1 \times S_1} \times \mathbb{P}(A_2 \mid S) \to \mathbb{R}$ where for $\overline{\alpha} \in \Gamma^{A_1 \times S_1}$ and $u_2 \in \mathbb{P}(A_2 \mid S)$:
    \begin{align}
    \lefteqn{J_{u_1}(\overline{\alpha}, u_2)  = \mathbb{E}_{(s_1,b_1),u_1,u_2} [r(s,a)] + \beta \mbox{$\sum_{a_1 \in A_1, s_1' \in S_1} $} u_1(a_1) \langle \alpha^{a_1,u_2,s_1'}, (s_1, b_1) \rangle} \nonumber \\
    & \lefteqn{= \mathbb{E}_{(s_1,b_1),u_1,u_2} [r(s,a)] + \beta \mbox{$\sum_{a_1 \in A_1, s_1' \in S_1}$} u_1(a_1)  \mbox{$\int_{s_E \in S_E }$}  \big( \mbox{$\sum_{a_2 \in A_2}$}  u_2(a_2 \mid s_1, s_E)}  \nonumber \\
    &  \cdot \mbox{$\sum_{s_E' \in S_E}$} \delta((s_1, s_E), (a_1, a_2))(s_1', s_E') \alpha^{a_1, s_1'}(s_1', s_E') \big)  b_1(s_E) \textup{d} s_E & \mbox{by \eqref{eq:new-beta-function}.} \label{eq:J-u-1-expression}
    \end{align}
    Substituting \eqref{eq:J-function-3} and \eqref{eq:J-u-1-expression} into \eqref{eq:operator-equivalence-1} we have:
    \begin{align}
        & \mbox{$\max_{u_1 \in \mathbb{P}(A_1)}$} \mbox{$ \min_{u_2 \in \mathbb{P}(A_2 \mid S)} $} J(u_1, u_2) \nonumber \\
        & \qquad = \mbox{$ \max_{u_1 \in \mathbb{P}(A_1)}$} \mbox{$\min_{u_2 \in \mathbb{P}(A_2 \mid S)}$}  \mbox{$\sup_{\overline{\alpha} \in \Gamma^{A_1 \times S_1} }$} J_{u_1}(\overline{\alpha}, u_2) \,. \label{eq:max-min-max}
    \end{align}
    We next show that Sion's Minimax Theorem \cite{MS:58} applies to the game with strategy spaces $\Gamma^{A_1 \times S_1}$ and $\mathbb{P}(A_2 \mid S)$ and payoff function $J_{u_1}$. Sion's Minimax Theorem requires that:
    \begin{itemize}
    \item
    $\Gamma^{A_1 \times S_1}$ is convex;
    \item
    $\mathbb{P}(A_2 \mid S)$ is compact and convex;
    \item
    for any $u_2 \in \mathbb{P}(A_2 \mid S)$ the function $J_{u_1}(\cdot, u_2) : \Gamma^{A_1 \times S_1} \rightarrow \mathbb{R}$ is  upper semicontinuous and quasi-concave;
    \item
    for any $\overline{\alpha} \in \Gamma^{A_1 \times S_1}$ the function $J_{u_1}(\overline{\alpha}, \cdot) : \mathbb{P}(A_2 \mid S) \rightarrow \mathbb{R}$ is lower semicontinuous and quasi-convex. 
    \end{itemize}
    The first properties clearly hold and the second to follow from \eqref{eq:J-u-1-expression} which demonstrate that both $J_{u_1}(\cdot, u_2)$ and $J_{u_1}(\overline{\alpha}, \cdot)$ are linear. 
    
    Therefore using Sion's Minimax Theorem, we have:
    \[ \begin{array}{c}
    \min_{u_2 \in \mathbb{P}(A_2 \mid S)}  \sup_{\overline{\alpha} \in \Gamma^{A_1 \times S_1} } J_{u_1}(\overline{\alpha}, u_2) = \sup_{\overline{\alpha} \in \Gamma^{A_1 \times S_1} } \min_{u_2 \in \mathbb{P}(A_2 \mid S)} J_{u_1}(\overline{\alpha}, u_2) \, 
    \end{array} \]
    and combining with \eqref{eq:max-min-max} it follows that $\max\nolimits_{u_1 \in \mathbb{P}(A_1)} \min\nolimits_{u_2 \in \mathbb{P}(A_2 \mid S)} J(u_1, u_2)$ equals:
    \begin{align} 
    & \lefteqn{\max\nolimits_{u_1 \in \mathbb{P}(A_1)} \sup\nolimits_{\overline{\alpha} \in \Gamma^{A_1 \times S_1} } \min\nolimits_{u_2 \in \mathbb{P}(A_2 \mid S)}  J_{u_1}(\overline{\alpha}, u_2)} \nonumber \\
    & \lefteqn{= \max\nolimits_{u_1 \in \mathbb{P}(A_1)} \sup\nolimits_{\overline{\alpha} \in \Gamma^{A_1 \times S_1} } \min\nolimits_{u_2 \in \mathbb{P}(A_2 \mid S)}  \mbox{$\int_{s_E \in S_E }$} \mbox{$\sum_{a_2}$} u_2(a_2 \mid s_1, s_E) \mbox{$\sum_{a_1}$} u_1(a_1)} \nonumber \\
    & \lefteqn{\quad \cdot r((s_1, s_E),(a_1,a_2)) b_1(s_E) \textup{d} s_E +   \beta \mbox{$\int_{s_E \in S_E }$}  \Big( \mbox{$\sum_{a_2}$} u_2(a_2 \mid s_1, s_E) \mbox{$\sum_{a_1, s_1'}$} u_1(a_1)}  \nonumber \\
    & \lefteqn{\quad \cdot \mbox{$\sum_{s_E' \in S_E}$} \delta((s_1, s_E), (a_1, a_2))(s_1', s_E') \alpha(s_1', s_E') \Big)  b_1(s_E) \textup{d} s_E}  & \mbox{by \eqref{eq:J-u-1-expression}} \nonumber \\
    & \lefteqn{= \max\nolimits_{u_1 \in \mathbb{P}(A_1)} \sup\nolimits_{\overline{\alpha} \in \Gamma^{A_1 \times S_1} } \mbox{$\int_{s_E \in S_E}$} \min\nolimits_{u_2 \in \mathbb{P}(A_2 \mid S)} \mbox{$\sum_{a_2}$} u_2(a_2 \mid s_1, s_E)} \nonumber \\
    & \lefteqn{\quad  \Big( \mbox{$\sum_{a_1}$} u_1(a_1) r((s_1, s_E),(a_1,a_2)) + \beta \mbox{$\sum_{a_1, s_1'}$} u_1(a_1)} \nonumber \\
    & \quad  \mbox{$\sum_{s_E' \in S_E}$} \delta((s_1, s_E), (a_1, a_2))(s_1', s_E') \alpha(s_1', s_E') \Big)  b_1(s_E) \textup{d} s_E & \mbox{rearranging} \nonumber \\
    & \lefteqn{= \max\nolimits_{u_1 \in \mathbb{P}(A_1)} \sup\nolimits_{\overline{\alpha} \in \Gamma^{A_1 \times S_1} } \mbox{$\int_{s_E \in S_E}$} \min\nolimits_{a_2 \in A_2} \Big( \mbox{$\sum_{a_1}$} u_1(a_1) r((s_1, s_E),(a_1,a_2))} \nonumber \\
    & \lefteqn{\quad + \beta \mbox{$\sum_{a_1, s_1'}$} u_1(a_1) \mbox{$\sum_{s_E' \in S_E}$} \delta((s_1, s_E), (a_1, a_2))(s_1', s_E') \alpha(s_1', s_E') \Big)  b_1(s_E) \textup{d} s_E} \nonumber \\
     && \hspace{-1cm} \mbox{since $\agent_2$ is fully informed} \nonumber \\
     & \lefteqn{= \max\nolimits_{u_1 \in \mathbb{P}(A_1)} \sup\nolimits_{\overline{\alpha} \in \Gamma^{A_1 \times S_1} } \mbox{$\int_{s_E \in S_E}$} \big( \min\nolimits_{a_2 \in A_2} f_{u_1,\overline{\alpha}, a_2}(s_1, s_E) \big)   b_1(s_E) \textup{d} s_E} \nonumber \\
     && \hspace{-1cm}  \mbox{by \eqref{eq:f-u1-alpha-a2}}
     \nonumber \\
    & \lefteqn{= \max\nolimits_{u_1\in \mathbb{P}(A_1)} \sup\nolimits_{\overline{\alpha} \in \Gamma^{A_1 \times S_1} } \langle f_{u_1,\overline{\alpha}}, (s_1, b_1) \rangle}  & \mbox{by \defiref{defi:simplified-operator}} \nonumber 
    \end{align}
     which demonstrates that \eqref{eq:operator-equivalence-1} and \eqref{eq:operator-equivalence-3} are equal, i.e., $T$ and $T_\Gamma$ are equivalent.

     \startpara{Fixed point} To show the unique fixed point of $T$ and $T_\Gamma$ is $V^{\star}$.
     We first prove that $V^{\star}$ is a fixed point of the operator $T$, i.e., $V^{\star} = [T V^{\star}]$. According to the proof of \thomref{thom:convexity-continuity}, for $(s_1, b_1) \in S_B$ the value function $V^\star$ can be represented by:
    \begin{align*}
        V^{\star}(s_1, b_1) & = \sup\nolimits_{\sigma_1 \in \Sigma_1} V_{\sigma_1}(s_1, b_1) \\
        & = \sup\nolimits_{\sigma_1 \in \Sigma_1} \mbox{$\int\nolimits_{s_E \in S_E}$} b_1(s_E) \big( \inf\nolimits_{\sigma_2 \in \Sigma_2} \mathbb{E}_{(s_1,s_E)}^{\sigma_1,\sigma_2}[Y] \big) \textup{d} s_E  & \mbox{by \eqref{eq:V-sigma1}} \\
        & = \sup\nolimits_{\sigma_1 \in \Sigma_1} \langle \inf\nolimits_{\sigma_2 \in \Sigma_2} \mathbb{E}_{(s_1,s_E)}^{\sigma_1,\sigma_2}[Y], (s_1, b_1) \rangle   & \\
        & = \sup\nolimits_{\alpha \in \Gamma } \langle \alpha, (s_1, b_1) \rangle
    \end{align*}
    where $\Gamma \triangleq \{ \inf\nolimits_{\sigma_2 \in \Sigma_2} \mathbb{E}_{(s_1,s_E)}^{\sigma_1,\sigma_2}[Y] \mid \sigma_1 \in \Sigma_1 \}$. According to the operator equivalence above, we have:
    \begin{equation}
        [TV^{\star}] (s_1, b_1) = \mbox{$\max_{u_1\in \mathbb{P}(A_1)} $} \mbox{$\sup_{\overline{\alpha} \in \Gamma^{A_1 \times S_1} }$} \langle f_{u_1,\overline{\alpha}}, (s_1, b_1) \rangle
    \end{equation}
    for all $(s_1, b_1) \in S_B$, where $\Gamma^{A_1 \times S_1} \triangleq \{ \{ \alpha^{a_1,s_1'} \}_{a_1 \in A_1, s_1' \in S_1} \mid \alpha^{a_1,s_1'} \in \textup{Conv}(\Gamma) \}$ and $\Gamma$ is given above. Now, by following the same argument as in the proof of \cite[Lemma 6.7]{KH-BB-VK-CK:23}, we can show that $V^{\star}(s_1, b_1) = [TV^{\star}] (s_1, b_1)$ for all $(s_1, b_1) \in S_B$, i.e., $V^{\star} = [T V^{\star}]$.

    Next we demonstrate that the operator $T$ is a contraction mapping on the space $\mathbb{F}(S_B)$ with respect to the supremum norm $\| J \| = \sup_{(s_1, b_1) \in S_B} |J(s_1, b_1)|$. Therefore consider any $J_1, J_2 \in \mathbb{F}(S_B)$ and for any belief $(s_1, b_1) \in S_B$, let $(u_1^{1\star}, u_2^{1\star})$ and $(u_1^{2\star}, u_2^{2\star})$ be the minimax strategy profiles in the stage games $[TJ_1](s_1, b_1)$ and $[TJ_2](s_1, b_1)$, respectively. Also, let $\bar{J}_1 (u_1, u_2)$ and $\bar{J}_2(u_1, u_2)$ be the values of state $(s_1, b_1)$ of the stage game under the strategy pair $(u_1, u_2) \in \mathbb{P}(A_1) \times \mathbb{P}(A_2 \mid S)$ when computing the backup values in \eqref{eq:J-function} for $J_1$ and $J_2$, respectively. Without loss of generality, we assume $[TJ_1](s_1, b_1) \leq [TJ_2](s_1, b_1)$, and thus since $(u_1^{1\star}, u_2^{1\star})$ is minimax strategy profile for  $[TJ_1](s_1, b_1)$:
    \begin{align}
        \bar{J}_1(u_1^{2\star}, u_2^{1\star}) & \leq \bar{J}_1 (u_1^{1\star}, u_2^{1\star})  \nonumber \\
        & = [TJ_1](s_1, b_1)  & \mbox{by definition of $\bar{J}_1$} \nonumber \\
        & \leq [TJ_2](s_1, b_1) & \mbox{without loss of generality} \nonumber \\
        & = \bar{J}_2(u_1^{2\star}, u_2^{2\star}) & \mbox{by definition of $\bar{J}_2$}  \nonumber \\
        & \leq \bar{J}_2(u_1^{2\star}, u_2^{1\star}) & \mbox{since $(u_1^{2\star}, u_2^{2\star})$ is minimax strategy}. \label{eq:bound-TJ2-TJ1}
    \end{align}
    Now using \eqref{eq:bound-TJ2-TJ1} for any $(s_1, b_1) \in S_B$ we have
    \begin{align}
        & \lefteqn{| [TJ_2](s_1, b_1) - [TJ_1](s_1, b_1) | \leq \bar{J}_2(u_1^{2\star}, u_2^{1\star}) - \bar{J}_1(u_1^{2\star}, u_2^{1\star})} \nonumber \\
        & \lefteqn{= \beta  \mbox{$\sum_{a_1, s_1'}$} P(a_1, s_1' \mid (s_1, b_1), u_1^{2\star}, u_2^{1\star} ) \big( J_2(s'_1, b_1^{s_1,a_1, u_2^{1\star}, s_1'}) - J_1(s'_1, b_1^{s_1,a_1, u_2^{1\star}, s_1'}) \big)} \nonumber \\
        && \mbox{by \eqref{eq:J-function}} \nonumber \\
        & \leq \beta  \mbox{$\sum_{a_1, s_1'}$} P(a_1, s_1' \mid (s_1, b_1), u_1^{2\star}, u_2^{1\star} ) \| J_2 - J_1 \| & \mbox{by definition of $\| \cdot \|$} \nonumber  \\
        & = \beta \| J_2 - J_1 \| & \hspace{-4cm} \mbox{since $P(\cdot \mid (s_1, b_1), u_1^{2\star}, u_2^{1\star} )$ is a distribution.} \label{contratction-eq}
    \end{align}
    Now by definition of the supremum norm:
    \begin{align*}
        \| [TJ_2] - [TJ_1] \| & = \sup\nolimits_{(s_1, b_1) \in S_B} |[TJ_2](s_1, b_1) - [TJ_1](s_1, b_1)| \\
        & \leq \sup\nolimits_{(s_1, b_1) \in S_B} \beta \| J_2 - J_1 \| & \mbox{by \eqref{contratction-eq}} \\
        & = \beta \| J_2 - J_1 \| & \mbox{rearranging}
    \end{align*}
    and hence, since $\beta \in (0,1)$, we have that $T$ is a contraction mapping. Thus, the fact that the value function $V^{\star}$ is the unique fixed point of $T$ now follows directly from Banach's fixed point theorem. \qed
\end{proof}


%
\begin{lema}[PWC function]\label{lema:PWC-intermidiate-func}
    For any $a \in A$, $s_1' \in S_1$ and $\alpha \in \mathbb{F}_C(S)$, if $\alpha^{a, s_1'} : S \to \mathbb{R}$ is the function where for any $s \in S$:
    \begin{equation*}
        \alpha^{a, s_1'}(s) = \mbox{$\sum_{ (s_1', s_E') \in \Theta_{s}^a }$} \delta(s, a) (s_1', s_E') \alpha (s_1', s_E') 
    \end{equation*}
    then $\alpha^{a, s_1'}$ is PWC.
\end{lema}
\begin{proof}
Let $a = (a_1, a_2)$. Since $\alpha$ is PWC, there exists an FCP $\Phi$ of $S$ such that $\alpha$ is constant in each region of $\Phi$. 
According to \aspref{asp:transitions-rewards},
there exists a pre-image FCP $\Phi'$ of $\Phi + \Phi_{P}$ for joint action $a$, where $\Phi_{P}$ is the perception FCP for $\agent_1$. Consider any region $\phi' \in \Phi'$ and let $\phi$ be any region of $\Phi + \Phi_{P}$ such that $\Theta_s^{a} \cap \phi \neq \emptyset$ for all $s \in \phi'$. Since $\Phi_{P}$ is the perception FCP for $\agent_1$, there exists $s_1' \in S_1$ such that if $s' \in \phi$, then $s' = (s_1', s_E')$ for some $s_E' \in S_E$ and let $\phi_E = \{ s_E \in S_E \mid (s_1', s_E) \in \phi\}$. If $s, \tilde{s} \in \phi'$ such that $s = (s_1, s_E)$ and $\tilde{s} = (\tilde{s}_1, \tilde{s}_E)$, then using \aspref{asp:transitions-rewards} we have $\sum_{s' \in \Theta_s^{a} \cap \phi} \delta(s, a)(s') = \sum_{\tilde{s}' \in \Theta_{\tilde{s}}^{a} \cap \phi} \delta(\tilde{s}, a)(\tilde{s}')$ and $s_1 = \tilde{s}_1$. Now combining this fact with \defiref{semantics-def}, it follows that:
\begin{align*}
\mbox{$\sum\nolimits_{(s_1', s_E') \in \Theta_{s}^{a} , s_E' \in \phi_E}$} \delta(s, a)(s_1', s_E') = \mbox{$\sum\nolimits_{(s_1', \tilde{s}_E') \in \Theta_{\tilde{s}}^{a}, \tilde{s}_E' \in \phi_E}$} \delta(\tilde{s}, a)(s_1', \tilde{s}_E') \, .
\end{align*}
Since $\alpha^{a_1, s_1'}(s_1', s_E') = \alpha^{a_1, s_1'}(s_1', \tilde{s}_E')$ for any $(s_1', s_E'), (s_1', \tilde{s}_E') \in \phi$ and $S_E^{s_1'} = \{ s_E' \in S_E \mid \obs_1(\loc_1', s_E') = \per_1' \}$ is equal to $\{ \phi_E \mid \phi \in \Phi^{s_1'} \}$ for some finite set of regions $\Phi^{s_1'} \subseteq \Phi + \Phi_{P}$, it follows that 
\begin{align*}
\lefteqn{\mbox{$\sum\nolimits_{(s_1', s_E') \in \Theta_{s}^{a}, s_E' \in S_E^{s_1'}}$} \delta(s, a)(s_1', s_E') \alpha^{a_1, s_1'}(s_1', s_E')} \\
& = \mbox{$\sum\nolimits_{(s_1', \tilde{s}_E') \in \Theta_{\tilde{s}}^{a}, \tilde{s}_E' \in S_E^{s_1'} }$} \delta(\tilde{s}, a)(s_1', \tilde{s}_E') \alpha^{a_1, s_1'}(s_1', \tilde{s}_E')
\end{align*}
and therefore $\alpha^{a, s_1'}(s) = \alpha^{a, s_1'}(\tilde{s}) $, implying that $\alpha^{a, s_1'}$ is constant in each region of $\Phi'$. \qed
\end{proof}

\begin{lema}[LP for minimax and P-PWLC, extended \lemaref{lema:LP-minimax-P-PWLC}]\label{lema:LP-minimax-P-PWLC-extended}
    If $V \in \mathbb{F}(S_B)$ is P-PWLC with PWC $\alpha$-functions $\Gamma$, for any $(s_1, b_1) \in S_B$, $[TV](s_1, b_1)$ is given by the LP over the real-valued variables $( v_{\phi})_{\phi \in \Phi_\Gamma}$, $(\lambda_{\alpha}^{a_1, s_1'})_{(a_1,s_1') \in A_1 \times S_1, \alpha \in \Gamma}$ and $(p^{a_1})_{a_1 \in A_1}:$
    \begin{align}
        & \hspace{-0.2cm} \mbox{\rm maximise} \; \;
        \mbox{$\sum_{\phi \in \Phi_{\Gamma}}$} v_{\phi} \mbox{$\int_{(s_1, s_E) \in \phi}$} b_1(s_E) \textup{d} s_E  \;\; \mbox{\rm subject to}
        \nonumber  \\
        &  v_{\phi}  \leq \mbox{$\sum_{a_1 \in A_1}$} p^{a_1} r((s_1, s_E),(a_1,a_2)) + \beta \mbox{$\sum_{a_1, s_1', s_E'}$} \delta((s_1, s_E), (a_1, a_2))(s_1', s_E') \nonumber\\
         & \quad  \cdot \mbox{$\sum_{\alpha \in \Gamma} \lambda_{\alpha}^{a_1, s_1'}$} \alpha (s_1', s_E')\nonumber\\ 
         & \lambda^{a_1, s_1'}_{\alpha} \ge 0,  \nonumber \\
         & p^{a_1} {=} \mbox{$\sum_{\alpha \in \Gamma} $} \lambda_{\alpha}^{a_1, s_1'}\nonumber\\
         & \mbox{$ \sum_{a_1 \in A_1}$} p^{a_1}  {=} 1 \label{eq:LP-minimax-2}
    \end{align} 
    for all $\phi \in \Phi_\Gamma$, $a_2 \in A_2$, $(a_1,s_1') \in A_1 \times S_1$ and $\alpha \in \Gamma$ where $s_E \in \phi$. Moreover, if $(\overline{v}^\star, \overline{\lambda}_1^\star , \overline{p}_1^\star)$ is the optimal solution to the LP \eqref{eq:LP-minimax-2}, then the maximiser of the maxsup operator in \defiref{defi:simplified-operator} is $(\overline{p}_1^{\star}, \overline{\alpha}^{\star})$, where $\overline{\alpha}^{\star} \in \Gamma^{A_1 \times S_1}$ is such that for $(a_1, s_1') \in A_1 \times S_1$, if $a_1 \in A_1$ and $p^{\star a_1} > 0$, then $\alpha^{\star a_1, s_1'} = \mbox{$\sum_{\alpha \in \Gamma}$} (\lambda^{\star a_1, s_1'}_{\alpha} / p^{\star a_1}) \alpha$
and $\alpha^{\star a_1, s_1'}(s) = L$ for all $s \in S$ otherwise.
\end{lema}
\begin{proof}
    Since $V$ is P-PWLC, then according to Definitions \ref{defi:simplified-operator} and \ref{defi:PWLC} and \thomref{thom:operator-equivalence}: 
    \begin{align}
        & [TV](s_1, b_1) = \mbox{$\max_{u_1\in \mathbb{P}(A_1)} $} \mbox{$\sup_{\overline{\alpha} \in \Gamma^{A_1 \times S_1} } $} \langle f_{u_1,\overline{\alpha}}, (s_1, b_1) \rangle \nonumber \\
        & = \mbox{$\max_{u_1 \in \mathbb{P}(A_1)}$} \mbox{$\sup_{\overline{\alpha} \in \Gamma^{A_1 \times S_1} }$} \mbox{$\int_{s_E \in S_E}$} \big( \mbox{$\min_{a_2}$} f_{u_1, \overline{\alpha}, a_2}(s_1, s_E) \big) b_1(s_E) \textup{d} s_E 
        \label{eq:TV-LB-max-sup-min}
    \end{align}
    which can be formulated as the following optimization problem:
    \begin{align*}
        [TV](s_1, b_1) = & \max\nolimits_{u_1 \in \mathbb{P}(A_1),\overline{\alpha} \in \Gamma^{A_1 \times S_1}, \overline{v}} \mbox{$\sum_{\phi \in \Phi_{\Gamma}}$} v_{\phi} \mbox{$\int_{(s_1,s_E) \in \phi}$}  b_1(s_E) \textup{d} s_E \\
        & \textup{ subject to } v_{\phi} \leq f_{u_1, \overline{\alpha}, a_2}(s_1, s_E) \quad \mbox{for all $\phi \in \Phi_{\Gamma}$ and $a_2 \in A_2$}
    \end{align*}
    where $\overline{v} = ( v_{\phi} )_{\phi \in \Phi_{\Gamma}}$, $f_{u_1, \overline{\alpha}, a_2}$ is constant over $\phi$ and $(s_1, s_E) \in \phi$. 
    Using \eqref{eq:f-u1-alpha-a2}, the constraint $v_{\phi} \leq f_{u_1, \overline{\alpha}, a_2}(s_1, s_E)$ can be written as:
    \begin{align*}
        v_{\phi} \leq &  \mbox{$\sum_{a_1 \in A_1}$} u_1(a_1) r((s_1, s_E),(a_1,a_2)) \\
        & + \beta \mbox{$\sum_{(a_1, s_1') \in A_1 \times S_1, s_E' \in S_E}$} u_1(a_1) \delta((s_1, s_E), (a_1, a_2))(s_1', s_E') \alpha^{a_1, s_1'} (s_1', s_E') .
    \end{align*}
    Since $\alpha^{a_1, s_1'} \in \textup{Conv}(\Gamma)$, 
    we have $\alpha^{a_1, s_1'} = \sum_{\alpha \in \Gamma} \lambda_{\alpha}^{a_1, s_1'} \alpha$ for some vector of real-values $(\lambda_{\alpha}^{a_1, s_1'} )_{(a_1,s_1) \in A_1 \times S_1}$ such that $\sum_{\alpha \in \Gamma} \lambda_{\alpha}^{a_1, s_1'} = 1$, and therefore:
    \begin{align*}
        v_{\phi} &  \leq \mbox{$\sum_{a_1 \in A_1}$} u_1(a_1) r((s_1, s_E),(a_1,a_2))  + \beta \mbox{$\sum_{(a_1, s_1') \in A_1 \times S_1, s_E' \in S_E}$} \\
        & \quad \;\; u_1(a_1) \delta((s_1, s_E), (a_1, a_2))(s_1', s_E') \mbox{$\sum_{\alpha \in \Gamma} \lambda_{\alpha}^{a_1, s_1'}$} \alpha (s_1', s_E') \\
        & =  \mbox{$\sum_{a_1 \in A_1}$} p_{a_1} r((s_1, s_E),(a_1,a_2))  + \\
        & \quad \;\; + \beta \mbox{$\sum_{(a_1, s_1') \in A_1 \times S_1, s_E' \in S_E}$} \delta((s_1, s_E), (a_1, a_2))(s_1', s_E') \mbox{$\sum_{\alpha \in \Gamma} \lambda_{\alpha}^{a_1, s_1'}$} \alpha (s_1', s_E') 
    \end{align*}
    where $p_{a_1} = u_1(a_1)$ for all $a_1 \in A_1$ and in the equality we scale $\lambda_{\alpha}^{a_1, s_1'} = p_{a_1} \lambda_{\alpha}^{a_1, s_1'}$ for all $a_1 \in A_1$, $s_1' \in S_1$ and $\alpha \in \Gamma$, which gives the constraints:
    \begin{align*}
        \lambda_{\alpha}^{a_1, s_1'} & \; \ge 0 \\ 
        p_{a_1} & = \;\mbox{$\sum_{\alpha \in \Gamma} $} \lambda_{\alpha}^{a_1, s_1'} \\
        \mbox{$\sum_{a_1 \in A_1} $} p_{a_1} & = \; 1 
    \end{align*}
    and hence the fact we can solve the LP problem \eqref{eq:LP-minimax-2} to compute $[TV](s_1, b_1)$ follows directly. \qed
\end{proof}


\begin{proof}[\textbf{Proof of \thomref{thom:P-PWLC-closure}}]
\textbf{P-PWLC closure.} Consider the LP in \lemaref{lema:LP-minimax-P-PWLC}, i.e., \eqref{eq:LP-minimax-2} in the extended \lemaref{lema:LP-minimax-P-PWLC-extended}, which computes the minimax or maxsup backup $[TV](s_1, b_1)$ when $V$ is P-PWLC. The polytope of feasible solutions of the LP defined by the constraints is independent of the environment belief $b_1$, because $b_1$ only appears in the objective. Therefore, the set $Q_{s_1}$ of vertices of this polytope is also independent of $b_1$. For each $b_1 \in \mathbb{P}(S_E)$, the optimal value of an LP representing $[TV](s_1, b_1)$ can be found with the vertices $Q_{s_1}$, as the objective is linear in $\hat{V}$ for any given $b_1$. There is a finite number of vertices $q \in Q_{s_1}$, and each vertex $q \in Q_{s_1}$ corresponds to some assignment of variables $u_1^q$ and  $\overline{\alpha}^q$ ($u_1^q$ and $\overline{\alpha}^q$ are computed by \eqref{eq:LP-minimax-2}). 
Since $Q_{s_1}$ is finite, then letting $Q = \{ q \in Q_{s_1} \mid s_1 \in S_1\}$, which is finite, we have:
\begin{align*}
    [TV](s_1, b_1) = \mbox{$\max_{q \in Q}$} \langle f_{u_1^q, \overline{\alpha}^q}, (s_1, b_1) \rangle \, .
\end{align*}
Moreover, since $f_{u_1, \overline{\alpha}, a_2}$ is PWC for any $u_1 \in \mathbb{P}(A_1), \overline{\alpha} \in \Gamma^{A_1 \times S_1} $ and $a_2 \in A_2$, then it follows from \defiref{defi:simplified-operator}, the function $f_{u_1^p, \overline{\alpha}^p}$ is PWC. This implies that $[TV] \in \mathbb{F}(S_B)$ and P-PWLC.  

\startpara{Convergence} Using the fixed point in \thomref{thom:operator-equivalence}, the conclusion directly follows from Banach's fixed point theorem and the fact we have proved in \thomref{thom:P-PWLC-closure} that if $V \in \mathbb{F}(S_B)$ and P-PWLC, so is $[TV]$. \qed 
\end{proof}


\begin{proof}[\textbf{Proof of \lemaref{lema:new-pwc-alpha}}]
    By following the proof of \thomref{thom:P-PWLC-closure} and how $\overline{p}_1^{\star}$ and $\overline{\alpha}^{\star}$ are constructed, we can easily verify that in \algoref{alg:point-based-update-belief} $\alpha^{\star}$ is a PWC $\alpha$-function satisfying \eqref{eq:update-lb-condition}. 

    For $V_1, V_2 \in \mathbb{F}(S_B)$, we use the notation $V_1 \leq V_2$ if $V_1(\hat{s}_1, \hat{b}_1) \leq V_2(\hat{s}_1, \hat{b}_1)$ for all $(\hat{s}_1, \hat{b}_1) \in S_B$. 
    Since $\Gamma' = \Gamma \cup \{\alpha^{\star}\}$, 
    then it follows from \defiref{defi:PWLC} 
    that $V_{\mathit{lb}}^{\Gamma} \leq V_{\mathit{lb}}^{\Gamma'}$.

In \algoref{alg:point-based-update-belief}, if the backup at line 4 is executed, then the maxsup operator is applied to some states in $\phi$ which may result in non-optimal minimax backup for other states in $\phi$, and if the backup at line 5 is executed, $\alpha^{\star}$ is assigned the lower bound $L$ over $\phi$. Therefore we have for any $(\hat{s}_1,\hat{b}_1) \in S_B$:
\begin{align}
\langle \alpha^{\star}, (\hat{s}_1,\hat{b}_1) \rangle 
& \leq [TV_{\mathit{lb}}^{\Gamma}](\hat{s}_1,\hat{b}_1) \nonumber \\
& \leq [TV^{\star}](\hat{s}_1,\hat{b}_1) & \mbox{since $V_{\mathit{lb}}^{\Gamma} \leq V^{\star}$} \nonumber \\
& = V^{\star} (\hat{s}_1,\hat{b}_1) & \mbox{by \thomref{thom:operator-equivalence}.} \label{eq:lb-alpha-v-1}
    \end{align}
Combining this inequality with $V_{\mathit{lb}}^{\Gamma} \leq V^{\star}$, we have $V_{\mathit{lb}}^{\Gamma'} \leq V^{\star}$ as required. \qed
\end{proof}


\begin{proof}[\textbf{Proof of \lemaref{lema:upper-bound-update}}]
     Combining \thomref{thom:convexity-continuity}, \eqref{eq:new-ub} and \eqref{eq:K-UB-condition-2}, 
     the conclusion can be obtained by following the argument in the proof of \cite[Lemma 4]{RY-GS-GN-DP-MK:23} for NS-POMDPs. \qed
\end{proof}

\noindent
The following lemma is required to prove the convergence of the algorithm.

\begin{lema}[Finite terminal belief points]\label{lema:finite-terminal-beliefs}
   For any $t \ge 0$, if  $\Psi_t \subseteq S_B$ of belief points where the trials performed by the procedure $\mathit{Explore}$ of \algoref{alg:NS-HSVI} terminated at exploration depth $t$, then $\Psi_t$ is a finite set.
\end{lema}

\begin{proof}
    Consider any $t \ge 0$ and suppose that $\Psi_t \subseteq S_B$ is the set of belief points where the trials performed by the procedure $\mathit{Explore}$ terminated at depth $t$.
    In order to prove that $\Psi_t$ is a finite set, we first need to show the following continuity of the lower and upper bounds. Using the same argument in the proof \thomref{thom:convexity-continuity}, we can prove that the lower bound $V_{\mathit{lb}}^{\Gamma}$ also has the continuity property of \thomref{thom:convexity-continuity}, i.e., for any $(s_1, b_1), (s_1, b_1') \in S_B$:
    \begin{equation}\label{eq:lower-bound-continuity}
        |V_{\mathit{lb}}^{\Gamma}(s_1, b_1) - V_{\mathit{lb}}^{\Gamma}(s_1, b_1')| \leq K(b_1, b_1') \, .
    \end{equation}
    We still consider two beliefs $(s_1, b_1), (s_1, b_1') \in S_B$. Let $(\lambda_i^{\star\prime} )_{i \in I_{s_{\scale{.75}{1}}}}$ be the solution for $V_{\mathit{ub}}^{\Upsilon}(s_1, b_1')$ in \eqref{eq:new-ub}, i.e., 
    \begin{equation}\label{terminal1-eqn}
    V_{\mathit{ub}}^{\Upsilon}(s_1, b_1') = \mbox{$\sum\nolimits_{i \in I_{s_{\scale{.75}{1}}}}$}\lambda_i^{\star\prime} y_i + K_{\mathit{ub}}(b_1', \mbox{$\sum\nolimits_{i \in I_{s_{\scale{.75}{1}}}}$} \lambda_i^{\star\prime} b_1^i) \, . 
    \end{equation}
    Now since $(\lambda_i^{\star\prime} )_{i \in I_{s_{\scale{.75}{1}}}}$ satisfies the constraints in \eqref{eq:new-ub} for $I_{s_{\scale{.75}{1}}}$, it follows that:
    \begin{align*}
        \lefteqn{V_{\mathit{ub}}^{\Upsilon}(s_1, b_1) 
        \leq \mbox{$\sum\nolimits_{i \in I_{s_{\scale{.75}{1}}}}$}\lambda_i^{\star\prime} y_i + K_{\mathit{ub}}(b_1, \mbox{$\sum\nolimits_{i \in I_{s_{\scale{.75}{1}}}}$} \lambda_i^{\star} b_1^i)} \\
        & = \big(V_{\mathit{ub}}^{\Upsilon}(s_1, b_1') - K_{\mathit{ub}}(b_1', \mbox{$\sum\nolimits_{i \in I_{s_{\scale{.75}{1}}}}$} \lambda_i^{\star} b_1^i) \big) + K_{\mathit{ub}}(b_1, \mbox{$\sum\nolimits_{i \in I_{s_{\scale{.75}{1}}}}$} \lambda_i^{\star\prime} b_1^i) & \mbox{by \eqref{terminal1-eqn}} \\
       & = V_{\mathit{ub}}^{\Upsilon}(s_1, b_1') + \big( K_{\mathit{ub}}(b_1, \mbox{$\sum\nolimits_{i \in I_{s_{\scale{.75}{1}}}}$} \lambda_i^{\star\prime} b_1^i) - K_{\mathit{ub}}(b_1', \mbox{$\sum\nolimits_{i \in I_{s_{\scale{.75}{1}}}}$} \lambda_i^{\star} b_1^i) \big)  & \mbox{rearranging} \\
        & \leq  V_{\mathit{ub}}^{\Upsilon}(s_1, b_1')  + K_{\mathit{ub}}(b_1, b_1') & \mbox{by \eqref{eq:K-UB-condition-2}.}
    \end{align*}
    Using similar steps we can also show that:
    \[
    V_{\mathit{ub}}^{\Upsilon}(s_1, b_1') \leq V_{\mathit{ub}}^{\Upsilon}(s_1, b_1) +  K_{\mathit{ub}}(b_1, b_1')
    \]
    and hence:
    \begin{equation}\label{eq:upper-bound-continuity}
    |V_{\mathit{ub}}^{\Upsilon}(s_1, b_1) - V_{\mathit{ub}}^{\Upsilon}(s_1, b_1')| \leq K_{\mathit{ub}}(b_1, b_1') \,.
    \end{equation}
    Let a belief point $(s_1^t, b_1^t) \in \Psi_t$. Since the procedure $\mathit{Explore}$ terminates at $(s_1^t, b_1^t)$ with exploration depth $t$, then the action-observation pair $(\hat{a}_1, \hat{s}_1)$ computed by \eqref{eq:max-action-observation} (from line 7 of \algoref{alg:NS-HSVI}) satisfies
    \[
     P(\hat{a}_1, \hat{s}_1 \mid (s_1^t, b_1^t), u_1^{\mathit{ub}}, u_2^{\mathit{lb}}) \mathit{excess}_{t+1}(\hat{s}_1, b_1^{s_1^t,\hat{a}_1,u_2^{\mathit{lb}}, \hat{s}_1}) \leq 0 \, .
    \]
    Thus, for any $(a_1, s_1') \in A_1 \times S_1$, if $ P(a_1, s_1' \mid (s_1^t, b_1^t), u_1^{\mathit{ub}}, u_2^{\mathit{lb}}) > 0$, then we have $\mathit{excess}_{t+1}(s_1', b_1^{s_1^t,a_1,u_2^{\mathit{lb}}, s_1'}) \leq 0$, i.e.,
    \begin{equation}\label{eq:rho-t-1-bound}
        V_{\mathit{ub}}^{\Upsilon}(s_1', b_1^{s_1^t,a_1,u_2^{\mathit{lb}}, s_1'}) -  V_{\mathit{lb}}^{\Gamma}(s_1', b_1^{s_1^t,a_1,u_2^{\mathit{lb}}, s_1'}) \leq \rho(t+1) \,.
    \end{equation}
    Let $(u_1^{\mathit{lb}}, u_2^{\mathit{lb}})$ and $(u_1^{\mathit{ub}}, u_2^{\mathit{ub}})$ be the minimax strategy profiles in stage games $[TV_{\mathit{lb}}^{\Gamma}](s_1^t, b_1^t)$ and $[TV_{\mathit{ub}}^{\Upsilon}](s_1^t, b_1^t)$, respectively. Then, we denote by $J^{\mathit{lb}}(u_1, u_2)$ and $J^{\mathit{ub}}(u_1, u_2)$ the value of the stage game at $(s_1^t, b_1^t)$ under the strategy pair $(u_1, u_2) \in \mathbb{P}(A_1) \times \mathbb{P}(A_2 \mid S)$ when computing the backup values in \eqref{eq:J-function} via $V_{\mathit{lb}}^{\Gamma}$ and $V_{\mathit{ub}}^{\Upsilon}$, respectively. Thus, since $(u_1^{\mathit{lb}}, u_2^{\mathit{lb}})$ is a minimax strategy profile:
    \begin{align}
        J^{\mathit{lb}}(u_1^{\mathit{ub}}, u_2^{\mathit{lb}}) 
        & \leq J^{\mathit{lb}}(u_1^{\mathit{lb}}, u_2^{\mathit{lb}}) \nonumber \\
        & = [TV_{\mathit{lb}}^{\Gamma}](s_1^t, b_1^t) & \mbox{by definition of $J^{\mathit{lb}}$}\nonumber \\
        & \leq [TV_{\mathit{ub}}^{\Upsilon}](s_1^t, b_1^t) & \mbox{by Lemmas~\ref{lema:new-pwc-alpha} and \ref{lema:upper-bound-update}} \nonumber \\
        & = J^{\mathit{ub}}(u_1^{\mathit{ub}}, u_2^{\mathit{ub}}) & \mbox{by definition of $J^{\mathit{ub}}$} \nonumber \\
        & \leq J^{\mathit{ub}}(u_1^{\mathit{ub}}, u_2^{\mathit{lb}}) & \mbox{$(u_1^{\mathit{ub}}, u_2^{\mathit{ub}})$ is a minimax strategy profile.} \label{eq:equilibrium-based-bounds}
    \end{align}
    Now using
     \eqref{eq:equilibrium-based-bounds} we have: 
    \begin{align}
        \lefteqn{[TV_{\mathit{ub}}^{\Upsilon}](s_1^t, b_1^t) - [TV_{\mathit{lb}}^{\Gamma}](s_1^t, b_1^t) \leq J^{\mathit{ub}}(u_1^{\mathit{ub}}, u_2^{\mathit{lb}}) - J^{\mathit{lb}}(u_1^{\mathit{ub}}, u_2^{\mathit{lb}})}  \nonumber \\
        & = \beta \mbox{$\sum_{a_1, s_1' \in A_1 \times S_1}$} P(a_1, s_1' \mid (s_1^t, b_1^t), u_1^{\mathit{ub}}, u_2^{\mathit{lb}} ) \nonumber \\
        & \quad ( V_{\mathit{ub}}^{\Gamma}(s'_1, b_1^{s_1^t,a_1, u_2^{\mathit{lb}}, s_1'}) - V_{\mathit{lb}}^{\Gamma}(s'_1, b_1^{s_1^t,a_1, u_2^{\mathit{lb}}, s_1'})) & \mbox{by \eqref{eq:J-function}} \nonumber \\
        & \lefteqn{\;\leq \beta \mbox{$\sum_{a_1, s_1'  \in A_1 \times S_1}$} P(a_1, s_1' \mid (s_1^t, b_1^t), u_1^{\mathit{ub}}, u_2^{\mathit{lb}} ) \rho(t+1)} & \mbox{by \eqref{eq:rho-t-1-bound}}  \nonumber \\
        & = \beta \rho(t + 1) & \mbox{since $P$ is a distribution.} \label{eq:difference-bound-beta-rho}
    \end{align}
    Substituting \eqref{eq:difference-bound-beta-rho} into the excess gap $\mathit{excess}_{t}(s_1^t, b_1^t)$ we have that the excess gap after performing the point-based update at $(s_1^t, b_1^t)$ in line 10 of \algoref{alg:NS-HSVI}:
    \begin{align*}
         \mathit{excess}_{t}(s_1^t, b_1^t) &  \leq \beta \rho(t + 1) - \rho(t) \\
         & = \rho(t) - 2(U - L)\bar{\varepsilon} - \rho(t) & \mbox{by definition of $\rho(t + 1)$} \\
         & = - 2(U - L) \bar{\varepsilon} & \mbox{rearranging.}
    \end{align*}
    Due to the continuity \eqref{eq:lower-bound-continuity} and \eqref{eq:upper-bound-continuity}, for any $(s_1, b_1), (s_1, b_1') \in S_B$, we have
    \begin{equation}\label{excess2-eq}
        V_{\mathit{ub}}^{\Upsilon}(s_1, b_1) -  V_{\mathit{lb}}^{\Gamma}(s_1, b_1) \leq V_{\mathit{ub}}^{\Upsilon}(s_1, b_1') - V_{\mathit{lb}}^{\Gamma}(s_1, b_1') + 2K_{\mathit{ub}}(b_1, b_1') \,.
    \end{equation}
    Now, for every belief $(s_1^t, b_1) \in S_B$ satisfying $K_{\mathit{ub}}(b_1, b_1^t) \leq (U-L)\bar{\varepsilon}$, substituting \eqref{excess2-eq} into the excess gap $\mathit{excess}_{t}(s_1^t, b_1)$:
    \begin{align*}
        \lefteqn{\mathit{excess}_{t}(s_1^t, b_1)  \leq V_{\mathit{ub}}^{\Upsilon}(s_1^t, b_1^t) -  V_{\mathit{lb}}^{\Gamma}(s_1^t, b_1^t) + 2K_{\mathit{ub}}(b_1, b_1^t) - \rho(t)} \\
        & \beta \rho(t + 1) + 2K_{\mathit{ub}}(b_1, b_1^t) - \rho(t)  & \mbox{by \eqref{eq:difference-bound-beta-rho}} \\
        & \leq \rho(t) - 2(U - L) \bar{\varepsilon} + 2K_{\mathit{ub}}(b_1, b_1^t) - \rho(t)  & \mbox{by definition of $\rho(t + 1)$} \\
        & \leq - 2(U - L) \bar{\varepsilon} + 2(U-L)\bar{\varepsilon} & \mbox{since $K_{\mathit{ub}}(b_1, b_1^t) \leq (U-L)\bar{\varepsilon}$}  \\
        & = 0  & \mbox{rearranging}
    \end{align*}
    which means that $(s_1^t, b_1) \notin \Psi_t$. Since $\mathbb{P}(S_E)$ is compact and thus totally bounded, we can conclude that $\Psi_t$ is finite. \qed
\end{proof}


\begin{proof}[\textbf{Proof of \thomref{thom:NS-HSVI}}]
    By the choice of $\bar{\varepsilon}$, the sequence $(\rho(t))_{t \in \mathbb{N}}$ is monotonically increasing and unbounded. Since $L \leq V_{\mathit{lb}}^{\Gamma}(s_B) \leq V_{\mathit{ub}}^{\Upsilon}(s_B) \leq U$ for all $s_B \in S_B$, the difference between $V_{\mathit{lb}}^{\Gamma}$ and $V_{\mathit{ub}}^{\Upsilon}$ is bounded by $U -L$. Therefore, there exists $T_{\max}$ such that $\rho(T_{\max}) \ge U - L \ge V_{\mathit{ub}}^{\Upsilon}(s_B) - V_{\mathit{lb}}^{\Gamma}(s_B)$ for all $s_B \in S_B$, and therefore the recursive procedure $\mathit{Explore}$ always terminates.

    To demonstrate that \algoref{alg:NS-HSVI} terminates, we reason about the sets $\Psi_t \subseteq S_B$ of belief points where the trials performed by the procedure $\mathit{Explore}$ terminated at exploration depth $t$. Initially, $\Psi_t = \emptyset$ for every $0 \leq t < T_{max}$. Whenever the $\mathit{Explore}$ recursion terminates at exploration depth $t$ (i.e., the condition on line 9 does not hold), the belief $s_B^t$ (which was the last belief considered during the trial) is added into the set $\Psi_t$, i.e., $\Psi_t \triangleq \Psi_t \cup \{ s_B^t \}$. 
    Since the agent state space $S_1$ is finite and the number of possible termination depth is finite ($0 \leq t < T_{max}$) and the set $\Psi_t$ is finite by \lemaref{lema:finite-terminal-beliefs}, the algorithm has to terminate. Then, combining Lemmas \ref{lema:new-pwc-alpha} and \ref{lema:upper-bound-update}, the conclusion follows directly. \qed
\end{proof}


\begin{lema}[LP for upper bound]\label{lema:pb-upper-bound}
    For particle-based belief $(s_1, b_1)$, let $P(s_E ; b_1) $ be the probability of particle $s_E$ under $b_1$. if the function $K_{\mathit{ub}}$ equals $K$, i.e.,
    \begin{equation}\label{eq:pb-k-ub}
   K_{\mathit{ub}}(b_1, b_1') = \mbox{$\frac{1}{2} $} (U - L)  \mbox{$\sum\nolimits_{ b_1(s_E) + b_1'(s_E) > 0 } $} |P(s_E ; b_1) - P(s_E ; b_1')|
\end{equation}
    then $V_{\mathit{ub}}^{\Upsilon}(s_1,b_1)$ is the optimal value of the LP: 
 \begin{align*}
     \lefteqn{\mbox{\rm minimise} \; \;  \mbox{$\sum_{k \in I_{s_{\scale{.75}{1}}}}$} \lambda_k y_k  + 1/2 (U - L) \mbox{$\sum\nolimits_{s_E \in S_E^{+}}$} c_{s_E} \; \; \mbox{\rm subject to}} \\
  & \quad c_{s_E} \ge | P(s_E; b_1) - \mbox{$\sum_{k \in I_{s_{\scale{.75}{1}}}}$} \lambda_k P(s_E; b_1^k) |,  \; \lambda_k \ge 0 \;\; \mbox{and} \; \; \mbox{$\sum_{k \in I_{s_{\scale{.75}{1}}}}$} \lambda_k = 1 
\end{align*}
for $s_E \in S_E^{+}$ and $k \in I_{s_{\scale{.75}{1}}}$, where $S_E^{+} = \{ s_E \in S_E \mid b_1(s_E) + \sum_{k \in I_{s_{\scale{.75}{1}}}} b_1^k(s_E) > 0 \}$. 
\end{lema}
\begin{proof}
   The result follows directly from \eqref{eq:new-ub} and \eqref{eq:pb-k-ub}. \qed
\end{proof}


\begin{theorem}[LP for minimax operator over upper bound, extended version of \thomref{thom:LP-upper-bound}]
    For the function $K_{\mathit{ub}}$, see \eqref{eq:pb-k-ub}, and particle-based belief $(s_1, b_1)$ represented by $\{ (s_E^i, \kappa_i) \}_{i=1}^{n_b}$, we have that $[TV_{\mathit{ub}}^{\Upsilon}](s_1,b_1)$ is the optimal value of the LP \eqref{eq:LP-minimax-upper}. 
\end{theorem}
\begin{proof}
    We first prove that given any $s_1 \in S_1$, $V_{\mathit{ub}}^{\Upsilon}(s_1, \cdot)$ is a convex function. Consider any two beliefs $b_1, b_1' \in \mathbb{P}(S_E)$ and $\tau, \tau' \ge 0$ such that $\tau + \tau'=1$. Let $(\lambda_k^{\star} )_{k \in I_{s_{\scale{.75}{1}}}}$ and $( \lambda_k^{\prime\star} )_{k \in I_{s_{\scale{.75}{1}}}}$ be optimal solutions of  \eqref{eq:new-ub} for $V_{\mathit{ub}}^{\Upsilon}(s_1, b_1)$ and $V_{\mathit{ub}}^{\Upsilon}(s_1, b_1')$ respectively, i.e., 
    \begin{align}
        V_{\mathit{ub}}^{\Upsilon}(s_1, b_1) & = \mbox{$\sum\nolimits_{k \in I_{s_{\scale{.75}{1}}}}$}\lambda_k^{\star} y_k + K_{\mathit{ub}}(b_1, \mbox{$\sum\nolimits_{k \in I_{s_{\scale{.75}{1}}}}$} \lambda_k^{\star} b_1^k) \nonumber \\
        V_{\mathit{ub}}^{\Upsilon}(s_1, b_1') & = \mbox{$\sum\nolimits_{k \in I_{s_{\scale{.75}{1}}}}$}\lambda_k^{\prime\star} y_k + K_{\mathit{ub}}(b_1', \mbox{$\sum\nolimits_{k \in I_{s_{\scale{.75}{1}}}}$} \lambda_k^{\prime\star} b_1^k) \,. \label{eq:two-upper-bounds}
    \end{align} 
    From the constraints of \eqref{eq:new-ub} it follows that:
    \begin{equation}\label{s4-eq}
    \mbox{$\tau \lambda_k^{\star} + \tau' \lambda_k^{\prime\star} \ge 0$ for all $k \in I_{s_1}$ and  $\mbox{$\sum\nolimits_{k \in I_{s_{\scale{.75}{1}}}}$} (\tau \lambda_k^{\star} + \tau' \lambda_k^{\prime\star}) = 1$.}
    \end{equation}
    Also let:
    \begin{align}
    S_E^1 & = \{ s_E \in S_E \mid b_1(s_E) + b_1'(s_E) + \mbox{$\sum_{k \in I_{s_{\scale{.75}{1}}}}$} b_1^k(s_E) > 0 \} \label{s1-eq} \\
    S_E^2 &= \{ s_E \in S_E \mid b_1(s_E) + \mbox{$\sum_{k \in I_{s_{\scale{.75}{1}}}}$} b_1^k(s_E) > 0 \} \label{s2-eq} \\
    S_E^3 &= \{ s_E \in S_E \mid b_1'(s_E) + \mbox{$\sum_{k \in I_{s_{\scale{.75}{1}}}}$} b_1^k(s_E) > 0 \} \, . \label{s3-eq}
    \end{align}
    Now using \eqref{eq:pb-k-ub} and \eqref{s1-eq} we have:
    \begin{align}
        & K_{\mathit{ub}}( \tau b_1 + \tau' b_1', \mbox{$\sum\nolimits_{k \in I_{s_{\scale{.75}{1}}}}$} (\tau \lambda_k^{\star} + \tau' \lambda_k^{\prime\star}) b_1^k) \nonumber \\
        & \lefteqn{\; = \mbox{$\frac{1}{2} $} (U - L)  \mbox{$\sum\nolimits_{ s_E \in S_E^1 }$} |\tau b_1(s_E) + \tau' b_1'(s_E) - \mbox{$\sum\nolimits_{k \in I_{s_{\scale{.75}{1}}}}$} (\tau \lambda_k^{\star} + \tau' \lambda_k^{\prime\star}) b_1^k(s_E) |} \nonumber \\
       & \leq  \mbox{$\frac{1}{2} $} (U - L) \mbox{$\sum\nolimits_{ s_E \in S_E^1 }$} \Big( \Big| \tau \big( b_1(s_E) -  \mbox{$\sum\nolimits_{k \in I_{s_{\scale{.75}{1}}}}$} \lambda_k^{\star} b_1^k(s_E) \big) \nonumber \\
        & \qquad + \tau' \big( b_1'(s_E) -  \mbox{$\sum\nolimits_{k \in I_{s_{\scale{.75}{1}}}}$} \lambda_k^{\prime\star} b_1^k(s_E) \big) \Big| \Big)  & \mbox{rearranging} \nonumber \\
        & =  \mbox{$\frac{1}{2} $} (U - L) \mbox{$\sum\nolimits_{ s_E \in S_E^1 }$} \Big( \tau |b_1(s_E) -  \mbox{$\sum\nolimits_{k \in I_{s_{\scale{.75}{1}}}}$} \lambda_k^{\star} b_1^k(s_E) | \nonumber \\
        & \qquad + \tau'  |b_1'(s_E) -  \mbox{$\sum\nolimits_{k \in I_{s_{\scale{.75}{1}}}}$} \lambda_k^{\prime\star} b_1^k(s_E) | \Big)  & \mbox{since $\tau,\tau' \geq 0$} \nonumber \\
        & = \mbox{$\frac{1}{2} $} (U - L)  \tau  \mbox{$\sum\nolimits_{ s_E \in S_E^2 }$} \big|b_1(s_E) -  \mbox{$\sum\nolimits_{k \in I_{s_{\scale{.75}{1}}}}$} \lambda_k^{\star} b_1^k(s_E) \big| \nonumber \\
        & \quad +  \mbox{$\frac{1}{2} $} (U - L)  \tau'   \mbox{$\sum\nolimits_{ s_E \in S_E^3 }$} \big |b_1'(s_E) -  \mbox{$\sum\nolimits_{k \in I_{s_{\scale{.75}{1}}}}$} \lambda_k^{\prime\star} b_1^k(s_E)  \big|  & \mbox{by \eqref{s2-eq} and \eqref{s3-eq}} \nonumber \\
        & = \tau K_{\mathit{ub}}( b_1, \mbox{$\sum\nolimits_{k \in I_{s_{\scale{.75}{1}}}}$}\lambda_k^{\star} b_1^k) + \tau' K_{\mathit{ub}}( b_1', \mbox{$\sum\nolimits_{k \in I_{s_{\scale{.75}{1}}}}$} \lambda_k^{\prime\star}b_1^k) \label{eq:weighted-inequality-KUB}
    \end{align}
    Next, from \eqref{eq:new-ub} we have:
    \begin{align*}
        \lefteqn{V_{\mathit{ub}}^{\Upsilon}(s_1, \tau b_1 + \tau' b_1') = \textup{ min}_{(\lambda_k)_{k \in I_{s_{\scale{.75}{1}}}}} \; \mbox{$\sum\nolimits_{k \in I_{s_{\scale{.75}{1}}}}$}\lambda_k y_k + K_{\mathit{ub}}( \tau b_1 + \tau' b_1', \mbox{$\sum\nolimits_{k \in I_{s_{\scale{.75}{1}}}}$} \lambda_k b_1^k)} \\
        & \leq \mbox{$\sum\nolimits_{k \in I_{s_{\scale{.75}{1}}}}$}(\tau \lambda_k^{\star} + \tau' \lambda_k^{\prime\star}) y_k + K_{\mathit{ub}}( \tau b_1 + \tau' b_1', \mbox{$\sum\nolimits_{k \in I_{s_{\scale{.75}{1}}}}$} (\tau \lambda_k^{\star} + \tau' \lambda_k^{\prime\star}) b_1^k) & \mbox{by \eqref{s4-eq}} \\
        & \leq \mbox{$\sum\nolimits_{k \in I_{s_{\scale{.75}{1}}}}$}(\tau \lambda_k^{\star} + \tau' \lambda_k^{\prime\star}) y_k + \tau K_{\mathit{ub}}( b_1, \mbox{$\sum\nolimits_{k \in I_{s_{\scale{.75}{1}}}}$}\lambda_k^{\star} b_1^k) \\
        & \qquad + \tau' K_{\mathit{ub}}( b_1', \mbox{$\sum\nolimits_{k \in I_{s_{\scale{.75}{1}}}}$} \lambda_k^{\prime\star}b_1^k) & \mbox{by \eqref{eq:weighted-inequality-KUB}} \\
        & = \tau V_{\mathit{ub}}^{\Upsilon}(s_1, b_1) + \tau' V_{\mathit{ub}}^{\Upsilon}(s_1, b_1')  & \mbox{by \eqref{eq:two-upper-bounds}}
    \end{align*}
    and hence $V_{\mathit{ub}}^{\Upsilon}(s_1, \cdot)$ is convex in $\mathbb{P}(S_E)$. 
    
    The inequality \eqref{eq:upper-bound-continuity} shows that $V_{\mathit{ub}}^{\Upsilon}(s_1, \cdot)$ is continuous in $\mathbb{P}(S_E)$. By following the proof of \cite[Proposition 4.12]{KH-BB-VK-CK:23}, we can prove that there exists a set $\Gamma'$ of functions $\mathbb{F}(S)$ such that $V_{\mathit{ub}}^{\Upsilon}(s_1, b_1) = \sup_{\alpha \in \Gamma'} \langle \alpha, (s_1,b _1) \rangle $ for all $(s_1, b_1) \in S_B$. Therefore, according to \thomref{thom:operator-equivalence}, for any $(s_1, b_1) \in S_B$:
    \begin{align}
         [T V_{\mathit{ub}}^{\Upsilon}](s_1, b_1) & =  \mbox{$\max_{u_1\in \mathbb{P}(A_1)}$} \mbox{$\min_{u_2\in \mathbb{P}(A_2 \mid S)} $} \mathbb{E}_{(s_1,b_1),u_1,u_2} [r(s,a)] \nonumber \\
        & \quad + \beta \mbox{$\sum_{a_1, s_1'}$} P(a_1, s_1' \mid (s_1, b_1), u_1, u_2 ) V_{\mathit{ub}}^{\Upsilon} (s'_1, b_1^{s_1,a_1, u_2, s_1'}) \nonumber \\
         & = \mbox{$\min_{u_2 \in \mathbb{P}(A_2 \mid S)} $}  \mbox{$ \max_{u_1\in \mathbb{P}(A_1)} $} \mathbb{E}_{(s_1,b_1),u_1,u_2} [r(s,a)] \nonumber \\
        &  \quad + \beta \mbox{$ \sum_{a_1, s_1'} $} P(a_1, s_1' \mid (s_1, b_1), u_1, u_2 ) V_{\mathit{ub}}^{\Upsilon}(s'_1, b_1^{s_1,a_1, u_2, s_1'}) \,. \label{eq:maxmin-minmax-equivalence-upper}
    \end{align}
     We now define a payoff function $J: \mathbb{P}(A_1) \times \mathbb{P}(A_2 \mid S) \to \mathbb{R}$ to be the objective of the maximin and minimax optimisation in \eqref{eq:maxmin-minmax-equivalence-upper} such that for $u_1 \in \mathbb{P}(A_1)$ and $u_2 \in \mathbb{P}(A_2 \mid S)$, letting $E_1 = \mathbb{E}_{(s_1,b_1),u_1,u_2} [r(s,a)]$, $p^{a_1} = u_1(a_1)$,  $p^{a_1, u_2, s_1'} = P(s_1' \mid (s_1, b_1), a_1, u_2 )$ 
    then we have:
    \begin{align}
        & \lefteqn{J(u_1, u_2) = E_1 + \beta \mbox{$\sum_{a_1, s_1'} $} p^{a_1} p^{a_1, u_2, s_1'}  V_{\mathit{ub}}^{\Upsilon}(s'_1, b_1^{s_1,a_1, u_2, s_1'})} \nonumber \\
        & =E_1 + \beta \mbox{$\sum_{a_1, s_1' \in A_1 \times S_1}$} p^{a_1} p^{a_1, u_2, s_1'} 
        \mbox{$\min_{(\lambda_k)_{k \in I_{s'_{\scale{.75}{1}}}}}$} \nonumber \\
        & \quad \left(\mbox{$\sum\nolimits_{k \in I_{s'_{\scale{.75}{1}}}}$}\lambda_k y_k +  K_\mathit{ub} \big( b_1^{s_1,a_1, u_2, s_1'}, \mbox{$\sum\nolimits_{k \in I_{s'_{\scale{.75}{1}}}}$} \lambda_k b_1^{s_1,a_1, u_2, s_1'} \big) \right) & \mbox{by \eqref{eq:new-ub}.} \nonumber 
    \end{align}
    Now combining this with \eqref{eq:pb-k-ub} we have:
       \begin{align}
        & \lefteqn{J(u_1, u_2) = E_1 + \beta \mbox{$\sum_{a_1, s_1'} $} p^{a_1} p^{a_1, u_2, s_1'}  V_{\mathit{ub}}^{\Upsilon}(s'_1, b_1^{s_1,a_1, u_2, s_1'})} \nonumber \\
        & =E_1 + \beta \mbox{$\sum_{a_1, s_1' \in A_1 \times S_1}$} p^{a_1} p^{a_1, u_2, s_1'} \mbox{$\min_{\overline{\nu}, \overline{d}}$}  \big(\mbox{$\sum\nolimits_{k \in I_{s'_{\scale{.75}{1}}}}$}\nu_k y_k +  \mbox{$\frac{1}{2} $} (U - L)  \mbox{$\sum\nolimits_{ s_E \in S_E^{+} }$} d_{s_E} \big) \nonumber 
    \end{align}
    where 
    \[ \overline{\nu} =(\nu_k^{a_1, s_1'})_{ (a_1, s_1') \in A_1 \times S_1, k \in I_{s_{\scale{.75}{1}}'}} \;\; \mbox{and} \; \; \overline{c} = ( d_{s_{\scale{.75}{E}}'}^{a_1, s_1'})_{(a_1, s_1') \in A_1 \times S_1 , s_{\scale{.75}{E}}' \in S_{\scale{.75}{E}}^{a_{\scale{.75}{1}}, s_{\scale{.75}{1}}'}}
    \]
    are real-valued vectors of variables subject to the following linear constraints:
    \begin{align}
        d_{s_E'}^{a_1, s_1'} \ge & |P(s_E'; b_1^{s_1,a_1, u_2, s_1'}) - \mbox{$\sum\nolimits_{k \in I_{s'_{\scale{.75}{1}}}}$}  \nu_k^{a_1, s_1'} P(s_E'; b_1^k) | \nonumber \\
        \nu_k^{a_1, s_1'} \ge & 0 \textup{ for } k \in I_{s'_{\scale{.75}{1}}} \nonumber \\
        \mbox{$\sum\nolimits_{k \in I_{s'_{\scale{.75}{1}}}}$} \nu_k^{a_1, s_1'} = & 1 \label{eq:c-lambda-a1-s1}
    \end{align}
and $S_E^{a_1, s_1'}  = \{ s_E' \in S_E \mid \mbox{$\sum_{a_2 \in A_2}$} b_1^{s_1,a_1, a_2, s_1'}(s_E') + \mbox{$\sum_{k \in I_{s'_{\scale{.75}{1}}}}$} b_1^k(s_E') > 0 \}$. Letting
    \[
    C^{a_1, s_1'}  = \mbox{$\frac{1}{2} $} (U - L)  \mbox{$\sum\nolimits_{ s_{\scale{.75}{E}}' \in S_{\scale{.75}{E}}^{a_{\scale{.75}{1}},s_{\scale{.75}{1}}'} }$} d_{s_{\scale{.75}{E}}'}^{a_1,s_1'} 
    \]
   it follows that $J(u_1, u_2)$ equals:
    \begin{align}
       & \mbox{$\min_{\overline{\nu}, \overline{c}}$} \big( E_1 + \beta \mbox{$\sum_{(a_1, s_1') \in A_1 \times S_1}$} p^{a_1} p^{a_1, u_2, s_1'} \big(\mbox{$\sum\nolimits_{k \in I_{s'_{\scale{.75}{1}}}}$}\nu_k^{a_1, s_1'} y_k +  C^{a_1, s_1'} \big) \big) \label{eq:J-function-2-upper}    
   \end{align}
    %
    Now, given any $u_2 \in \mathbb{P}(A_2 \mid S)$, let $\Lambda$ be the feasible set for $(\overline{\nu}, \overline{c})$, which is convex using \eqref{eq:c-lambda-a1-s1}. We then define a game with strategy spaces $\Lambda$ and $\mathbb{P}(A_1)$  and payoff function $J_{u_2}:\Lambda \times\mathbb{P}(A_1) \to \mathbb{R}$ which is the objective of \eqref{eq:J-function-2-upper}, i.e., for $(\overline{\nu}, \overline{c}) \in \Lambda$ and $u_1 \in \mathbb{P}(A_1)$:
    \begin{align} \label{eq:J-u2-upper}
    J_{u_2}((\overline{\nu}, \overline{c}), u_1) = E_1 + \beta \mbox{$\sum_{a_1, s_1'}$} p^{a_1} p^{a_1,  u_2, s_1'} \big(\mbox{$\sum\nolimits_{k \in I_{s'_{\scale{.75}{1}}}}$}\nu_k^{a_1, s_1'} y_k +  C^{a_1,s_1'} \big) \, .
    \end{align}
    Combining \eqref{eq:maxmin-minmax-equivalence-upper}, \eqref{eq:J-function-2-upper} and \eqref{eq:J-u2-upper} we have:
    \begin{align}
        [T V_{\mathit{ub}}^{\Upsilon}](s_1, b_1) &  = \mbox{$ \min_{u_2 \in \mathbb{P}(A_2 \mid S)} $} \mbox{$\max_{u_1 \in \mathbb{P}(A_1)}$} J(u_1, u_2) \nonumber \\
        & = \mbox{$\min_{u_2 \in \mathbb{P}(A_2 \mid S)}$}  \mbox{$ \max_{u_1 \in \mathbb{P}(A_1)}$} \mbox{$\min_{(\overline{\nu}, \overline{c}) \in \Lambda}$} J_{u_2}((\overline{\nu}, \overline{c}), u_1) \,. \label{eq:max-min-max2}
    \end{align}
    We next show that the von Neumann's Minimax Theorem \cite{JvN:28} applies to the game with payoff function $J_{u_2}$ and strategy spaces $\Lambda$ and $\mathbb{P}(A_1)$. This theorem requires that:
    \begin{itemize}
    \item
    $\Lambda$ and $\mathbb{P}(A_1)$ are compact convex sets;
    \item
    $J_{u_2}$ is a continuous function that is concave-convex, i.e., $J_{u_2}((\overline{\nu}, \overline{c}), \cdot)$ is concave for fixed $(\overline{\nu}, \overline{c})$ and $J_{u_2}(\cdot, u_1)$ is convex for fixed $u_1$. 
    \end{itemize}
    Clearly $\Lambda$ and $\mathbb{P}(A_1)$ are compact convex sets and by \eqref{eq:J-u2-upper}, $J_{u_2}$ is bilinear in $\overline{\nu}, \overline{c}$ and $u_1$, and thus concave in $\mathbb{P}(A_1)$ and convex in $\Lambda$. Hence we can apply von Neumann's Minimax Theorem, which gives us:
    \[
    \begin{array}{c}
    \max_{u_1 \in \mathbb{P}(A_1)} \min_{(\overline{\nu}, \overline{c}) \in \Lambda} J_{u_2}((\overline{\nu}, \overline{c}), u_1) = \min_{(\overline{\nu}, \overline{c}) \in \Lambda} \max_{u_1 \in \mathbb{P}(A_1)} J_{u_2}((\overline{\nu}, \overline{c}), u_1) \, .
    \end{array}
    \]
    Therefore, using this result and \eqref{eq:max-min-max2} we have that:
    \begin{align*}
    [T V_{\mathit{ub}}^{\Upsilon}](s_1, b_1) & = \mbox{$ \min_{u_2 \in \mathbb{P}(A_2 \mid S)} $} \mbox{$\min_{(\overline{\nu}, \overline{c}) \in \Lambda}$}  \mbox{$\max_{u_1 \in \mathbb{P}(A_1)}$} J_{u_2}((\overline{\nu}, \overline{c}), u_1) \\
    & = \mbox{$ \min_{u_2 \in \mathbb{P}(A_2 \mid S)} $} \mbox{$\min_{(\overline{\nu}, \overline{c}) \in \Lambda}$}  \mbox{$\max_{u_1 \in \mathbb{P}(A_1)}$} \big( E_1 + \\
    & \quad + \beta \mbox{$\sum_{a_1, s_1'}$} p^{a_1} p^{a_1, u_2, s_1'} \big(\mbox{$\sum\nolimits_{k \in I_{s'_{\scale{.75}{1}}}}$}\nu_k^{a_1, s_1'} y_k +  C^{a_1,s_1'} \big)  \big) & \mbox{by \eqref{eq:J-u2-upper}} \\
    & = \mbox{$ \min_{u_2 \in \mathbb{P}(A_2 \mid S)} $} \mbox{$\min_{(\overline{\nu}, \overline{c}) \in \Lambda}$}  \mbox{$\max_{a_1 \in A_1}$} \big( E_1 + \\
    & \quad + \beta \mbox{$\sum_{s_1' \in S_1}$} p^{a_1, u_2, s_1'} \big(\mbox{$\sum\nolimits_{k \in I_{s'_{\scale{.75}{1}}}}$}\nu_k^{a_1, s_1'} y_k +  C^{a_1,s_1'} \big)  \big)
    \end{align*}
    where the final equality follows from the fact that, for fixed $u_2$ and $\overline{\nu}$ and $\overline{c}$, the objective is linear in $u_1$,
    from which $[TV_{\mathit{ub}}^{\Upsilon}](s_1, b_1)$
    can be formulated as the LP problem given by minimise $v$ subject to:
    \begin{align}
        v & \; \ge E_1 + \beta \mbox{$\sum_{s_1' \in S_1}$} p^{a_1,  u_2, s_1'} \big(\mbox{$\sum\nolimits_{k \in I_{s'_{\scale{.75}{1}}}}$}\nu_k^{a_1, s_1'} y_k  +  C^{a_1,s_1'} \big)  \big) \; \; \mbox{for all $a_1 \in A_1$.} \label{eq:LP-backup-upper}
    \end{align}
    Letting $\lambda_k^{a_1, s_1'} = p^{a_1,  u_2, s_1'} \nu_k^{a_1, s_1'}$ and $c_{s_E'}^{a_1, s_1'} = p^{a_1,  u_2, s_1'} d_{s_E'}^{a_1, s_1'} $, we can reformulate \eqref{eq:LP-backup-upper} as   minimise $v$ subject to:
    \begin{align*}
        v & \; \ge  \mbox{$\sum_{i=1}^{n_b}\sum_{a_2}$}  \kappa_i p_{a_2}^{s_1, s_E^i} r((s_1, s_E^i),(a_1, a_2)) +  \beta \mbox{$\sum_{s_1' \in S_1}$} v_{a_1, s_1'} \quad   \\
        v_{a_1, s_1'} & \; = \mbox{$\sum_{k \in I_{s'_{\scale{.75}{1}}}}$} \lambda_k^{a_1, s_1'} y_k  + \mbox{$\frac{1}{2}$} (U - L) \mbox{$\sum\nolimits_{s_E' \in S_E^{a_1, s_1'} }$} \hat{c}_{s_E'}^{a_1, s_1'} \quad  .
    \end{align*}
    for all $a_1 \in A_1$ and $s_1' \in S_1$, where  $u_2(a_2 |s_1,  s_E^i) = p_{a_2}^{s_1, s_E^i}$.
     We next compute the constraints for $\lambda_k^{a_1, s_1'}$ and $\hat{c}_{s_E'}^{a_1, s_1'}$. According to the belief update \eqref{eq:belief-update-u2}:
    \begin{align*}
        & \lefteqn{ p^{a_1,  u_2, s_1'} b_1^{s_1,a_1, u_2, s_1'}(s_E') = P(s_1' \mid (s_1, b_1), a_1, u_2 ) \frac{P( s_1', s_E' \mid (s_1, b_1), a_1, u_2)}{P( s_1' \mid (s_1, b_1), a_1, u_2 )}} \\
        & = P( s_1', s_E' \mid (s_1, b_1), a_1, u_2)  & \mbox{rearranging} \\
        & = \mbox{$\sum_{i=1}^{n_b}\sum_{a_2}$}   \kappa_i p_{a_2}^{s_1, s_E^i} \delta((s_1, s_E^i), (a_1,a_2))(s_1', s_E') \,
    \end{align*}
    where the final equality follows from the definition of a particle-based belief.
    Since $\nu_k^{a_1, s_1'}$ and $d_{s_E'}^{a_1, s_1'}$ are subject to the linear constraints \eqref{eq:c-lambda-a1-s1}, it follows that:
    \begin{align}
        c_{s_E'}^{a_1, s_1'} &\;  \ge \Big| \mbox{$\sum_{i=1}^{n_b}\sum_{a_2}$}   \kappa_i p_{a_2}^{s_1, s_E^i} \delta((s_1, s_E^i), (a_1,a_2))(s_1', s_E') \nonumber \\ & \qquad - \mbox{$\sum\nolimits_{k \in I_{s'_{\scale{.75}{1}}}}$}  \lambda_k^{a_1, s_1'} P(s_E' ; b_1^k) \Big| \nonumber \\
       \mbox{$\sum\nolimits_{k \in I_{s'_{\scale{.75}{1}}}}$} \lambda_k^{a_1, s_1'}  &\; = \mbox{$\sum_{i=1}^{n_b}\sum_{a_2, s_E'}$}   \kappa_i p_{a_2}^{s_1, s_E^i} \delta((s_1, s_E^i), (a_1,a_2))(s_1', s_E') \nonumber  \\
        \lambda_k^{a_1, s_1'}  &\; \ge 0  \label{eq:c-hat-lambda-hat-a1-s1}
    \end{align}
    for all $(a_1,s_1') \in A_1 \times S_1$, $1 \leq i \leq n_b$ and $s_E' \in S_E$, $k \in I_{s'_{\scale{.75}{1}}}$.
    Thus, the optimization problem can be reformulated as the LP problem in \eqref{eq:LP-minimax-upper}. \qed
\end{proof}




\section{Further Case Study Details and Statistics}\label{sec:appendix-examples}

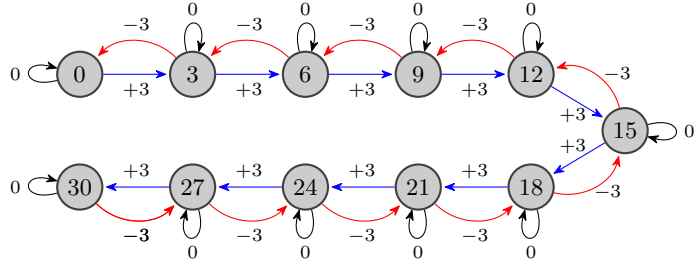
\begin{figure}[t]
    \centering
    \begin{tikzpicture}
   [node distance=1.5cm,on grid,>={Stealth[round]},bend angle=45,auto,
   every place/.style= {minimum size=6mm,thick,draw=black!75,fill=black!20},
   every transition/.style={thick,draw=black!75,fill=black!20},
   red place/.style= {place,draw=red!75,fill=red!20},
   every label/.style= {red}]

  \node[place]      (speed_0)                            {$0$};
  \node[place]      (speed_3)        [right=of speed_0]  {$3$};
  \node[place]      (speed_6)       [right=of speed_3]  {$6$};
  \node[place]      (speed_9)       [right=of speed_6]  {$9$};
  \node[place]      (speed_12)       [right=of speed_9]  {$12$};
  \node[place]      (speed_18)       [below =of speed_12]  {$18$};
  \node[place]      (speed_15)       at ($(speed_12)!0.5!(speed_18)$) [xshift=1.25cm] {$15$};
  \node[place]      (speed_21)       [left=of speed_18]  {$21$};
  \node[place]      (speed_24)       [left=of speed_21]  {$24$};
  \node[place]      (speed_27)       [left=of speed_24] {$27$};
  \node[place]      (speed_30)       [left=of speed_27]  {$30$};
  \draw [->,loop left, black]             (speed_0) to node[black] {\scriptsize$0$} (speed_0);
  \draw [->,post, blue]                 (speed_0) to node[below,black] {\scriptsize${+}3$} (speed_3);
  \draw [->,pre,bend left, red]         (speed_0) to node[above,black] {\scriptsize${-}3$} (speed_3);
  \draw [->,loop above, black]           (speed_3) to node[black,above] {\scriptsize$0$} (speed_3);
  \draw [->,post, blue]                 (speed_3) to node[below,black] {\scriptsize${+}3$} (speed_6);
  \draw [->,pre,bend left, red]         (speed_3) to node[above,black] {\scriptsize${-}3$} (speed_6);
  \draw [->,loop above]                 (speed_6) to node[black] {\scriptsize$0$} (speed_6.south);
  \draw [->,pre,bend left, red]         (speed_6) to node[above,black] {\scriptsize${-}3$} (speed_9);
  \draw [->,loop above, black]           (speed_9) to node[black,above] {\scriptsize$0$} (speed_9);
  \draw [->,post, blue]                 (speed_6) to node[below,black] {\scriptsize${+}3$} (speed_9);
  \draw [->,pre,bend left, red]         (speed_9) to node[above,black] {\scriptsize${-}3$} (speed_12);
  \draw [->,loop above, black]           (speed_12) to node[black,above] {\scriptsize$0$} (speed_12);
  \draw [->,post, blue]                 (speed_9) to node[below,black] {\scriptsize${+}3$} (speed_12);  
  \draw [->,pre,bend left, red]         (speed_12) to node[right,black] {\scriptsize${-}3$} (speed_15);
  \draw [->,post, blue,pos=0.4]                 (speed_12) to node[below,black] {\scriptsize${+}3$} (speed_15); 
  \draw [->,post, blue]                 (speed_15) to node[above,black,pos=0.6] {\scriptsize${+}3$} (speed_18);
  \draw [->,loop right]                 (speed_15) to node[black] {\scriptsize$0$} (speed_15);
  \draw [->,pre,bend left, red]         (speed_15) to node[below,black,pos=0.3] {\scriptsize${-}3$} (speed_18);
  \draw [->,post, blue]                 (speed_18) to node[above,black] {\scriptsize${+}3$} (speed_21);
  \draw [->,loop below]                 (speed_18) to node[black] {\scriptsize$0$} (speed_18);
  \draw [->,pre,bend left, red]         (speed_18) to node[below,black] {\scriptsize${-}3$} (speed_21);
  \draw [->,post, blue]                 (speed_21) to node[above,black] {\scriptsize${+}3$} (speed_24);
  \draw [->,loop below]                 (speed_21) to node[black] {\scriptsize$0$} (speed_21);
  \draw [->,pre,bend left, red]         (speed_21) to node[below,black] {\scriptsize${-}3$} (speed_24);
  \draw [->,post, blue]                 (speed_24) to node[above,black] {\scriptsize${+}3$} (speed_27);
  \draw [->,loop below]                 (speed_24) to node[black] {\scriptsize$0$} (speed_24);
  \draw [->,pre,bend left, red]         (speed_24) to node[below,black] {\scriptsize${-}3$} (speed_27);
  \draw [->,post, blue]                 (speed_27) to node[above,black] {\scriptsize${+}3$} (speed_30);
  \draw [->,loop below]                 (speed_27) to node[black] {\scriptsize$0$} (speed_27);
  \draw [->,pre,bend left, red]         (speed_27) to node[below,black] {\scriptsize${-}3$} (speed_30);
  \draw [->,pre,bend left, red]         (speed_27) to node[below,black] {\scriptsize${-}3$} (speed_30);
  \draw [->,loop left, black]            (speed_30) to node[black] {\scriptsize$0$} (speed_30);
\end{tikzpicture}
    \caption{Pedestrian-vehicle interaction: local transition diagram over the local states, i.e., vehicle speeds ($\textup{m/s}$), with actions corresponding to the possible accelerations of the vehicle, i.e., $+3$, $0$ and $-3$ ($\textup{m/s}^2$).}\label{fig:local_transition}
    \vspace*{0.4cm}
\end{figure}

\noindent
Finally, we give some additional details and statistics for the models
developed for the two case studies used for evaluation in \sectref{sec:experiments}.

\startpara{Pedestrian-vehicle interaction}
The one-sided NS-POSG for the pedestrian-vehicle scenario is defined as follows:
    \begin{itemize}
        \item $S_1 = \Loc_1 {\times} \Per_1$, where:
        \begin{align*}
        \Loc_1 & = \; \{30, 27, 24, 21, 18, 15, 12, 9, 6, 3, 0 \} \\
        \Per_1 & = \; \{``\emph{unlikely to cross}", ``\emph{likely to cross}",  ``\emph{very likely to cross}" \}
        \end{align*}
        are the vehicle's discrete speeds (km/h) and perceived pedestrian intentions, respectively.

        \item $S_E = \{ ((x_1, y_1),(x_2, y_2)) \in (\mathbb{R}^2)^2 \mid (0 \leq x_1, x_2 \leq 20) \wedge (0 \leq y_1, y_2 \leq 10) \}$, where $(x_1, y_1)$ and $(x_2, y_2)$ are the top-left coordinates of the 2D fixed-size bounding boxes of size $0.5{\times} 1.5 \; (\textup{m}^2)$
        around the pedestrian at the previous and current steps, respectively.
        

        \item $A = A_1 {\times} A_2$, where $A_1 = \{ -3, 0, 3 \}$ ($\textup{m/s}^2$) are the possible accelerations of the vehicle, and $A_2 = \{\mathit{cross}, \mathit{back}\} $ are the possible directions the pedestrian can choose to move. 

        \item For each local state $v_1 \in \Loc_1$ and environment state $((x_1, y_1),(x_2, y_2)) \in S_E$, we let $\obs_1(v_1, ((x_1, y_1),(x_2, y_2)))= f^{\max}_\mathit{ped} ((x_1, y_1),(x_2, y_2))$, where $f_\mathit{ped} : S_E \to \mathbb{P}(\Per_1)$ is a data-driven pedestrian intention classifier implemented via a feed-forward NN with ReLU activation functions and trained over the PIE dataset in \cite{AR-IK-TK-JKT:19}. Note here $\obs_1$ is independent of the local state of $\agent_1$.
        \item For $(v_1, \per_1) \in \Loc_1 {\times} \Per_1$, $v_1' \in \Loc_1$ and $(a_1, a_2) \in A$ we have:
         \[
         \delta_1((v_1, \per_1),(a_1, a_2)) (v_1') = \left\{ \begin{array}{cl}
         1 & \mbox{if $v_1' = g_{\mathit{next}}(v_1, a_1)$} \\
         0 & \mbox{otherwise}
         \end{array}  \right.
         \]
         where $g_{\mathit{next}} : \Loc_1 \times A_1 \to \Loc_1$ is the speed update function of the vehicle with the transition diagram in Fig.~\ref{fig:local_transition}.
         

         \item For $v_1 \in \Loc_1$, $((x_1, y_1), (x_2, y_2)), ((x_1', y_1'),(x_2', y_2')) \in S_E$ and $(a_1, a_2) \in A$ we have
         $\delta_E(v_1,((x_1, y_1),(x_2, y_2)),(a_1, a_2)) ((x_1', y_1'), (x_2', y_2')) = 1$ where
         \begin{align*}
             x_1' & = x_2, & y_1' & = y_2, \\
             x_2' & = x_2 + \mathit{move}(a_2)v_2 \Delta t,&  y_2' & = y_2 - v_1 \Delta t - \frac{a_1}{2} {\Delta t}^2,
         \end{align*}
         $v_2=4.5$ ($\textup{m/s}$) is the speed of the pedestrian, $\mathit{move}(a_2)$ is the direction of the movement of the pedestrian action, i.e., $\mathit{move}(\mathit{cross}) {=} -1$ and $\mathit{move}(\mathit{back}) {=} 1$, and $\Delta t = |g_{\mathit{next}}(v_1, a_1)-v_1|/|a_1| $ if $a_1 \neq 0$ and $0.3$ ($\textup{s}$) otherwise.
    \end{itemize}

\noindent
     A crash occurs if the environment state is in the set:
      \[
     \mathcal{R}_{\mathit{crash}} = \{ ((x_1, y_1), (x_2, y_2)) \in S_E \mid (0 \leq x_2 \leq 0.5) \wedge (0 \leq y_2 \leq 2.5) \}
     \]
     i.e., the current bounding box around the pedestrian has a distance of no more than $0.5$ and $1.0$ ($\textup{m}$) along the $x$ and $y$ coordinates
     to the vehicle, respectively (recall the bounding box has size $0.5{\times} 1.5 \; (\textup{m}^2$)).
     The reward structure is such that, for any $(s_1,s_E) \in S$ and $a \in A$, $r((s_1, s_E), a) = - 200$ if $s_E \in \mathcal{R}_{\mathit{crash}}$ and $0$ otherwise. To model a scenario where the pedestrian and the vehicle require different update frequencies, we can make the following adjustments: $1)$ select the time step of the model as that corresponding to the larger frequency; $2)$ augment the state with the current time step and the last taken action of the agent with smaller frequency; $3)$ build the environment transition function such that the last taken action of the agent with smaller frequency is used for the update within the lower frequency, which is tracked by the current time step.

\startpara{Pursuit-evasion game}
We modify the example presented in \cite{KH-BB-VK-CK:23} by
considering a continuous environment $\mathcal{R} \triangleq \{ (x, y) \in \mathbb{R}^2 \mid 0 \leq x, y \leq 3\}$ that is partitioned into multiple cells by perception functions. In this game, we have a pursuer agent $\agent_p$
that tries to catch an evader agent $\agent_e$. In each step, the evader moves by picking from the set of actions $A_e \triangleq \{\textit{up}, \textit{down}, \textit{left}, \textit{right}\}$. The pursuer moves in a similar manner with additional diagonal movements, and thus has the action set $A_p \triangleq \{\textit{up}, \textit{down}, \textit{left}, \textit{right}, \textit{upleft}, \textit{upright}, \textit{downleft}, \textit{downright}\}$.

The evader has full observation and knows the exact location of both players. The pursuer has partial observation, that is, it knows which cell it is in, but does not know its exact location and does not know which cell the evader is in. 
The perception function of the pursuer employs an NN classifier $f_\mathcal{R} : \mathcal{R} \rightarrow \mathbb{P}(\mathit{Grid})$, where $\mathit{Grid} \triangleq \{(i,j)\mid  1 \leq i,j \leq 3 \}$, which takes the location (coordinates) of the pursuer as input and outputs a probability distribution over nine abstract grid points (cells), thus partitioning the environment as illustrated by \figref{fig:pursuit_evasion_fcps}. This is modelled as a one-sided NS-POSG as follows.

\begin{figure}[t]
    \centering    
    \includegraphics[width=0.3\textwidth]{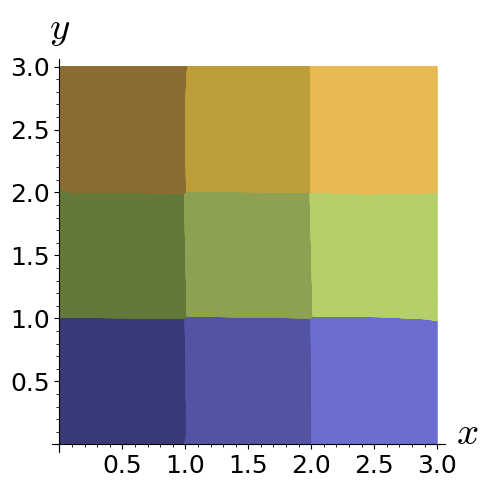}
    \hspace{1.50cm}
    \includegraphics[width=0.3\textwidth]{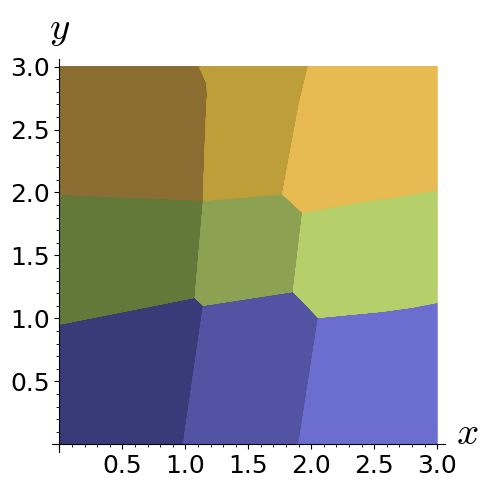}
    \caption{Representation of a regular (left) and coarse (right) perception function for the pursuit-evasion example. Each graph depicts the boundaries of the nine abstract grid cells learnt by the classifiers. The pre-images of the regular and coarse perception functions are composed of 48 and 50 polytopes, respectively.}
    \label{fig:pursuit_evasion_fcps}
\end{figure}

\begin{itemize}
    \item $S_1 = \Loc_1 {\times} Per_1$, where the local state $Loc_1 = \{\bot\}$ is a dummy state and $\Per_1 = \mathit{Grid}$.
    \item $S_E =  \{ ((x_{p}, y_{p}), (x_{e}, y_{e})) \in \mathcal{R}^2 \}$.
    \item $A = A_1 {\times} A_2$, where $A_1 = A_p$ and $A_2 = A_e$.
    \item $\obs_i(\bot ,(x_p,y_p)) = f_\mathcal{R}^{\max}(x_p, y_p)$ and the NN classifier $f_{\mathcal{R}}: \mathcal{R} \rightarrow \mathbb{P}(\mathit{Grid})$ is implemented via a feed-forward NN with one hidden ReLU layer and 14 neurons.
    \item For environment states $s_E = ((x_{p}, y_{p}), (x_{e}, y_{e}))$, $s_E' = ((x_{p}', y_{p}'), (x_{e}', y_{e}'))$, local state $\loc_1$ and joint action $a=(a_1,a_2)$:
    \[
    \delta_E(\bot, s_E, a ) (s_E') = \mbox{$\prod_{i \in \{ p, e \} }$}\delta_{E_i}((x_i, y_i), a_i )(x_i', y_i')
    \]
    where for $i \in \{ p, e \}$: 
     \[
     \delta_{E_i}((x_i, y_i), a_i)(x'_i, y'_i) = \\
    \left\{ \begin{array}{cl}
    1 & \mbox{if $x_i' = x_i +  d_{a_i}^x$ and $y'_i = y_i + d_{a_i}^y$} \\
    0 & \mbox{otherwise}
    \end{array}     \right.  
    \]
    and the pairs $(d_{a_i}^x, d_{a_i}^y)$ indicate the direction of movement under action $a_i$, e.g., $(d_{\mathit{up}}^x,d_{\mathit{up}}^y)=(0,1)$ and $(d_{\mathit{left}}^x,d_{\mathit{left}}^y)=(-1,0)$.
\end{itemize}
The capture condition in \cite{KH-BB-VK-CK:23} is also used, that is, the evader is captured if it is in the same cell as the pursuer, so the set of capture states $\mathcal{R}_\mathit{capture}$ is given by:
\[
\{ ((x_{p}, y_{p}), (x_{e}, y_{e})) \in S_E \mid \, \exists (i, j) \in \mathit{Grid}, \, (i{-}1 {\leq} x_{p}, x_e {<} i) \wedge (j{-}1 {\leq} y_{p}, y_e {<} j ) \} \,.  
\]
Unlike \cite{KH-BB-VK-CK:23}, the game does not end once the evader is captured, yielding the possibility of multiple captures. In case the pursuer is successful, that is, it enters the same cell as the evader, it receives a reward of 100. This can be modelled by assigning that value to any state-action pair with the state in $\mathcal{R}_\mathit{capture}$.

A model where the pursuer agent $\agent_p$ has two pursuers under its control was also developed. For that model, the actions available to the pursuer agent are pairs corresponding to chosen directions of the pursuers, where they can now only move horizontally or vertically, 
i.e., the actions available to the pursuer agent are given by $A_{p} \triangleq (\{\textit{up}, \textit{down}, \textit{left}, \textit{right}\})^2$. Additionally, the perception function of the pursuer agent is modified to take the coordinates of both pursuers and output the perceived grid cell for each of them. In this scenario, a capture happens if the evader is in the same grid cell as either pursuer.

\startpara{Impact of perception} In \figref{fig:pursuit_evasion_strat_both}(b), when the pursuer starts from cell E (marked by the red dot), its perception function incorrectly observes that it is in cell G (denoted by the blue border). In the next step, when it is in fact in cell B, the perception function shows it to be in cell A. As the pursuer's actions are tied to where it perceives itself to be
(e.g., in cell A it cannot use left or down actions to avoid moving out of bounds), 
the number of explored regions is reduced. Hence, the evader can safely hide in different portions of the environment. In particular, we see that it moves from its initial position to cells A or I and stays there, as those cells would not be visited by the pursuer. The pursuer does not actually capture the evader even if its perceived location cell almost completely coincides with a green area that the evader is highly likely to stay in.

\begin{table}[t]
\scriptsize
\setlength{\tabcolsep}{3.5pt}  
\centering
{
\begin{tabular}{|c||c|c||c|c|c|c|c|c|c|c|} \hline
\multirow{2}{*}{Model} & Initial & \multirow{2}{*}{$\beta$} & \multirow{2}{*}{$|\Gamma|$} & \multicolumn{2}{|c|}{Lower bound} & \multirow{2}{*}{$|\Upsilon|$} & \multicolumn{2}{|c|}{Upper bound} &
\multirow{2}{*}{\shortstack[c]{Iter.}} & Time \\ \cline{5-6} \cline{8-9}
& pts. & & & init. & final & & init. & final & & (min) \\ \hline \hline
\multirow{6}{*}{\shortstack[c]{Pursuit-evasion \\ (one pursuer)}} 
& 1 & 0.7 & 184 & 0 & 5.0653 & 265 & 333.33 & 9.1819 & 169 & 15 \\ \cline{2-11}
& 1 & 0.7 & 515 & 0 & 5.2798 & 788 & 333.33 & 6.6317 & 264 & 120 \\ \cline{2-11}
& 2 & 0.7 & 413 & 0 & 4.5299 & 998  & 333.33 & 11.570 & 299 & 120 \\ \cline{2-11}
& 1 & 0.8 & 468 & 0 & 9.8827 & 731 & 500.00 & 16.289 & 170 & 120 \\ \cline{2-11}
& 1 & 0.9 & 331 & 0 & 22.387 & 731 & 1000.0 & 58.906 & 130 & 120 \\  \cline{2-11}
& {\bf 1} & {\bf 0.99} & {\bf 5} & {\bf 0} & {\bf 34.973} & {\bf 128} & {\bf 10000} & {\bf 35.972} & {\bf 44} & {\bf 3} \\
\hline\hline
\multirow{2}{*}{\shortstack[c]{Pursuit-evasion \\ (two pursuers)}} 
 & \multirow{2}{*}{1} & \multirow{2}{*}{0.7} & \multirow{2}{*}{509} & \multirow{2}{*}{0} & \multirow{2}{*}{14.134} & \multirow{2}{*}{790} & \multirow{2}{*}{333.33} & \multirow{2}{*}{39.943} & \multirow{2}{*}{274} & \multirow{2}{*}{120} \\
 & & & & & & & & & & \\
\hline\hline
\multirow{3}{*}{\shortstack[c]{Pedestrian-vehicle}} 
& 1 & 0.7 & 1,928 & 0 & 620.54 & 4936 & 666.67 & 666.67 & 297 & 120 \\ \cline{2-11}
& 2 & 0.7 & 2,783 & 0 & 526.34 & 8532 & 666.67 & 666.67 & 363 & 120 \\ \cline{2-11}
& 1 & 0.8 & 2,089 & 0 & 805.92 & 5708 & 1000.0 & 1000.0 & 330 & 120 \\ 
\hline
\end{tabular}}
\vspace*{0.2cm} 
\caption{Statistics for a set of one-sided NS-POSG solution instances. The bold entries for the pursuit-evasion model correspond to that with the coarser perception function (see Fig.~\ref{fig:pursuit_evasion_fcps}).}
\vspace*{-0.2cm} 
\label{tab:stats}
\end{table}

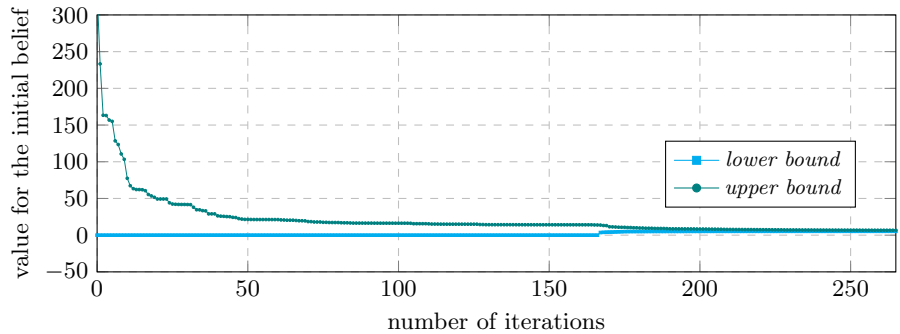
\begin{figure}[t]
    \centering
\small{
\begin{tikzpicture}
\begin{axis}[
    title style={yshift=-2ex},
    ylabel={value for the initial belief},
    xlabel={number of iterations},
    xmin=0, xmax=265,
    xtick={0, 50, 100, 150, 200, 250},
    ymin=-50, ymax=300,
    ytick={-50, 0, 50, 100, 150, 200, 250, 300},
    ymajorgrids=true,
    grid style=dashed,
    grid=both,
    height=5.0cm,
    width=1.0\textwidth,
    legend entries={
                {\emph{lower bound}},
                {\emph{upper bound}}
                },
    legend style={at={(0.95,0.25)},
                anchor=south east, 
                nodes={scale=0.9, transform shape}}            
]
\addlegendimage{mark=square*,cyan,mark size=1.5pt}
\addlegendimage{mark=*,teal,mark size=1.5pt}
]

\addplot[mark=square*,cyan,opacity=1.0,mark size=0.5pt, very thin] table [x=k, y=lb, col sep=comma]{figures/pursuit_evasion/3x3_0_7_single_pursuer.txt};
\addplot[mark=*,teal,opacity=1.0,mark size=0.5pt, very thin] table [x=k, y=ub, col sep=comma]{figures/pursuit_evasion/3x3_0_7_single_pursuer.txt};

\end{axis}
\end{tikzpicture}
}
    \caption{Lower and upper bound values for a pursuit-evasion game with one pursuer when $\beta=0.7$.}
    \label{fig:pursuit_evasion_lb_ub}
\end{figure} 

\startpara{Statistics}
\tabref{tab:stats} shows statistics for solving various instances,
varying the 
number of points in the initial belief and discount factor $\beta$.
The table presents
the initial and final values of the upper and lower bounds,
the number of $\alpha$-functions generated for the lower bound computation ($|\Gamma|$),
the number of belief points for the upper bound computation ($|\Upsilon|$), and the number of iterations and the time required. In the experiments, we have set a timeout of 120 minutes (except for the first row where the timeout was reduced to 15 minutes). In addition, \figref{fig:pursuit_evasion_lb_ub} shows how the lower and upper bound values change for the initial belief as the number of iterations increases for one instance of the pursuit-evasion game.

Since our algorithm is \emph{anytime}, lower and upper bounds hold throughout computations and we successfully generate meaningful strategies
(discussed further below) on a range of models.
However, computation is generally slow due to the number of LP problems to solve 
(whose size increases with $|\Gamma|$ and $|\Upsilon|$), as well as expensive operations over polyhedra and the probabilistic branching of mixed strategies to guide exploration. 

Both \tabref{tab:stats} and \figref{fig:pursuit_evasion_lb_ub} illustrate the impact of these factors. In the first two rows of \tabref{tab:stats} we observe the difference between a 15 and 120 minute timeout for the same instance of pursuit-evasion game with a single pursuer. As can be seen, the increase in the timeout causes the lower bound to improve (increase) by 0.2145, while the upper bound improves (decreases) by 2.55. With a timeout of 15 minutes we see that 169 iterations are performed; however, due to the number of $\alpha$-functions growing from 184 to 515, increasing the timeout to 120 minutes only allows 95 more iterations to be performed.

Considering \figref{fig:pursuit_evasion_lb_ub}, we initially see a sharp decrease of the upper bound, but improvement to either bound becomes progressively harder as computation progresses. The entry for the pursuit-evasion game with a single pursuer with a coarser perception function in \tabref{tab:stats} (highlighted in bold) is the only instance that converges before the timeout due to the fact that the number of reachable regions is smaller.

\tabref{tab:stats} also shows that, as expected, larger discount factors lead to larger lower and upper bounds. For the pursuit-evasion model with two pursuers, the larger difference between the bound values reached at the timeout is a consequence of a higher branching factor during exploration slowing down the computation, as we have 64 joint actions in each state. For the entries related to the pedestrian-vehicle interaction model in the table, the final upper bound values match their initial values due to the fact that the initial beliefs were selected so that it should be possible to avoid a crash if an optimal strategy was played.

We note that HSVI for \emph{finite} one-sided POSGs, in~\cite{KH-BB-VK-CK:23},
is already computationally very expensive, even with  multiple optimisations
(\cite{KH-BB-VK-CK:23} uses a timeout of 10 hours, versus 2 hours here).
}

\end{document}